\documentclass[a4paper,11pt,reqno]{amsart}

\usepackage[english]{babel}
\usepackage[OT2,T2A,T1]{fontenc}
\usepackage[utf8]{inputenc}		% Euler for math | Palatino for rm | Helvetica for ss | Courier for tt
\usepackage{eulervm}			% a better implementation of the euler package (not in gwTeX)
\usepackage[scr=dutchcal,bb=boondox]{mathalfa}
\usepackage{dsfont}
\usepackage{titlesec}
\titleformat{\section}
	{\centering\normalfont\scshape\bfseries}{\thesection.}{1em}{} % format section
\titleformat{\subsection}
 	{\centering\normalfont\bfseries}{\thesubsection.}{1em}{} % format subsection
\titleformat*{\paragraph}{\bfseries}

\usepackage{pdfsync}
\usepackage{tabularx,booktabs}			% tables
\usepackage{caption,subcaption}			% captions
\usepackage[dvipsnames]{xcolor}			% colori
\usepackage{bookmark}					% segnalibri
\usepackage{textcomp}
\usepackage{enumitem}
\usepackage{leftidx,tensor,accents}
\usepackage[margin=1in]{geometry}

\setitemize{leftmargin=2.5em,itemsep=.75em}

\usepackage{graphicx}		% images
\usepackage{float}
\usepackage{tikz,tikz-cd}	% TiKz images
\usetikzlibrary{decorations.pathreplacing,calligraphy,arrows,calc}

\definecolor{mygray}{gray}{0.9}

\usepackage[textsize=footnotesize,textwidth=0.85in]{todonotes}
\presetkeys%
	{todonotes}%
	{color=Dandelion}{}%

\usepackage{amsmath,amssymb,amsthm,mathtools}
\usepackage{bm,braket,faktor,stmaryrd}
\SetSymbolFont{stmry}{bold}{U}{stmry}{m}{n}
\numberwithin{equation}{section}

\usepackage[style=alphabetic,maxalphanames=5,giveninits,sorting=nyt,maxbibnames=99,backend=biber]{biblatex}
\renewbibmacro{in:}{} %Renove "In: "
\usepackage{csquotes}
\usepackage{hyperref}					% collegamenti ipertestuali
\definecolor{webbrown}{rgb}{0.65, 0.16, 0.16}
\hypersetup{
colorlinks=true, linktocpage=true, pdfstartpage=1, pdfstartview=FitV,breaklinks=true, pdfpagemode=UseNone, pageanchor=true, pdfpagemode=UseOutlines,plainpages=false, bookmarksnumbered, bookmarksopen=true, bookmarksopenlevel=1,hypertexnames=true, pdfhighlight=/O,urlcolor=webbrown, linkcolor=RoyalBlue, citecolor=ForestGreen}

\newcommand{\bbraket}[1]{\llbracket #1 \rrbracket}

\newcommand*\ph{{\mkern 2mu\cdot\mkern 2mu}} % place holder
\newcommand{\dd}{\mathrm{d}}
\newcommand{\bigO}{\mathrm{O}}

\newcommand{\NN}{\mathbb{Z}_{\geq 0}}
\newcommand{\ZZ}{\mathbb{Z}}
\newcommand{\QQ}{\mathbb{Q}}

\newcommand{\CC}{\mathbb{C}}

\newcommand{\Epsilon}{{\rm E}}
\newcommand{\Up}{{\rm U}}

\newcommand{\Mbar}{\overline{\mathcal{M}}}

\newcommand{\id}{{\rm id}}
\newcommand{\tr}{{\rm tr}}

\newcommand{\ev}{{\rm ev}}
\newcommand{\gl}{\varphi}
\newcommand{\fg}{\pi}
\newcommand{\Hom}{{\rm Hom}}
\newcommand{\End}{{\rm End}}

\newcommand{\GL}{{\rm GL}}
\newcommand{\Sym}[2]{#2^{\odot #1}}
\newcommand{\Sy}{\mathfrak{S}}

\DeclareFontFamily{U}{matha}{\hyphenchar\font45}
\DeclareFontShape{U}{matha}{m}{n}{
	<5> <6> <7> <8> <9> <10> gen * matha
	<10.95> matha10 <12> <14.4> <17.28> <20.74> <24.88> matha12
	}{}
\DeclareSymbolFont{matha}{U}{matha}{m}{n}
\DeclareFontFamily{U}{mathb}{}
\DeclareMathSymbol{\connsymb}{2}{matha}{"0C}
\DeclareFontShape{U}{mathb}{m}{n}{
		<-5.5> mathb5 <5.5-6.5> mathb6 
		<6.5-7.5> mathb7 <7.5-8.5> mathb8 <8.5-9.5> mathb9 <9.5-11> mathb10 
		<11-> mathb12
	}{}
\DeclareSymbolFont{mathb}{U}{mathb}{m}{n}
\DeclareFontSubstitution{U}{mathb}{m}{n}
\DeclareMathSymbol{\discsymb}{2}{mathb}{"0C}
\DeclareSymbolFont{cyrletters}{OT2}{wncyr}{m}{n}
\DeclareMathSymbol{\Sha}{\mathalpha}{cyrletters}{"58}
\newcommand{\conn}[1]{#1^{\connsymb}}
\newcommand{\disc}[1]{#1^{\discsymb}}
\newcommand{\cd}[1]{#1^{\star}}

\def\loc{\textup{loc}}

\newcommand{\ind}[3]{#1\indices*{_{#3}^{#2}}}
\newcommand{\indF}[4]{#1\indices*{_{#2; #4}^{#3}}}
\newcommand{\indC}[3]{#1\vphantom{C}\indices*{_{#3}^{#2}}}
\newcommand{\cdxi}[1]{\xi^{\star,#1}}
\newcommand{\preind}[2]{\prescript{#1}{}{#2}}

\newcommand{\bcdot}{\mathbin{\bm{\cdot}}}
\newcommand{\bcdots}{\mathbin{\bm{\cdots}}}

\DeclareMathOperator*{\Res}{Res}
\def\Aut{{\rm Aut}}
\newcommand{\Ker}{{\rm Ker}}
\renewcommand{\Im}{{\rm Im}}

\renewcommand{\Vert}{\mathsf{V}}
\newcommand{\Edge}{\mathsf{E}}
\newcommand{\Half}{\mathsf{H}}
\newcommand{\Leaf}{\mathsf{L}}

\newcommand{\dfsymb}[2]{\Bigl\{ \! \genfrac..{0pt}{1}{#1}{#2} \! \Bigr\}}

\theoremstyle{plain}
\newtheorem{thm}{Theorem}[section]

\newtheorem{prop}[thm]{Proposition}

\newtheorem{lem}[thm]{Lemma}

\newtheorem*{question*}{Question}
\theoremstyle{definition}
\newtheorem{defn}[thm]{Definition}
\newtheorem{rem}[thm]{Remark}
\newtheorem{ex}[thm]{Example}

\begin{document}
\renewcommand{\hbar}{\hslash}

\title[Symmetries of F-CohFTs and F-TR]{Symmetries of F-cohomological field theories and F-topological recursion}
\author[G.~Borot]{Gaëtan Borot}
\address[G.~Borot]{
	Institut für Mathematik und Institut für Physik, Humboldt-Universität zu Berlin, Unter den Linden 6, 10099 Berlin, Germany.
}
\email{gaetan.borot@hu-berlin.de}
\author[A.~Giacchetto]{Alessandro Giacchetto}
\address[A.~Giacchetto]{ Departement Mathematik, ETH Zürich, Rämisstrasse 101, Zürich 8044, Switzerland}
\email{alessandro.giacchetto@math.ethz.ch}
\author[G.~Umer]{Giacomo Umer}
\address[G.~Umer]{
	Institut für Mathematik, Humboldt-Universität zu Berlin, Unter den Linden 6, 10099 Berlin, Germany.
}
\email{umergiac@hu-berlin.de}
\date{}


% 14Hxx: Curves
	% 14H10: Families, moduli (algebraic)
	% 14H60: Vector bundles on curves and their moduli
	% 14H70: Relationships with integrable systems
	% 14H81: Relationships with physics

% 14Nxx: Projective and enumerative geometry
	% 14N35: Gromov--Witten invariants, quantum cohomology
	% 14N10: Enumerative problems (combinatorial problems)

% 37Kxx: Infinite-dimensional Hamiltonian systems
	% 37K10: Completely integrable systems, integrability tests, bi-Hamiltonian structures, hierarchies (KdV, KP, Toda, etc.)
	% 37K20: Relations with algebraic geometry, complex analysis, special functions
	% 37K30: Relations with infinite-dimensional Lie algebras and other algebraic structures

%05Axx: Enumerative combinatorics
	%05A15: Exact enumeration problems, generating functions (enumerative combinatorics)

% 05Exx: Algebraic combinatorics
	% 05E05: Symmetric functions
	% 05E10: Tableaux, representations of the symmetric group

% 81Rxx: Groups and algebras in quantum theory
	% 81R10: Infinite-dimensional groups and algebras motivated by physics, including Virasoro, Kac-Moody, $\mc{W}$-algebras and other current algebras and their representations
	% 81R12: Relations with integrable systems

% 20Cxx: Representation theory of groups
	% 20C35: Applications of group representations to physics

\begin{abstract}
	We define F-topological recursion (F-TR) as a non-symmetric version of topological recursion, which associates a vector potential to some initial data. We describe the symmetries of the initial data for F-TR and show that, at the level of the vector potential, they include the F-Givental (non-linear) symmetries studied by Arsie, Buryak, Lorenzoni, and Rossi within the framework of F-manifolds. Additionally, we propose a spectral curve formulation of F-topological recursion. This allows us to extend the correspondence between semisimple cohomological field theories (CohFTs) and topological recursion, as established by Dunin-Barkowski, Orantin, Shadrin, and Spitz, to the F-world. In the absence of a full reconstruction theorem à la Teleman for F-CohFTs, this demonstrates that F-TR holds for the ancestor vector potential of a given F-CohFT if and only if it holds for some F-CohFT in its F-Givental orbit. We turn this into a useful statement by showing that the correlation functions of F-topological field theories (F-CohFTs of cohomological degree 0) are governed by F-TR. We apply these results to the extended 2-spin F-CohFT. Furthermore, we exhibit a large set of linear symmetries of F-CohFTs, which do not commute with the F-Givental action.
\end{abstract}

\maketitle

%–––––––––––––––––––––––––––––––––––––––––––%
\section{Introduction}
\label{sec:intro}
%–––––––––––––––––––––––––––––––––––––––––––%
Cohomological field theories (CohFTs for short) were introduced by Kontsevich and Manin \cite{KM94} to encode the geometric properties of Gromov--Witten invariants under degeneration of the source curve. They are collections of cohomology classes $\Omega_{g,n} \in H^{\bullet}(\Mbar_{g,n})$ compatible with the natural morphisms on $\Mbar_{g,n}$, i.e. algebras over the modular operad $H^{\bullet}(\Mbar_{g,n})$. The intersection indices of a CohFT with $\psi$-classes produce ancestor potentials associated with Dubrovin's Frobenius manifolds \cite{Dub96}. The Givental group action, first identified on potentials of Frobenius manifolds \cite{Giv01a,Giv01b}, can be lifted to (all genera) CohFTs \cite{FSZ10}. Teleman showed that this action is transitive on semisimple CohFTs of a given dimension \cite{Tel12}. It was then established that ancestor potentials of semisimple Frobenius manifolds can be reconstructed by the Eynard--Orantin topological recursion \cite{DOSS14,Eyn14,DNOPS18}. This was revisited in the formalism of Airy structures \cite{KS18,ABCO24} and leads to the fact that the ancestor potential of semisimple Frobenius manifolds can be realised as (the logarithm of) the partition function of an Airy structure. The non semisimple cases can sometimes be studied as limits of semisimple ones.

The Givental--Teleman theory thus constitutes a powerful tool to reconstruct semisimple CohFTs from their degree zero part (called topological field theory, or TFT for short), and it can be effectively applied for the computation of Gromov--Witten invariants of targets with semisimple quantum cohomology, such as (equivariant) $\mathbb{P}^1$ and toric Calabi--Yau threefolds. The supplementary results of topological recursion also had numerous applications in this vein, see e.g. \cite{EO15,FLZ16,FLZ20}. Beyond their enumerative relevance, families of CohFTs and their non-semisimple limits have been used to gain understanding on the structure of the tautological ring of $\Mbar_{g,n}$ \cite{PPZ15,PPZ19,Jan17,Jan18,CJ18,CGG}.
 
The notion of Frobenius manifolds can be weakened while keeping most of this picture. Hert\-ling and Manin introduced F-manifolds, which are Frobenius manifolds without the data of a metric \cite{HM99}. In contrast to Frobenius manifolds which admit a scalar potential, flat F-manifolds only admit a vector potential. F-CohFTs were introduced in \cite{BR21} and yield a weaker notion of CohFT in the sense that the axiom of compatibility with glueing maps having connected source curves is dropped and one marked point plays a distinguished role. The intersection indices of F-CohFTs with $\psi$-classes produce ancestor vector potentials for flat F-manifolds, there is an F-analogue of Givental's group acting on the vector potentials, and this action can be lifted to F-CohFTs \cite{ABLR23}. However, this action is far from being transitive on semisimple F-CohFTs, which therefore cannot always be reconstructed from their degree zero part. The purpose of this article is twofold: first, we complete this picture by establishing a relation to the topological recursion formalism, giving the analogue of \cite{DOSS14} to the extent possible; second, we show that F-CohFTs have a much larger group of symmetries than previously known, so that one can reasonably ask (see below) if a reconstruction \`a la Teleman would hold for semisimple F-CohFTs.

First, we propose in Section~\ref{sec:setting} a notion of F-Airy structures and define their associated vector potential by a topological recursion formula. F-Airy structures are drastically simpler than Airy structures as the non-linear constraints that ensure symmetry of the amplitudes in Airy structures are no longer needed. %--- they may actually not deserve being called structures.
In this sense, F-topological recursion can be considered as a minimal framework to define topological recursion with non-symmetric output.

We then describe in Section~\ref{sec:action} a group of symmetries of F-Airy structures which contains the F-Givental group. Airy structures were ideals of the Weyl algebra of differential operators, selecting a unique partition function which is annihilated by this ideal. The simplicity of the F-Airy structure, however, complicates the task of finding the analogue of this algebraic description, and at the time of writing, it is not fully clear what ought to replace the Weyl algebra. In particular, it is not clear for concrete examples of F-CohFTs which algebraic structure replaces the linear Virasoro constraints acting on the ancestor potential of semisimple CohFTs.

In Section~\ref{sec:FCohFTs}, after recalling the notion of F-CohFTs and their known symmetries, we define a large group of \emph{abelian} symmetries acting \emph{linearly} on F-CohFTs by exploiting the combinatorics of boundary strata of \emph{non-compact type} in the moduli space of curves. This is in stark contrast with the non-linear nature of the F-Givental action. This linear symmetry, which we call \emph{tick}, does not commute with the F-Givental action. Therefore, together with the known symmetries, they generate a large group of symmetries of F-CohFTs for which we give a global description. This group is larger than the Givental group known to act on CohFTs.

\begin{question*} 
	Does the group generated by changes of basis, F-Givental transformations (sums over stable trees), translations and ticks act transitively on the set of semisimple F-CohFTs of a given dimension?
\end{question*}

In Section~\ref{sec:identification} we set up a dictionary between the formalism of F-CohFT and the one of F-Airy structures/F-topological recursion, at the level of their symmetries. This allows for a partial extension of the results of \cite{DOSS14} to the F-case: we prove that the vector potential of F-CohFTs that are in the F-Givental group orbit of a semisimple CohFT in terms of F-topological recursion. The absence of Teleman's reconstruction in the F-world does not allow going much further at the moment. We discuss the example of the extended $2$-spin class \cite{BR21}.

In Section~\ref{sec:residue:formulation}, we introduce a formalism of F-topological recursion from the perspective of spectral curves \cite{EO07}. Compared to the original topological recursion, the fundamental bidifferential $\omega_{0,2}$ is replaced in the F-world by two (possibly non-symmetric) fundamental bidifferentials, $\conn{\omega}_{0,2}$ and $\disc{\omega}_{0,2}$, which allow the freedom to use different weights for the connected and disconnected terms in the recursion formula. Not being bound to having symmetric outputs and having lost the relation to the Weyl algebra, one could even propose a more general setting where each topological type of term appearing in the recursion could have a different weight. We do not pursue this here due to the current lack of a geometric application.

\subsubsection*{Notation and conventions.}
%–––––––––––––––––––––––––––––––––––––––––––%
Let $V$ be a vector space over $\CC$. Throughout the text, we make use of the following notations and conventions.
\begin{itemize}
	\item The space of all symmetric tensors of order $n$ defined on $V$ is denoted as $\Sym{n}{V}$, with the convention $\Sym{0}{V} \coloneqq \CC$. The completed symmetric algebra generated by $V$ is denoted $\Sym{}{V} \coloneqq \prod_{n \geq 0} \Sym{n}{V}$, equipped with its usual structure of commutative monoid denoted by the symbol $\odot$.
	
	\item The canonical pairing between $V$ and its dual is denoted as $\braket{\cdot, \cdot} \colon V^{\ast} \otimes V \to \CC$.

	\item Denote by $[n]$ the set of integers $\set{1,\dots,n}$. If $v_1,\ldots,v_n \in V$ and $J = \set{j_1 < \cdots < j_k} \subseteq [n]$, we denote $v_J \coloneqq v_{j_1} \otimes \cdots \otimes v_{j_k}$ in $V^{\otimes k}$.

	\item If $V$ is finite-dimensional with basis $(\mathrm{e}_i)_{i \in I}$, denote by $(\mathrm{e}^i)_{i \in I}$ the dual basis. For $X \in \Hom(V^{\otimes n},V^{\otimes m})$, we denote the expansion coefficients as
	\begin{equation}\label{eq:tensor:indx}
		X(\mathrm{e}_{i_1} \otimes \cdots \otimes \mathrm{e}_{i_n})
		=
		\ind{X}{j_1,\dots,j_m}{i_1,\dots,i_n} \; \mathrm{e}_{j_1} \otimes \cdots \otimes \mathrm{e}_{j_m} \,.
	\end{equation}
	We use Einstein's convention of summing over repeated indices. Following the diagrammatic conventions of tensor categories, such a tensor can be represented by drawing a box with $n$ legs attached on the top, labelled by $i_1,\ldots,i_n$, and $m$ legs attached on the bottom, labelled by $j_1,\ldots,j_m$:
	\begin{equation}\label{eq:tensor:dgrm}
		\begin{tikzpicture}[baseline,scale=.5]
			\draw (-1,-1) -- (1,-1) -- (1,1) -- (-1,1) -- cycle;
			\node at (0,0) {\small$X$};

			\draw (-.7,1) -- (-.7,2);
			\node at (-.7,2) [above] {\small$i_1$};
			\draw (.7,1) -- (.7,2);
			\node at (.7,2) [above] {\small$i_n$};
			\node at (.05,1.5) {$\dots$};

			\draw (-.7,-1) -- (-.7,-2);
			\node at (-.7,-2) [below] {\small$j_1$};
			\draw (.7,-1) -- (.7,-2);
			\node at (.7,-2) [below] {\small$j_m$};
			\node at (.05,-1.5) {$\dots$};
		\end{tikzpicture}
		\,.
	\end{equation}
	To summarise, input vectors keep falling on our heads, operations act on them sequentially during their fall and we read the output at the bottom. As a result of \eqref{eq:tensor:indx}, the rule for indices is reversed: top indices are carried by output/bottom legs while bottom indices are carried by input/top legs.
\end{itemize}

Throughout the article we have tried to provide basis-free expressions, which are conceptually more attractive. In Sections~\ref{sec:identification} and \ref{sec:residue:formulation} this leads to technical complications to handle correctly infinite-dimensional issues; if one resorts to bases (as in \cite{BR21,ABLR23}) most of these details can be safely ignored, as all expressions remain finite and well-defined.

%–––––––––––––––––––––––––––––––––––––––––––%
\section{F-Airy structures}
\label{sec:setting}
%–––––––––––––––––––––––––––––––––––––––––––%
In this section, we define the notion of F-Airy structures and give its graphical interpretation. We show that F-topological field theories fit in this formalism, extending to the F-world the relation between topological recursion and topological field theories \cite{ABO18,KS18,ABCO24}.

%–––––––––––––––––––––––––––––––––––––––––––%
\subsection{Definition of F-Airy structures in finite dimension}
%–––––––––––––––––––––––––––––––––––––––––––%
Let $V$ be a finite-dimensional vector space over $\CC$.

\begin{defn}
	An \emph{F-Airy structure} on $V$ is the data of tensors
	\begin{equation}
	\begin{split}
		& A \in \Hom(\Sym{2}{V},V) \,, \\
		& B \in \Hom(V^{\otimes 2},V) \,, \\ 
		& \conn{C} \in \Hom(V,V^{\otimes 2}) \,, \\
		& \disc{C} \in \Hom(\Sym{2}{V},V) \,, \\
		& D \in V \,.
	\end{split}
	\end{equation}
\end{defn}

Given an F-Airy structure on $V$, we define the \emph{F-topological recursion} (F-TR) \emph{amplitudes} 
\begin{equation}
	F_{g,1 + n} \in \Hom(\Sym{n}{V},V) \,,
\end{equation}
indexed by integers $g, n \geq 0$, by induction on $2g - 2 + (1 + n) > 0$. For the base cases, set $F_{0,3} \coloneqq A$ and $F_{1,1} \coloneqq D$. For $2g - 2 + (1 + n) > 1$, the induction step is
\begin{equation}\label{eq:F-TR:coord:free}
\begin{split}
	F_{g,1+n}(v_{1} \otimes \cdots \otimes v_{n})
	& \coloneqq
	\sum_{m = 1}^n B\bigl( v_m \otimes F_{g,1+(n-1)}(v_{1} \otimes \cdots \widehat{v_m} \cdots \otimes v_{n}) \bigr) \\
	& \quad +
	\frac{1}{2} (\id \otimes \tr)\bigl(\conn{C} \circ F_{g-1,1+(n+1)}(v_{1} \otimes \cdots \otimes v_{n} \otimes -)\bigr) \\
	& \quad +
	\frac{1}{2} \disc{C}\Biggl(
		\sum_{\substack{h + h' = g \\ J \sqcup J' = [n]}}
			F_{h, 1+|J|}(v_{J}) \otimes F_{h', 1+|J'|}(v_{J'})
	\Biggr) , 
\end{split}
\end{equation}
where $\tr$ is the trace of endomorphisms of $V$. By convention, we set $F_{0,1}$ and $F_{0,2}$ (called unstable terms) to zero.

To make the structure clearer, let us fix a basis $(\mathrm{e}_i)_{i \in I}$ of $V$. Following \eqref{eq:tensor:indx}, denote the expansion coefficients of the initial data and the amplitudes as
\begin{equation}
\begin{split}
	A(\mathrm{e}_j \otimes \mathrm{e}_k) & = \ind{A}{i}{j,k} \; \mathrm{e}_i \,, \\
	\qquad\qquad
	B(\mathrm{e}_j \otimes \mathrm{e}_k) & = \ind{B}{i}{j,k} \; \mathrm{e}_i \,, \\
	\conn{C}(\mathrm{e}_k) & = \indC{\conn{C}}{i,j}{k} \; \mathrm{e}_i \otimes \mathrm{e}_j\,,
	\qquad\qquad
	F_{g,1+n}(\mathrm{e}_{i_1} \otimes \cdots \otimes \mathrm{e}_{i_n}) = \indF{F}{g}{i_0}{i_1, \dots, i_n} \; \mathrm{e}_{i_0} \,.
	\\
	\disc{C}(\mathrm{e}_j \otimes \mathrm{e}_k) & = \indC{\disc{C}}{i}{j,k} \; \mathrm{e}_i \,, \\
	\qquad\qquad
	D & = D^i \; \mathrm{e}_i \,,
\end{split}
\end{equation}
Then equation \eqref{eq:F-TR:coord:free} can be written as
\begin{equation}\label{eq:F-TR:coords}
\begin{split}
	\indF{F}{g}{i_0}{i_1, \dots, i_n}
	& =
	\sum_{m = 1}^{n}
		\ind{B}{i_0}{i_m,a} \, \indF{F}{g}{a}{i_1,\ldots,\widehat{i_m},\ldots,i_n}
	+
	\frac{1}{2} \, \indC{\conn{C}}{i_0,b}{a} \, \indF{F}{g - 1}{a}{i_1,\ldots,i_n,b}
	+
	\frac{1}{2} \, \indC{\disc{C}}{i_0}{a,b}
		\sum_{\substack{h + h' = g \\ J \sqcup J' = \{i_1,\ldots,i_n\}}}
		\indF{F}{h}{a}{J} \, \indF{F}{h'}{b}{J'}
\end{split}
\end{equation}
and represented diagrammatically, according to \eqref{eq:tensor:dgrm}, as
\begin{equation}\label{eq:F-TR:diagr}
	\begin{tikzpicture}[baseline, scale=.5]
		\draw[fill=mygray] (-1,-1) -- (1,-1) -- (1,1) -- (-1,1) -- cycle;
		\node at (0,0) {\small $g$};

		\draw (0,-1) -- (0,-2);
		\node at (0,-2) [left] {\small $i_0$};

		\draw (-.8,1) -- (-.8,2);
		\node at (-.8,2) [above] {\small $i_1$};

		\node at (0,1.5) {\small $\cdots$};

		\draw (.8,1) -- (.8,2);
		\node at (.8,2) [above] {\small $i_n$};

		\node at (2.7,-.5) {$\displaystyle = \sum_{m=1}^n \vphantom{\sum_{\substack{h+h'=g \\ J \sqcup J' = \{i_1,\dots,i_n\}}}}$};

		\begin{scope}[xshift=6.2cm,yshift=1.5cm]
			\draw[fill=mygray] (-1,-1) -- (1,-1) -- (1,1) -- (-1,1) -- cycle;
			\node at (0,0) {\small $g$};

			%\draw (0,-1) -- (0,-2);

			\draw (-.8,1) -- (-.8,2);

			\node at (0,1.5) {\small$\cdots$};

			\draw (.8,1) -- (.8,2);

			\node at (.1,2.5) {\small $i_1 \;\; \widehat{i_m} \;\; i_n$};
		\end{scope}

		\begin{scope}[xshift=5.4cm,yshift=-1.5cm]
			\draw (-1,-1) -- (1,-1) -- (1,1) -- (-1,1) -- cycle;
			\node at (0,0) {\small $B$};

			\draw (0,-1) -- (0,-2);
			\node at (0,-2) [left] {\small $i_0$};

			\draw (-.8,1) -- (-.8,2);
			\draw (.8,1) -- (.8,2);
			\node at (-.8,2) [above] {\small $i_m$};
		\end{scope}

		\node at (8.4,-.5) {$\displaystyle + \; \frac{1}{2} \vphantom{\sum_{\substack{h+h'=g \\ J \sqcup J' = \{i_1,\dots,i_n\}}}}$};

		\begin{scope}[xshift=11cm,yshift=1.5cm]
			\draw[fill=mygray] (-1,-1) -- (1,-1) -- (1,1) -- (-1,1) -- cycle;
			\node at (0,0) {\small $g-1$};

			%\draw (0,-1) -- (0,-2);

			\draw (-.8,1) -- (-.8,2);
			\node at (-.8,2) [above] {\small $i_1$};

			\node at (-.1,1.5) {\small $\cdots$};

			\draw (.5,1) -- (.5,2);
			\node at (.5,2) [above] {\small $i_n$};

			\draw (.8,1) -- (.8,1.2);
			\draw (.8,1.2) arc (180:0:.6);
		\end{scope}

		\begin{scope}[xshift=11cm,yshift=-1.5cm]
			\draw (-1,-1) -- (1,-1) -- (1,1) -- (-1,1) -- cycle;
			\node at (0,0) {\small $\conn{C}$};

			\draw (-.5,-1) -- (-.5,-2);
			\node at (-.5,-2) [left] {\small $i_0$};
			\draw (.5,-1) -- (.5,-1.2);
			\draw (.5,-1.2) arc (-180:0:.75);

			\draw (0,1) -- (0,2);
		\end{scope}

		\draw (13,-2.7) -- (13,2.7);

		\node at (16,-.5) {$\displaystyle + \; \frac{1}{2} \!\! \sum_{\substack{h+h'=g \\ J \sqcup J' = \{i_1,\dots,i_n\}}}$};

		\begin{scope}[xshift=19.5cm,yshift=1.5cm]
			\draw[fill=mygray] (-1,-1) -- (1,-1) -- (1,1) -- (-1,1) -- cycle;
			\node at (0,0) {\small $h$};

			\draw (-.8,1) -- (-.8,2);

			\node at (0,1.5) {\small $\cdots$};

			\draw (.8,1) -- (.8,2);

			\node at (0,2.5) {\small $J$};
		\end{scope}

		\begin{scope}[xshift=22cm,yshift=1.5cm]
			\draw[fill=mygray] (-1,-1) -- (1,-1) -- (1,1) -- (-1,1) -- cycle;
			\node at (0,0) {\small $h'$};

			\draw (-.8,1) -- (-.8,2);

			\node at (0,1.5) {\small $\cdots$};

			\draw (.8,1) -- (.8,2);

			\node at (0,2.5) {\small $J'$};
		\end{scope}

		\begin{scope}[xshift=20.75cm,yshift=-1.5cm]
			\draw (-1,-1) -- (1,-1) -- (1,1) -- (-1,1) -- cycle;
			\node at (0,0) {\small $\disc{C}$};

			\draw (0,-1) -- (0,-2);
			\node at (0,-2) [left] {\small $i_0$};

			\draw (-.65,1) -- (-.65,2);
			\draw (.65,1) -- (.65,2);
		\end{scope}
	\end{tikzpicture} \;.
\end{equation}
There are three notable differences compared to the usual Airy structure formalism (cf. \cite{KS18,ABCO24}).
\begin{itemize}
	\item The amplitudes have one distinguished output.
	
	\item Apart from the specified symmetries, the tensors $(A,B,\conn{C},\disc{C},D)$ do not need to satisfy any constraint.
	
	\item The connected and disconnected terms can have different weights: $\conn{C}$ and $\disc{C}$.
\end{itemize}
Due to the first property, F-Airy structures do not have a scalar potential (or partition function), but rather a \emph{vector potential}. This is the $\hbar^{-1}V\bbraket{\hbar}$-valued formal function on $V$
\begin{equation}\label{eq:potential:FAiry}
	\Phi(x)
	\coloneqq
	\sum_{g,n \geq 0} \frac{\hbar^{g - 1}}{n!} \, F_{g,1+n}(x^{\otimes n})
	=
	\sum_{g,n \geq 0} \frac{\hbar^{g - 1}}{n!} \,
		\indF{F}{g}{i_0}{i_1,\ldots,i_n} \, x^{i_1} \cdots x^{i_n} \, \mathrm{e}_{i_0} \,,
\end{equation}
where $x = x^i \, \mathrm{e}_i$ denotes the variable in $V$.

If we are given a non-degenerate pairing $\eta$ on $V$ (i.e. an identification $V \cong V^*$), inputs and outputs can be considered on the same footing. After such identifications, if the tensors satisfy the condition $\conn{C} = \disc{C} = C$ and $(A,B,C,D)$ form an Airy structure, then $F_{g,1+n}$ are the TR amplitudes and are fully symmetric under permutation of all tensor factors. Besides, $\Phi(x) = \nabla F(x)$, where $F$ is a single formal function on $V$:
\begin{equation}
	F(x)
	=
	\sum_{g,n\geq 0} \frac{\hbar^{g - 1}}{(1 + n)!} \, \eta\!\left( x, F_{g,1+n}(x^{\otimes n}) \right)
	=
	\sum_{g,n \geq 0} \frac{\hbar^{g - 1}}{(1 + n)!} \,
		\indF{F}{g}{i}{i_1,\ldots,i_n} \, \eta_{i,j} \, x^{j} x^{i_1} \cdots x^{i_n} \,.
\end{equation}
Our proposal of F-topological recursion can be considered as a minimal framework to define topological recursion in absence of symmetry.

%–––––––––––––––––––––––––––––––––––––––––––%
\subsection{A graphical interpretation}
%–––––––––––––––––––––––––––––––––––––––––––%
As in the usual setting, F-TR amplitudes can be written as a sum over specific types of graphs. These graphs are obtained by repeatedly applying the F-TR formula, equation~\eqref{eq:F-TR:diagr}, which in turn generates trivalent graphs with vertices decorated by the initial data $(A,B,\conn{C},\disc{C},D)$. Besides, such graphs naturally come with a spanning tree which keeps track of the first input of the F-TR amplitudes at each step of the recursion.

\begin{defn}
	For any $g, n \geq 0$ such that $2g - 2 + (1 + n) > 0$, define the set $\mathbb{G}_{g,1+n}$ of pairs $\bm{G} = (G,\mathfrak{t})$ where:
	\begin{itemize}
		\item $G$ is a trivalent connected graph with first Betti number $b_1(G) = g$ and $1+n$ leaves, labelled as $\ell_0,\dots,\ell_n$;

		\item $\mathfrak{t} \subseteq G$ is a spanning tree which contains the first leaf $\ell_0$ (considered as the root) and no other leaf;

		\item the edges $e$ of $G$ which are not in $\mathfrak{t}$ must connect parent vertices.
	\end{itemize}
	An automorphism of $\bm{G}$ is a permutation of the set of edges which preserves the graph structure. We denote by $\Aut(\bm{G})$ the automorphism group of $\bm{G}$. We insist that $\bm{G}$ does not include the data of a cyclic order of edges or leaves incident at a vertex. Therefore, the number of automorphisms of a given $\bm{G}$ is a power of $2$. We also define
	\begin{equation}
		|\mathbb{G}_{g,1+n}|
		\coloneqq
		\sum_{\bm{G} \in \mathbb{G}_{g,1+n}} \frac{1}{|\Aut(\bm{G})|}
		\in \QQ_{> 0} \,.
	\end{equation}
\end{defn}

Consider the case $2g - 2 + (1 + n) = 1$, i.e. the sets $\mathbb{G}_{0,3}$ and $\mathbb{G}_{1,1}$. They both consist of a single element, given respectively as
\begin{equation}
	\begin{tikzpicture}[baseline]
		\node at (-1.8,0) {$\bm{G}_{0,3} = $};
		\draw [thick, ForestGreen] (0,-1) -- (0,0);
		\node at (0,-1) [right] {$\ell_0$};
		\node at (0,0) {$\bullet$};
		\draw (0,0) -- (135:1);
		\node at (135:1) [left] {$\ell_1$};
		\draw (0,0) -- (45:1);
		\node at (45:1) [right] {$\ell_2$};
		\node at (2.5,0) {and};
		\begin{scope}[xshift = 6cm]
			\node at (-1.3,0) {$\bm{G}_{1,1} = $};
			\draw [thick, ForestGreen] (0,-1) -- (0,0);
			\node at (0,-1) [right] {$\ell_0$};
			\node at (0,0) {$\bullet$};
			\draw (0,.5) circle (.5cm);
			\node at (1,0) {$\vphantom{\bm{G}_{1,1}}.$};
		\end{scope}
	\end{tikzpicture}
\end{equation}
The spanning tree is depicted in green. The graph $\bm{G}_{0,3}$ has only one trivial automorphism, while $\bm{G}_{1,1}$ has also the automorphism exchanging the two half-edges forming the loop. In particular, $|\mathbb{G}_{0,3}| = 1$ and $|\mathbb{G}_{1,1}| = \frac{1}{2}$.

It is not hard to see that $\mathbb{G}_{g,1+n}$ has a recursive structure. Indeed, by removing the vertex incident to the root $\ell_0$, we obtain a new graph for which one of the following mutually exclusive alternatives holds.
\begin{itemize}
	\item[B)] 
	The new graph $\bar{\bm{G}}$ belongs to $\mathbb{G}_{g,1+(n-1)}$ if one of the edges incident to the removed vertex is a leaf. The other edge is considered as the root of $\bar{\bm{G}}$. This situation occurs exactly $n$ times.

	\item[C${\vphantom{C}}^{\connsymb}$)] 
	The new graph $\bar{\bm{G}}$ belongs to $\mathbb{G}_{g-1,1+(n+1)}$, with an arbitrary choice of first and second leaf to be made.

	\item[C${\vphantom{C}}^{\discsymb}$)] The new graph is a disjoint union of $\bar{\bm{G}} \sqcup \bar{\bm{G}}'$, where $\bar{\bm{G}} \in \mathbb{G}_{h,1+ |J|}$ and $\bar{\bm{G}}' \in \mathbb{G}_{h',1+|J'|}$ for a splitting $h + h' = g$ of the genus and a splitting $J \sqcup J'$ of the leaves of $\bm{G}$ distinct from $\ell_0$. The roots of $\bar{\bm{G}}$ and $\bar{\bm{G}}'$ correspond to the two edges connected to $\ell_0$. 
\end{itemize}
Pictorially, we can represent the three cases as follows.
\begin{equation}\label{eq:B:Cconn:Cdisc}
	\begin{tikzpicture}[baseline]
		\draw [thick, ForestGreen] (0,-.8) -- (0,0);
		\node at (0,-.8) [left] {$\ell_0$};
		\draw (0,0) -- (135:1);
		\draw [thick, ForestGreen] (0,0) -- (45:.7);
		\node at (135:1) [left] {$\ell_m$};
		\draw [densely dotted](45:1.2) circle (.5cm);
		\node at (45:1.2) {$\bar{\bm{G}}$};
		\node at (0,0) {$\bullet$};
		\node at (0,-1.7) {(B)};
		\begin{scope}[xshift = 5cm]
			\draw [thick, ForestGreen] (0,-.8) -- (0,0);
			\node at (0,-.8) [left] {$\ell_0$};
			\draw (0,0) -- (-.4,.5);
			\draw [thick, ForestGreen] (0,0) -- (.4,.5);
			\draw [densely dotted,rounded corners] (-.8,.5) -- (.8,.5) -- (.8,1.3) -- (-.8,1.3) --cycle;
			\node at (-.5,.2) {\small$e$};
			\node at (.5,.2) {\small$\epsilon$};
			\node at (0,.9) {$\bar{\bm{G}}$};
			\node at (0,0) {$\bullet$};
			\node at (0,-1.7) {(C${\vphantom{C}}^{\connsymb}$)};
		\end{scope}
		\begin{scope}[xshift = 10cm]
			\draw [thick, ForestGreen] (0,-.8) -- (0,0);
			\node at (0,-.8) [left] {$\ell_0$};
			\draw [thick, ForestGreen] (0,0) -- (45:.7);
			\draw [thick, ForestGreen] (0,0) -- (135:.7);
			\node at (-.5,.2) {\small$e\vphantom{e'}$};
			\node at (.5,.2) {\small$e'$};
			\draw [densely dotted](45:1.2) circle (.5cm);
			\node at (45:1.2) {$\bar{\bm{G}}'$};
			\draw [densely dotted](135:1.2) circle (.5cm);
			\node at (135:1.2) {$\bar{\bm{G}}$};
			\node at (0,0) {$\bullet$};
			\node at (0,-1.7) {(C${\vphantom{C}}^{\discsymb}$)};
		\end{scope}
	\end{tikzpicture}
\end{equation}
In the first case, we have $|\Aut(\bm{G})| = |\Aut(\bar{\bm{G}})|$, while in the two last cases we have $|\Aut(\bm{G})| = 2 \, |\Aut(\bar{\bm{G}})|$ and $|\Aut(\bm{G})| = 2 \, |\Aut(\bar{\bm{G}})| \cdot |\Aut(\bar{\bm{G}}')|$. As a consequence, we find the recursive formula
\begin{equation}\label{eq:rec:Aut}
	|\mathbb{G}_{g,1+n}|
	=
	n \, |\mathbb{G}_{g,1+(n-1)}| +
	\frac{1}{2} \, |\mathbb{G}_{g-1,1+(n+1)}| +
	\frac{1}{2} \sum_{\substack{h + h' = g \\ J \sqcup J' = \{\ell_1,\ldots,\ell_n\}}}
		|\mathbb{G}_{h,1+|J|}| \cdot |\mathbb{G}_{h',1+|J'|}| \,.
\end{equation}
Consider now an F-Airy structure $(A, B, \conn{C}, \disc{C}, D)$ on $V$, together with a choice of basis $(\mathrm{e}_i)_{i \in I}$. Fix a graph $\bm{G} \in \mathbb{G}_{g,1+n}$ and a map of sets (called colouring) $c \colon \Edge^{\varnothing}(\bm{G}) \to I$, where $\Edge^{\varnothing}(\bm{G})$ denotes the set of leaves and edges of $\bm{G}$ that are not loops. We define a weight $w(\bm{G}, c)$ of the coloured graph as follows. We declare the base cases
\begin{equation}
	w(\bm{G}_{0,3},c) \coloneqq A^{c(\ell_0)}_{c(\ell_{1}),c(\ell_{2})} \,,
	\qquad\qquad
	w(\bm{G}_{1,1},c) \coloneqq D^{c(\ell_0)} \,.
\end{equation}
For $2g - 2 + (1+n) > 1$, define the weight recursively using the above decomposition. Denoting by $\bar{c}$ the restriction of $c$ to $\bar{\bm{G}}$ in the cases (B) and (C${\vphantom{C}}^{\connsymb}$) and by $\bar{c}$ (resp. $\bar{c}'$) the restriction of $c$ to $\bar{\bm{G}}$ (resp. $\bar{\bm{G}}'$) in the case (C${\vphantom{C}}^{\discsymb}$), set
\begin{itemize}
	\item[B)] 
	$w(\bm{G},c)
	\coloneqq
	\ind{B}{c(\ell_0)}{c(\ell_m),c(e)} \, w(\bar{\bm{G}},\bar{c}) \,,$ \\

	\item[C${\vphantom{C}}^{\connsymb}$)] 
	$w(\bm{G},c)
	\coloneqq
	\indC{\conn{C}}{c(\ell_0),c(\epsilon)}{c(e)} \, w(\bar{\bm{G}},\bar{c}) \,,$ \\

	\item[C${\vphantom{C}}^{\discsymb}$)] 
	$w(\bm{G},c)
	\coloneqq
	\indC{\disc{C}}{c(\ell_0)}{c(e),c(e')} \, w(\bar{\bm{G}},\bar{c}) \, w(\bar{\bm{G}}',\bar{c}') \,.$
\end{itemize}
These definitions are tailored so that the recursive formula \eqref{eq:F-TR:coords} is equivalent to the following.

\begin{prop}
	The F-TR amplitudes are given by
	\begin{equation}
		\indF{F}{g}{i_0}{i_1,\dots,i_n}
		=
		\sum_{\bm{G} \in \mathbb{G}_{g,1+n}} \sum_{c} \frac{w(\bm{G},c)}{|\Aut(\bm{G})|} \,,
	\end{equation}
	where the second sum ranges over all colourings $c \colon \Edge^{\varnothing}(\bm{G}) \to I$ satisfying $c(\ell_k) = i_k$.
\end{prop}

%–––––––––––––––––––––––––––––––––––––––––––%
\subsection{F-topological field theories}
%–––––––––––––––––––––––––––––––––––––––––––%
In the usual setting, the simplest example of Airy structure is that of a topological field theory (i.e. a Frobenius algebra) \cite{KS18,ABCO24}. The analogue of topological field theories in the F-world was introduced in \cite{ABLR23}.

\begin{defn}\label{defn:F-TFT}
	An \emph{F-topological field theory} (F-TFT for short) is the data $(V, \bcdot, w)$ of a commutative associative algebra $(V, \bcdot)$, not necessarily unital, together with a distinguished element $w \in V$. To an F-TFT is associated the collection of linear maps (called amplitudes)
	\begin{equation}\label{eq:FTFT}
		\mathscr{F}_{g,1+n} \in \Hom(\Sym{n}{V}, V) \,,
		\qquad\qquad
		\mathscr{F}_{g,1+n}(v_1 \otimes \cdots \otimes v_n) \coloneqq v_1 \bcdots v_n \bcdot w^g \,,
	\end{equation}
	indexed by $g,n \ge 0$ such that $2g-2+(1+n)>0$.
\end{defn}

The next result states that the maps $\mathscr{F}_{g,1+n}$ coincide, up to a combinatorial prefactor, with F-TR amplitudes.

\begin{prop}\label{prop:F-TFT}
	Let $(V, \bcdot , w)$ be an F-TFT. The data
	\begin{equation}
	\begin{aligned}
		A = B = C^{\discsymb}
		& \ \colon \Sym{2}{V} \longrightarrow V
		\qquad & v_1 \otimes v_2 & \longmapsto v_1 \bcdot v_2 \\
		C^{\connsymb} & \ \colon V \longrightarrow V^{\otimes 2}
		 \qquad &v& \longmapsto v \otimes w \\
		D & = \tfrac{1}{2} \, w \in V &&
	\end{aligned} 
	\end{equation}
	define an F-Airy structure on $V$ and the amplitudes of the associated F-TFT are computed by F-TR:
	\begin{equation}
	\label{eq:F-TFT:F-TR}
		|\mathbb{G}_{g,1+n}| \cdot \mathscr{F}_{g,1+n} = F_{g,1+n} \,.
	\end{equation}
\end{prop}

\begin{proof}
	Unlike the usual case, there is nothing to check for the tensors $(A, B, C^{\connsymb}, C^{\discsymb}, D)$: they automatically provide an F-Airy structure. Observe now that for the case $2g-2+(1+n) = 1$, equation~\eqref{eq:F-TFT:F-TR} holds trivially following the definition of $A$ and $D$ and the values $|\mathbb{G}_{0,3}| = 1$ and $|\mathbb{G}_{1,1}| = \frac{1}{2}$. For the general case, suppose that the recursion is satisfied for all $(g_0, n_0)$ such that $2g_0 -2 + (1+n_0) < 2g - 2 + (1+n)$. From the coordinate-free definition of F-TR, equation \eqref{eq:F-TR:coord:free}, and the induction hypothesis we find
	\begin{multline}
		F_{g,1+n}(v_{1} \otimes \cdots \otimes v_{n})
		=
		|\mathbb{G}_{g,1+(n-1)}|
		\sum_{m=1}^n B \bigl( v_m \otimes \mathscr{F}_{g,1+(n-1)}(v_{1} \otimes \cdots \widehat{v_m} \cdots \otimes v_{n}) \bigr) \\
		+
		\frac{1}{2} \, |\mathbb{G}_{g-1,1+(n+1)}| \,
		(\id \otimes \tr)\bigl(\conn{C} \circ \mathscr{F}_{g-1,1+(n+1)}(v_{1} \otimes \cdots \otimes v_{n} \otimes -) \bigr) \\
		+
		\frac{1}{2} \sum_{\substack{h + h' = g \\ J \sqcup J' = \{i_1,\ldots,i_n\}}}
		|\mathbb{G}_{h, 1+|J|}| \cdot |\mathbb{G}_{h', 1+|J'|}| \,
		\disc{C}\left(
			\mathscr{F}_{h, 1+|J|}(v_{J}) \otimes \mathscr{F}_{h', 1+|J'|}(v_{J'})
		\right)
	\end{multline}
	with the convention $\mathscr{F}_{0,1} = 0$ and $\mathscr{F}_{0,2} = 0$. From the definition of F-TFT amplitudes and that of $B$, we see that
	\begin{equation}
	\begin{split}
		B\bigl( v_m \otimes \mathscr{F}_{g,1+(n-1)}(v_{1} \otimes \cdots \widehat{v_m} \cdots \otimes v_{n})\big)
		& =
		v_m \bcdot (v_{1} \bcdots \widehat{v_m} \bcdots v_{n} \bcdot w^{g}) \\
		& =
		\mathscr{F}_{g,1+n}(v_{1} \otimes \cdots \otimes v_{n}) \,.
	\end{split}
	\end{equation}
	Notice that we repeatedly used the commutativity and the associativity of the product. The same holds for $\disc{C}$. As for $\conn{C}$, we find
	\begin{equation}
	\begin{split}
		(\id \otimes \tr)\bigl(\conn{C} \circ \mathscr{F}_{g-1,1+(n+1)}(v_{1} \otimes \cdots \otimes v_{n} \otimes - ) \bigr)
		& =
		\mathscr{F}_{g-1,1+(n+1)}(v_{1} \otimes \cdots \otimes v_{n} \otimes w) \\
		& =
		\mathscr{F}_{g,1+n}(v_{1} \otimes \cdots \otimes v_{n}) \,.
	\end{split}
	\end{equation}
	The claimed induction step follows by comparison with the recursion \eqref{eq:rec:Aut} for $|\mathbb{G}_{g,1+n}|$.
\end{proof}

%–––––––––––––––––––––––––––––––––––––––––––%
\subsection{Infinite-dimensional settings}
%–––––––––––––––––––––––––––––––––––––––––––%
A more complicated set of examples of F-Airy structures, namely those built on loop spaces, is discussed in Section~\ref{sec:identification}. It then becomes necessary to give an appropriate definition of F-Airy structure over infinite-dimensional graded vector spaces. There are several ways to do so, and the one we propose here covers the examples considered in Sections~\ref{sec:identification} and \ref{sec:residue:formulation}.

Let $(V_d)_{d \geq 0}$ be a sequence of finite-dimensional vector spaces and consider the graded vector space $V \coloneqq \bigoplus_{d \geq 0} V_d$, which may be infinite-dimensional. Let $V_{\leq d} \coloneqq \bigoplus_{d' = 0}^{d} V_{d'}$ and the natural inclusions and projections $\iota_{d} \colon V_{\leq d} \rightarrow V$ and $\pi_d \colon V \rightarrow V_{\leq d}$. Introduce the completed tensor product
\begin{equation}
	V \widehat{\otimes} V
	\coloneqq
	\prod_{d \geq 0} \bigg(\bigoplus_{d' = 0}^{d} V_{d'} \otimes V_{d - d'}\bigg)\,.
\end{equation}
An F-Airy structure on $V$ is by definition the data of tensors
\begin{equation}
\begin{split}
	& A \in \Hom(\Sym{2}{V},V) \,, \\
	& B \in \Hom(V^{\otimes 2},V) \,, \\ 
	& \conn{C} \in \Hom(V,V \widehat{\otimes} V) \,, \\
	& \disc{C} \in \Hom(\Sym{2}{V},V) \,, \\
	& D \in V \,,
\end{split}
\end{equation}
together with an increasing function $\varphi \colon \ZZ_{\geq 0} \rightarrow \ZZ_{\geq 0}$ and $\tilde{d} \in \ZZ_{\geq 0}$ such that the following finite-dimensionality conditions hold.
\begin{itemize}
	\item There exist $\tilde{A} \in \Hom(\Sym{2}{(V_{\leq \tilde{d}})},V_{\leq \tilde{d}})$ and $\tilde{D} \in V_{\leq \tilde{d}}$ such that $A = \iota_{\tilde{d}} \circ \tilde{A} \circ (\pi_{\tilde{d}})^{\otimes 2}$ and $D = \iota_{\tilde{d}} \circ \tilde{D}$.

	\item For any $d \geq 0$, we have
	\begin{equation}
	\begin{split}
		B(V \otimes V_{\leq d}) & \subseteq V_{\leq \varphi(d)} \,, \\
		(\id \otimes \pi_d)\conn{C}(V_{\leq d})  & \subseteq V_{\leq \varphi(d)} \otimes V_{\leq d} \,, \\
		\disc{C}(V_{\leq d} \otimes V_{\leq d}) & \subseteq V_{\leq \varphi(d)} \,.
	\end{split}
	\end{equation}
\end{itemize}
The datum of $(\varphi, \tilde{d})$ is often implicit when one specifies infinite-dimensional F-Airy structures. An easy induction on $2g - 2 + (1 + n) > 0$ then shows that such F-Airy structures still admit amplitudes.

\begin{lem}
	If $(A,B,\conn{C},\disc{C},D)$ is an F-Airy structure of $V$ in the above sense, then the F-TR formula \eqref{eq:F-TR:coord:free} is a well-posed definition for tensors $F_{g,1+n} \in \Hom(\Sym{n}{V},V)$. They are such that there exist
	\begin{equation}
		d \colon \ZZ_{\ge 0} \times \ZZ_{\ge 0} \longrightarrow \ZZ_{\ge 0}
		\qquad\text{and}\qquad
		\tilde{F}_{g,1+n} \in \Hom(\Sym{n}{V_{\leq d(g,n)}},V_{\leq d(g,n)})
	\end{equation}
	for which $F_{g,1+n} = \iota_{d(g,n)} \circ \tilde{F}_{g,1+n} \circ (\pi_{d(g,n)})^{\otimes n}$.
\end{lem}

Although Section~\ref{sec:action} focuses on finite-dimensional $V$, it can be adapted without difficulties to this graded infinite-dimensional setting. In the examples of Sections~\ref{sec:identification} and \ref{sec:residue:formulation} more will be said about how infinite-dimensionality is handled in practice.

%–––––––––––––––––––––––––––––––––––––––––––%
\section{Actions on F-Airy structures}
\label{sec:action}
%–––––––––––––––––––––––––––––––––––––––––––%
In this section we define three types of transformations of F-Airy structures on a fixed vector space $V$: changes of bases, Bogoliubov transformations, and translations. All transformations define left group actions on the set of F-Airy structures on $V$. We restrict ourselves to the finite-dimensional case, and comment on the infinite-dimensional setting in Section~\ref{sec:FCohFTs} in relation to F-CohFTs.

%–––––––––––––––––––––––––––––––––––––––––––%
\subsection{Change of bases}
\label{subsec:change:bases}
%–––––––––––––––––––––––––––––––––––––––––––%
The first action is rather obvious and induced by changes of bases on $V$, which can be chosen independently in the source and in the target. Given $\lambda_{\textup{s}}$, $\lambda_{\textup{t}} \in \GL(V)$ and a collection of tensors $F_{g,1+n} \in \Hom(\Sym{n}{V},V)$, define
\begin{equation}\label{eq:F:change:bases}
	\preind{\lambda}{F}_{g,1+n}
	\coloneqq
	\lambda_{\textup{t}} \circ F_{g,1+n} \circ (\lambda_{\textup{s}}^{-1})^{\otimes n} \,.
\end{equation}
From the F-TR formula \eqref{eq:F-TR:coord:free}, it follows that if $(F_{g,1+n})_{g,n \ge 0}$ are the amplitudes of an F-Airy structure $(A,B,\conn{C},\disc{C},D)$, so do $(\preind{\lambda}{F}_{g,1+n})_{g,n \ge 0}$ with initial data
\begin{equation}
\begin{split}
	& \preind{\lambda}{A}
		=
		\lambda_{\textup{t}} \circ A \circ (\lambda_{\textup{s}}^{-1})^{\otimes 2} \,, \\[.5ex]
	& \preind{\lambda}{B}
		=
		\lambda_{\textup{t}} \circ B \circ \bigl( \lambda_{\textup{s}}^{-1} \otimes \lambda_{\textup{t}}^{-1} \bigr) \,, \\[.5ex]
	& \preind{\lambda}{\vphantom{C}}{\conn{C}}
		=
		\lambda_{\textup{t}}^{\otimes 2} \circ \conn{C} \circ \lambda_{\textup{t}}^{-1} \,, \\[.5ex]
	& \preind{\lambda}{\vphantom{C}}{\disc{C}}
		=
		\lambda_{\textup{t}} \circ \disc{C} \circ \bigl( \lambda_{\textup{t}}^{-1} \bigr)^{\otimes 2} \,, \\[.5ex]
	& \preind{\lambda}{D}
		=
		\lambda_{\textup{t}} \circ D \,.
\end{split}
\end{equation}
On the vector potential $\Phi$, the change of bases simply reads
\begin{equation}
	\preind{\lambda}{\Phi} = \lambda_{\textup{t}} \circ \Phi \circ \lambda_{\textup{s}}^{-1} \,.
\end{equation}

%–––––––––––––––––––––––––––––––––––––––––––%
\subsection{Bogoliubov transformation}
\label{subsec:Bogoliubov}
%–––––––––––––––––––––––––––––––––––––––––––%
Changes of polarisation naturally arise to connect different methods of quantising symplectic vector spaces and are sometimes referred to as Bogoliubov transformations in quantum mechanics. In the context of Gromov--Witten theory, Givental introduced a specific change of polarisation coming from an R-matrix \cite{Giv01a}, and the transformation was later extended to CohFTs (see \cite{Pan19} for a self-contained exposition). In \cite{DOSS14}, the authors identified the R-action with a transformation of initial data for topological recursion. On the associated Airy structure this coincides with the natural action of a corresponding change of polarisation of a specific form and specified by the R-matrix.

Inspired by the definitions of \cite{ABLR23}, we now define the analogue of the changes of polarisation on F-Airy structures. As there is no symplectic structure here, to avoid immediate confusion we call them Bogoliubov transformations. For those coming from an R-matrix, we comment on their relation to the F-Givental action in Section~\ref{subsec:F-Giv}. We start by recalling the notion of stable trees.

\begin{defn}\label{def:stbl:tree}
	Let $g,n \geq 0$ such that $2g - 2 + (1 + n) > 0$. A \emph{stable tree} $\bm{T}$ of type $(g,1+n)$ is a tree $\bm{T}$ equipped with:
	\begin{itemize}
		\item a genus decoration $g(v)$ for each vertex $v$ of $\bm{T}$, subject to the local stability condition $2g(v) - 2 + (1 + n(v)) > 0$ and the global genus condition $g = \sum_v g(v)$;
		\item $1+n$ labelled leaves, denoted $\ell_0, \ell_1, \ldots,\ell_n$.
	\end{itemize}
	Here $1 + n(v)$ denotes the number of half-edges incident to $v$. Moreover, we consider $\bm{T}$ as being rooted at the leaf $\ell_0$. We also introduce the following notations (see Figure~\ref{fig:sble:tree} for an example):
	\begin{itemize}
		\item $\Vert(\bm{T})$, $\Edge(\bm{T})$, $\Half(\bm{T})$, and $\Leaf(\bm{T})$ are the sets of vertices, edges, half-edges, and leaves respectively;

		\item for every $v \in \Vert(\bm{T})$, $r_v$ is the half-edge that is the closest to the root, and $h \rightsquigarrow v$ refers to any half-edge $h$ different from $r_v$ (but including leaves) that is attached to $v$;

		\item every $e \in \Edge(\bm{T})$ is split into two half-edges $h_e'$ and $h_e''$ being the closest to and the furthest away from the root; we say that an edge $e$ enters (resp. exits) a vertex $v$ if $h_e'$ (resp. $h_e''$) is incident to $v$.
	\end{itemize}
	The set of stable trees of type $(g,1+n)$ is denoted $\mathbb{T}_{g,1+n}$.
\end{defn}

\begin{figure}[ht]
	\centering
	\begin{tikzpicture}[scale=1.6]
		\draw (0,-.8) -- (0,0) -- (0,1) -- (0,1.8);
		\draw ($(150:1) + (120:.8)$) -- (150:1) -- ($(150:1) + (60:.8)$);
		\draw (150:1) -- (0,0) -- (30:.8);

		\draw[ultra thick, BrickRed] (0,1) -- (0,.5);
		\draw[ultra thick, Mulberry] (0,0) -- (150:.5);
		\draw[ultra thick, BurntOrange] (150:.5) -- (150:1);

		\draw[fill=white] (0,0) circle (.125cm);
		\node (v0) at (0,0) {\tiny $1$};
		
		\draw[fill=white] (150:1) circle (.125cm);
		\node (v1) at (150:1) {\tiny $0$};

		\draw[fill=white] (0,1) circle (.125cm);
		\node (v2) at (90:1) {\tiny $2$};

 		\node at (0,-.8) [right] {$\ell_0$};
		\node at ($(150:1) + (110:1)$) {$\ell_1$};
		\node at ($(150:1) + (70:1)$) {$\ell_2$};
		\node at (0,1.8) [right] {$\ell_3$};
		\node at (30:1.1) {$\ell_4$};

		\node at (.3,1) {\small \textcolor{BrickRed}{$v$}};
		\node at (.2,.66) {\small \textcolor{BrickRed}{$r_v$}};
		\node at ($(150:.3) + (240:.25)$) {\small \textcolor{Mulberry}{$h_{e}'$}};
		\node at ($(150:.7) + (240:.25)$) {\small \textcolor{BurntOrange}{$h_{e}''$}};
	\end{tikzpicture}
	\caption{Example of a stable tree in $\mathbb{T}_{3,1+4}$ and corresponding notation. The genus decoration is depicted inside the vertices.}
	\label{fig:sble:tree}
\end{figure}
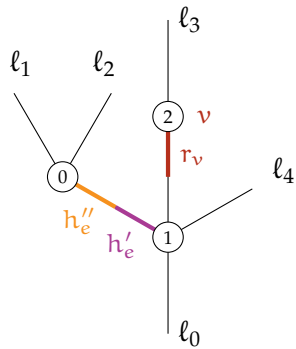

Given $\beta \in \End(V)$ and a collection of tensors $F_{g,1+n} \in \Hom(\Sym{n}{V},V)$ indexed by $g,n \geq 0$, we define new amplitudes by the formula
\begin{equation}\label{eq:Bglbv:action}
	\preind{\beta}{F}_{g,1+n}
	\coloneqq
	\sum_{\bm{T} \in \mathbb{T}_{g,1+n}}
		\Biggl( \bigotimes_{v \in \Vert(\bm{T})} F_{g(v),1+n(v)} \Biggr)
		\underset{\bm{T}}{\circ}
		\Biggl( \bigotimes_{e \in \Edge(\bm{T})} \beta \Biggr)
		\,.
\end{equation}
Here $\circ_{\bm{T}}$ means that we compose the tensors in the natural way along the edges of the stable tree $\bm{T}$: if an edge $e$ connects two vertices $v'$ and $v''$, the output of $F_{g(v''),1+n(v'')}$ (corresponding to $h_e''$) is inserted as input in $\beta$, while the output of $\beta$ (corresponding to $h_e'$) is inserted as input in $F_{g(v'),1+n(v')}$. In a basis of $V$ indexed by $I$, writing $\beta(\mathrm{e}_j) = \beta^i_j \; \mathrm{e}_i$, this means
\begin{equation}\label{eq:Bglbv:action:coords}
	\indF{(\preind{\beta}{F})}{g}{i_0}{i_1, \dots, i_n}
	=
	\sum_{\bm{T} \in \mathbb{T}_{g,1+n} }
	\sum_{j}
	\Biggl( \prod_{ v \in \Vert(\bm{T}) }
			\indF{F}{g(v)}{j(r_v)}{( j( h ) )_{ h \rightsquigarrow v }}
	\Biggr)
	\Biggl( \prod_{ e \in \Edge(\bm{T}) }
			\ind{\beta}{j(h_e')}{j(h_e'')}
	\Biggr)
	\,,
\end{equation}
where the second sum ranges over all weights on half-edges that respect the leaf decorations, that is the set of maps $j \colon \Half(\bm{T}) \to I$ such that $j(\ell_k) = i_k$. Notice the absence of symmetry factors in this formula, since stable trees do not have non-trivial automorphisms (we did not order the edges entering vertices). We can now state the following result.

\begin{thm}\label{thm:Bglbv:FAiry}
	If $(F_{g,1+n})_{g,n}$ are the amplitudes of an F-Airy structure $(A,B,\conn{C},\disc{C},D)$ on $V$, then the $(\preind{\beta}{F}_{g,1+n})_{g,n}$ defined in \eqref{eq:Bglbv:action} coincide with the amplitudes of the F-Airy structure given by
	\begin{equation}\label{eq:Bglbv:data}
	\begin{split}
		& \preind{\beta}{A} = A \,, \\
		& \preind{\beta}{B} = B + A \circ (\id_V \otimes \beta)\,, \\
		& \preind{\beta}{\vphantom{C}} \conn{C} = \conn{C} \,, \\
		& \preind{\beta}{\vphantom{C}} \disc{C} = \disc{C} + B \circ (\beta \otimes \id_V) + B \circ (\beta \otimes \id_V) \circ \sigma_{1,2} + A \circ \beta^{\otimes 2} \,, \\
		& \preind{\beta}{D} = D \,,
	\end{split}
	\end{equation}
	where $\sigma_{1,2} \colon V^{\otimes 2} \rightarrow V^{\otimes 2}$ is the permutation of the two tensor factors. For index lovers, the coefficients of the (non-trivially) modified tensors read:
	\begin{equation}\label{eq:Bglbv:data:coords}
	\begin{split}
		& \preind{\beta}{\vphantom{B}}\ind{B}{i}{j,k}
		=
		\ind{B}{i}{j,k}
		+
		\ind{A}{i}{j,a} \, \ind{\beta}{a}{k} \,, \\
		& \preind{\beta}{\vphantom{C}}\indC{\disc{C}}{i}{j,k}
		=
		\indC{\disc{C}}{i}{j,k}
		+ 
		\ind{A}{i}{a,b} \, \ind{\beta}{a}{j} \, \ind{\beta}{b}{k}
		+
		\ind{B}{i}{a,k} \, \ind{\beta}{a}{j}
		+
		\ind{B}{i}{b,j} \, \ind{\beta}{b}{k} \,.
	\end{split}
	\end{equation}
	For diagrammatic fans, the (non-trivially) modified tensors are pictured as follows:
	\begin{equation}\label{eq:Bglbv:data:dgrms}
		\begin{tikzpicture}[baseline, scale=.4]
		% B
			\draw (-1,-1) -- (1,-1) -- (1,1) -- (-1,1) -- cycle;
			\node at (0,0) {\small ${}^{\beta}B$};

			\draw (0,-1) -- (0,-2);
			\node at (0,-2) [left] {\small $i$};

			\draw (-.5,1) -- (-.5,2);
			\node at (-.5,2) [left] {\small $j$};

			\draw (.5,1) -- (.5,2);
			\node at (.5,2) [right] {\small $\vphantom{j}k$};

			\node at (2,0) {$=$};

			\begin{scope}[xshift=4cm]
				\draw (-1,-1) -- (1,-1) -- (1,1) -- (-1,1) -- cycle;
				\node at (0,0) {\small $B$};

				\draw (0,-1) -- (0,-2);
				\node at (0,-2) [left] {\small $i$};

				\draw (-.5,1) -- (-.5,2);
				\node at (-.5,2) [left] {\small $j$};

				\draw (.5,1) -- (.5,2);
				\node at (.5,2) [right] {\small $\vphantom{j}k$};

				\node at (2,0) {$+$};
			\end{scope}
			\begin{scope}[xshift=8cm,yshift=-.5cm]
				\draw (-1,-1) -- (1,-1) -- (1,1) -- (-1,1) -- cycle;
				\node at (0,0) {\small $A$};

				\draw (0,-1) -- (0,-2);
				\node at (0,-2) [left] {\small $i$};

				\draw (-.5,1) -- (-.5,3);
				\node at (-.5,3.1) [left] {\small $j$};

				\draw (.5,1) -- (.5,1.5);
				\draw (0,1.5) -- (1,1.5) -- (1,2.5) -- (0,2.5) -- cycle;
				\draw (.5,2.5) -- (.5,3);

				\node at (.5,2) {\small $\beta$};
				\node at (.5,3.1) [right] {\small $\vphantom{j}k$};

				\node at (2,0) {$\vphantom{=},$};
			\end{scope}
		% C disconnected
		\begin{scope}[xshift=15cm]
			
			\draw (-1,-1) -- (1,-1) -- (1,1) -- (-1,1) -- cycle;
			\node at (0,0) {\small ${}^{\beta} \disc{C}$};

			\draw (0,-1) -- (0,-2);
			\node at (0,-2) [left] {\small $i$};

			\draw (-.5,1) -- (-.5,2);
			\node at (-.5,2) [left] {\small $j$};

			\draw (.5,1) -- (.5,2);
			\node at (.5,2) [right] {\small $\vphantom{j}k$};

			\node at (2,0) {$=$};

			\begin{scope}[xshift=4cm]
				\draw (-1,-1) -- (1,-1) -- (1,1) -- (-1,1) -- cycle;
				\node at (0,0) {\small$\disc{C}$};

				\draw (0,-1) -- (0,-2);
				\node at (0,-2) [left] {\small $i$};

				\draw (-.5,1) -- (-.5,2);
				\node at (-.5,2) [left] {\small $j$};

				\draw (.5,1) -- (.5,2);
				\node at (.5,2) [right] {\small $\vphantom{j}k$};

				\node at (2,0) {$+$};
			\end{scope}
			\begin{scope}[xshift=8cm,yshift=-.5cm]
				\draw (-1,-1) -- (1,-1) -- (1,1) -- (-1,1) -- cycle;
				\node at (0,0) {\small $B$};

				\draw (0,-1) -- (0,-2);
				\node at (0,-2) [left] {\small $i$};

				\draw (-.5,1) -- (-.5,1.5);
				\draw (0,1.5) -- (-1,1.5) -- (-1,2.5) -- (0,2.5) -- cycle;
				\draw (-.5,2.5) -- (-.5,3);

				\node at (-.5,2) {\small $\beta$};
				\node at (-.5,3.1) [left] {\small $j$};

				\draw (.5,1) -- (.5,3);
				\node at (.5,3.1) [right] {\small $\vphantom{j}k$};
			\end{scope}
			\node at (10,0) {$+$};
			\begin{scope}[xshift=12cm,yshift=-.5cm]
				\draw (-1,-1) -- (1,-1) -- (1,1) -- (-1,1) -- cycle;
				\node at (0,0) {\small $B$};

				\draw (0,-1) -- (0,-2);
				\node at (0,-2) [left] {\small $i$};

				\draw (-.5,1) -- (-.5,1.5);
				\draw (0,1.5) -- (-1,1.5) -- (-1,2.5) -- (0,2.5) -- cycle;
				\draw (-.5,2.5) -- (-.5,3);

				\node at (-.5,2) {\small $\beta$};
				\node at (-.5,3.1) [left] {\small $\vphantom{j}k$};

				\draw (.5,1) -- (.5,3);
				\node at (.5,3.1) [right] {\small $j$};
			\end{scope}
			\node at (14,0) {$+$};
			\begin{scope}[xshift=16cm,yshift=-.5cm]
				\draw (-1,-1) -- (1,-1) -- (1,1) -- (-1,1) -- cycle;
				\node at (0,0) {\small $A$};

				\draw (0,-1) -- (0,-2);
				\node at (0,-2) [left] {\small $i$};

				\draw (-.7,1) -- (-.7,1.5);
				\draw (-.2,1.5) -- (-1.2,1.5) -- (-1.2,2.5) -- (-.2,2.5) -- cycle;
				\draw (-.7,2.5) -- (-.7,3);

				\node at (-.7,2) {\small $\beta$};
				\node at (-.7,3.1) [left] {\small $j$};

				\draw (.7,1) -- (.7,1.5);
				\draw (.2,1.5) -- (1.2,1.5) -- (1.2,2.5) -- (.2,2.5) -- cycle;
				\draw (.7,2.5) -- (.7,3);

				\node at (.7,2) {\small $\beta$};
				\node at (.7,3.1) [right] {\small $\vphantom{j}k$};

				\node at (2,0) {$\vphantom{=}.$};
			\end{scope}
		\end{scope}
		\end{tikzpicture}
	\end{equation}
\end{thm}

\begin{proof}
	In the $(0,3)$ and $(1,1)$ cases, the statement follows from the fact that there is only one stable tree of type $(0,3)$ and $(1,1)$ respectively. Suppose now that $2g-2+(1+n) > 1$. We can reformulate \eqref{eq:Bglbv:action:coords} by writing apart the term corresponding to the root vertex $v_0$:
	\begin{equation}\label{eq:beta:F:root}
		\indF{(\preind{\beta}{F})}{g}{i_0}{i_1, \dots, i_n}
		=
		\sum_{ \bm{T} \in \mathbb{T}_{g,1+n} } \sum_{j}
			\indF{F}{g(v_0)}{i_0}{( j( h ) )_{ h \rightsquigarrow v_0 }}
			\prod_{\substack{ e \in \Edge(\bm{T}) \\ h_e' \rightsquigarrow v_0 }} \ind{\beta}{j(h_e')}{j(h_e'')}
		\prod_{\substack{ v \in \Vert(\bm{T}) \\ v \neq v_0 }}
			\indF{F}{g(v)}{j(r_v)}{( j( h ) )_{ h \rightsquigarrow v }}
			\prod_{\substack{ e \in \Edge(\bm{T}) \\ h_e' \not\rightsquigarrow v_0 }} \ind{\beta}{j(h_e')}{j(h_e'')} \,.
	\end{equation}
	By looking at the topological type of the root vertex, namely $(g(v_0),1+n(v_0))$, three mutually exclusive situations can occur.

	\textit{Three-holed sphere.} In this case, $F_{g(v_0),1+n(v_0)} = A$. Among the two half-edges entering $v_0$, those which are not leaves appear in a factor of $\beta$ while the other do not, and there is either one leaf, decorated by $i_m$ for some $m \in [n]$, or no leaf at all (we cannot have two leaves because we assumed $(g,1+n) \neq (0,3)$). In the first case, $\bm{T} \setminus v_0$ is a stable tree $\bar{\bm{T}} \in \mathbb{T}_{g,1 + (n-1)}$, while in the second case it is the union of two stable trees $\bar{\bm{T}} \in \mathbb{T}_{h,1 + |J|}$ and $\bar{\bm{T}}' \in \mathbb{T}_{h',1 + |J'|}$ for a splitting $h+h'=g$ of the genus and a splitting $J \sqcup J' = \{i_1,\dots,i_n\}$ of the leaves' labels:
	\begin{equation}
		\begin{tikzpicture}[baseline]
			\draw (0,-.8) -- (0,0);
			\node at (0,-.8) [left] {\small$i_0$};
			\draw (0,0) -- (135:1);
			\draw (0,0) -- (45:.7);
			\node at (135:1) [left] {\small$i_m$};
			\draw [densely dotted](45:1.2) circle (.5cm);
			\node at (45:1.2) {$\bar{\bm{T}}$};
			\draw[fill=white] (0,0) circle (.2cm);
			\node at (0,0) {\tiny $0$};
			\node at (3,0) {or};
			\begin{scope}[xshift = 6cm]
				\draw (0,-.8) -- (0,0);
				\node at (0,-.8) [left] {\small$i_0$};
				\draw (0,0) -- (45:.7);
				\draw (0,0) -- (135:.7);
				\draw [densely dotted](45:1.2) circle (.5cm);
				\node at (45:1.2) {$\bar{\bm{T}}'$};
				\draw [densely dotted](135:1.2) circle (.5cm);
				\node at (135:1.2) {$\bar{\bm{T}}$};
				\draw[fill=white] (0,0) circle (.2cm);
				\node at (0,0) {\tiny $0$};
			\end{scope}
		\end{tikzpicture}
		\;.
	\end{equation}
	Summing over $\bm{T} \setminus v_0$, we recognise
	\begin{equation}\label{eq:beta:F:1}
		\sum_{m=1}^{n}
			\ind{A}{i_0}{i_m,a} \, \ind{\beta}{a}{b} \, \indF{(\preind{\beta}{F})}{g}{b}{i_1, \dots \widehat{i_{m}} \dots, i_n}
		+
		\frac{1}{2} \sum_{\substack{ h+h' = g \\ J \sqcup J' = \{ i_1, \dots, i_n\} }}
			\ind{A}{i_0}{a_{1},b_{1}} \,
			\ind{\beta}{a_{1}}{a_{2}} \, \ind{\beta}{b_{1}}{b_{2}} \,
			\indF{(\preind{\beta}{F})}{h}{a_{2}}{J} \,
			\indF{(\preind{\beta}{F})}{h'}{b_{2}}{J'} \,.
	\end{equation}

	\textit{One-holed torus.} If $2g-2+(1+n) > 1$, this situation can never happen (otherwise $v_0$ would be the only vertex).

	\textit{Higher topologies.} In this case, we can employ the recursion formula \eqref{eq:F-TR:coords} for $F_{g(v_0),1+n(v_0)}$. To set the notation, suppose that $v_0$ is attached to $s$ leaves different from the root and $t$ additional half-edges. The root is decorated by $i_0$, the $s$ leaves are decorated by $I_s = \{ i_{k_1},\dots,i_{k_s} \}$, while the remaining half-edges are decorated by $J_t = \{ j_1,\dots,j_t \}$. Denoting $I_s^{[m]} = I_s \setminus \{ i_{k_m} \}$ and likewise $J_s^{[l]} = J_s \setminus \{ j_l \}$, we have
	\begin{multline}
		\indF{F}{g(v_0)}{i_0}{I_s \sqcup J_t}
		=
		\sum_{m = 1}^{s} \ind{B}{i_0}{i_{k_m},a} \, \indF{F}{g(v_0)}{a}{I_s^{[m]} \sqcup J_t} 
		+
		\sum_{l = 1}^{t} \ind{B}{i_0}{j_{l},a} \, \indF{F}{g(v_0)}{a}{I_s \sqcup J_t^{[l]}} \\
		+
		\frac{1}{2} \indC{\conn{C}}{i_0,b}{a} \, \indF{F}{g(v_0) - 1}{a}{I_s \sqcup J_t,b}
		+
		\frac{1}{2} \indC{\disc{C}}{i_0}{a,b}
			\sum_{\substack{h + h' = g(v_0) \\ K \sqcup K' = I_s \sqcup J_t}}
			\indF{F}{h}{a}{K} \, \indF{F}{h'}{b}{K'} \,.
	\end{multline}
	As before, in each of the situations we can sum over $\bm{T} \setminus v_0$ and recognise some $\preind{\beta}{F}$. For the second sum involving $B$, there is an edge entering $v_0$ with index $j_l = a$, and this index comes in \eqref{eq:beta:F:root} from a factor $\beta_{a}^{c}$ preceded by a sum reconstructing some $\preind{\beta}{F}$ with output index $c$. This results in the contributions
	\begin{multline}\label{eq:beta:F:2}
		\sum_{m=1}^{n}
			\ind{B}{i_0}{i_m,a} \, \indF{(\preind{\beta}{F})}{g}{a}{i_1, \dots \widehat{i_{m}} \dots, i_n}
		+
		\sum_{\substack{ h+h' = g \\ J \sqcup J' = \{ i_1, \dots, i_n\} }}
			\ind{B}{i_0}{a,b} \, \ind{\beta}{a}{c} \, \indF{(\preind{\beta}{F})}{h}{b}{J} \, \indF{(\preind{\beta}{F})}{h'}{c}{J'}
		\\
		+
		\frac{1}{2} \, \indC{\conn{C}}{i_0,b}{a} \, \indF{(\preind{\beta}{F})}{g-1}{a}{i_1, \dots, i_n,b}
		+
		\frac{1}{2} \sum_{\substack{ h+h' = g \\ J \sqcup J' = \{ i_1, \dots, i_n\} }}
			\indC{\disc{C}}{i_0}{a,b} \, \indF{(\preind{\beta}{F})}{h}{a}{J} \, \indF{(\preind{\beta}{F})}{h'}{b}{J'}
		\,.
	\end{multline}
	The quantity $\indF{(\preind{\beta}{F})}{g}{i_0}{i_1,\dots,i_n}$ is the sum of \eqref{eq:beta:F:1} and \eqref{eq:beta:F:2}: this takes the form of the F-TR formula \eqref{eq:F-TR:coords} with the modified tensors $\preind{\beta}{B}$, $\preind{\beta}{\vphantom{C}}\conn{C}$ and $\preind{\beta}{\vphantom{C}}\disc{C}$ from equation~\eqref{eq:Bglbv:data}.
\end{proof}

We can relate the vector potentials $\Phi$ and $\preind{\beta}{\Phi}$ before and after Bogoliubov transformation through a fixed point equation.

\begin{lem}\label{lem:fixed:pnt}
	The vector potential $\preind{\beta}{\Phi}$ after Bogoliubov transformation is uniquely characterised by the fixed point equation
	\begin{equation}\label{eq:fixed:pnt}
		\preind{\beta}{\Phi}(x) = \Phi\Bigl( x + \hbar (\beta \circ \preind{\beta}{\Phi})(x) \Bigr) .
	\end{equation}
\end{lem}

\begin{proof}
	We multiply \eqref{eq:beta:F:root} by $\frac{\hbar^{g - 1}}{n!} x^{i_1} \cdots x^{i_n} \mathrm{e}_{i_0}$ and sum over leaves decorations $i_0,i_1,\ldots,i_n \in I$ and topologies $(g, n)$ such that $2g - 2 + (1 + n) > 0$. On the left-hand side, we find $\preind{\beta}{\Phi}(x)$. On the right-hand side, we have a sum over stable trees $\bm{T}$ of any topology where the $n$ leaves are labelled (but we divide by $n!$) and a sum over indices $j(h)$ for each half-edge (including leaves such that $j(h_a) = i_a$ for the $a$-th leaf).

	Looking at the summand corresponding to a fixed stable tree $\bm{T}$ on the right-hand side, by removing the root vertex $v_0$ we obtain a collection of leaves $\ell_1,\ldots,\ell_{s}$ (which we order by increasing carried label) and a collection of stable trees $\bm{T}_1,\ldots,\bm{T}_t$ (which we order by increasing minimum leaf label), such that $s + t \geq 1$. Conversely, if we are given an ordered set of leaves $\ell_1,\ldots,\ell_{s}$ and an ordered set of stable trees $\bm{T}_1,\ldots,\bm{T}_t$, connecting them to a root vertex of valency $1 + (s + t)$ gives a stable tree $\bm{T}$ after one chooses a partition of $[n]$ into $1 + t$ pairwise disjoint non-empty subsets $\Lambda_0,\ldots,\Lambda_t$, where $|\Lambda_0| = s$ and $|\Lambda_k| = n_k$ is the number of leaves in $\bm{T}_k$, excluding the root. The decomposition of any such $\bm{T}$, after a reordering of the list $(\bm{T}_1, \dots, \bm{T}_t)$ so that the minimum leaf label of $\bm{T}_k$ increases with $k$, produces the ordered set of leaves and stable trees we started with:
	
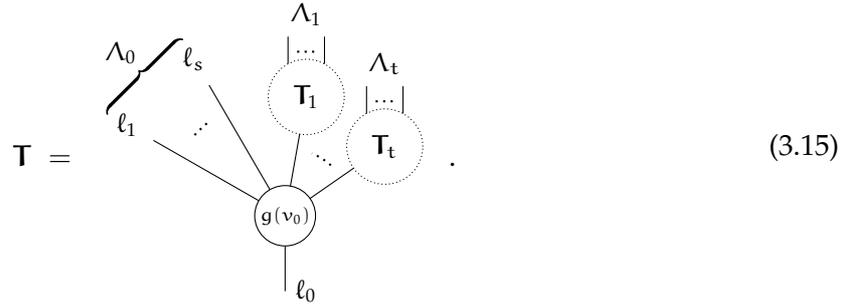
\begin{equation}
		\begin{tikzpicture}[baseline, scale=2]
		\begin{scope}[yshift=-.4cm]
			\node at (-1.6,.4) {$\bm{T} \; =$};
			\draw (0,-.5) -- (0,0);
			\node at (0,-.5) [right] {\small$\ell_0$};
			\draw (0,0) -- (120:1);
			\node at (120:1.2) {\small$\ell_s$};
			\draw (0,0) -- (150:1);
			\node at (150:1.2) {\small$\ell_1$};
			\draw [dotted, thick] (130:.8) arc (130:140:.8);
			\draw [decorate,decoration={calligraphic brace,amplitude=3pt},line width=1.25pt,xshift=0pt,yshift=0pt] (150:1.35) -- (120:1.35) node [black,midway,xshift=-9pt,yshift=9pt] {\small $\Lambda_0$};
			\draw (0,0) -- (80:.8);
			\draw (0,0) -- (35:.8);
			\draw [dotted, thick] (48:.45) arc (48:67:.45);
			\draw ($(80:.8) + (-.12,0)$) -- ($(80:.8) + (-.12,.4)$);
			\draw ($(80:.8) + (.12,0)$) -- ($(80:.8) + (.12,.4)$);
			\draw ($(35:.8) + (-.12,0)$) -- ($(35:.8) + (-.12,.4)$);
			\draw ($(35:.8) + (.12,0)$) -- ($(35:.8) + (.12,.4)$);
			\draw [dotted, thick] ($(80:.8) + (-.06,.3)$) -- ($(80:.8) + (.06,.3)$);
			\draw [dotted, thick] ($(35:.8) + (-.06,.3)$) -- ($(35:.8) + (.06,.3)$);
			\node at ($(80:.8) + (0,.55)$) {\small$\Lambda_1$};
			\node at ($(35:.8) + (0,.55)$) {\small$\Lambda_t$};
			\draw [densely dotted, fill=white] (80:.8) circle (.25cm);
			\draw [densely dotted, fill=white] (35:.8) circle (.25cm);
			\node at (80:.8) {\small$\bm{T}_1$};
			\node at (35:.8) {\small$\bm{T}_t$};
			\draw[fill=white] (0,0) circle (.2cm);
			\node at (0,0) {\tiny $g(v_0)$};
			\node at (1.1,.4) {$\vphantom{\bm{T}}.$};
		\end{scope}
		\end{tikzpicture}
	\end{equation}
	Given indices carried by half-edges and leaves, the weight of $\bm{T}$ only depends on the decomposition of $\bm{T}$ pictured above, and takes the factorised form:
	\begin{equation}
		\frac{\hbar^t}{n!} \,
		\hbar^{g(v_0) - 1} \,
		F_{g(v_0);(j(h))_{h \rightsquigarrow v_0}}^{i_0}
		\Biggl(
			\prod_{\substack{e \in \Edge(\bm{T}) \\ h_e' \rightsquigarrow v_0}}
				\hbar \beta_{j(h_e'')}^{j(h_e')}
		\Biggr)\Biggl(
			\prod_{\substack{\ell \in \Leaf(\bm{T}) \\ \ell \rightsquigarrow v_0}}
				x^{i_{\ell}}
		\Biggr)\Biggl(
			\prod_{k = 1}^{t}
				\hbar^{g(\bm{T}_k)- 1} \, w\bigl( \bm{T}_k; j|_{\bm{T}_k} \bigr)
		\Biggr).
	\end{equation}
	Here $g = g(v_0) + g(\bm{T}_1) + \cdots + g(\bm{T}_t)$, where $g(\bm{T}_k)$ is the genus of the stable tree $\bm{T}_k$. As we are summing over indices carried by half-edges and leaves, the factors of the decomposition play a symmetric role. Therefore, a given decomposition gives rise to $\frac{n!}{s! t! n_1! \cdots n_t!}$ equal terms, and since we allow for any order for the minimum leaf labels in the list of trees, this should be multiplied by $t!$. Thus, the sums over $\bm{T}_1, \dots, \bm{T}_t$ can be performed independently, and give rise to the coefficient of $\text{e}_{j(h_{e_k}'')}$ in $\preind{\beta}{\Phi}(x)$, where $e_k$ is the edge between $\bm{T}_k$ and $v_0$. Summing over $j(h_{e_k}'')$ recombines with a factor of $\hbar\beta$ and produces the coefficient of $\text{e}_{j(h_{e_k}')}$ in $\hbar \beta (\preind{\beta}{\Phi}(x))$. Taking the remaining sum over the indices of half-edges and leaves, we recognise the tensor $F_{g(v_0),1+s+t}$ applied to
	\begin{equation}
		\frac{1}{s! t!} x^{\otimes s} \otimes \Bigl( \hbar \beta \bigl( \preind{\beta}{\Phi}(x) \bigr) \Bigr)^{\otimes t}
		=
		\frac{1}{(s + t)!} \Bigl( x + \hbar \beta \bigl( \preind{\beta}{\Phi}(x) \bigr) \Bigr)^{\otimes (s + t)} \,.
	\end{equation}
	Summing over $g(v_0)$ and $s, t \geq 0$ with $s+t \geq 1$ yields $\Phi\big( x + \hbar\beta(\preind{\beta}{\Phi}(x)) \big)$.
\end{proof}

%–––––––––––––––––––––––––––––––––––––––––––%
\subsection{Translation}
\label{subsec:translation}
%–––––––––––––––––––––––––––––––––––––––––––%
The vector potential $\Phi(x)$ of an F-Airy structure is an $\hbar^{-1}V\bbraket{\hbar}$-valued function defined on the formal neighbourhood of $0 \in V$. Suppose that $\Phi$ has a non-zero radius of convergence. Then we can consider its Taylor expansion at another point $t \ne 0$ in the domain of convergence. Even though the expansion at $0$ does not contains the terms $F_{0,1}$ and $F_{0,2}$ (i.e. no constant and linear term in the coefficient of $\hbar^{-1}$), this is not the case for the expansion at $\tau$. Nonetheless, it is natural to expect that the expansion at $\tau$ coincides with the vector potential $\preind{\tau}{\Phi}(x)$ of a translated F-Airy structure, after we remove the non-zero `translated' $F_{0,1}$ and $F_{0,2}$ terms.

In order to make sense of the considerations without convergence assumptions, we are going to consider families of F-Airy structures depending on a formal parameter near $0 \in W$, where $W$ is a vector space coming with an isomorphism $\iota \colon W \to V$. By family of F-Airy structures in this case we mean tensors
\begin{equation}
\begin{split}
	& \tilde{A} \in \Hom(\Sym{2}{V} \otimes \Sym{}{W},V) \,, \\
	& \tilde{B} \in \Hom(V^{\otimes 2} \otimes \Sym{}{W},V) \,, \\ 
	& \conn{\tilde{C}} \in \Hom(V \otimes \Sym{}{W},V \otimes V) \,, \\
	& \disc{\tilde{C}} \in \Hom(\Sym{2}{V} \otimes \Sym{}{W},V) \,, \\
	& \tilde{D} \in \Hom(\Sym{}{W},V) \,,
\end{split}
\end{equation}
where we recall $\Sym{}{W} = \prod_{m \geq 0} \Sym{m}{W}$ and likewise for $\Sym{}{(W^{\ast})}$. We could take $W = V$ and $\iota = \id_V$ without loss of generality, but the main point of the discussion is that $V$ and $W$ play two different roles and it is easier to keep track of this by giving them different names.
 
Given a collection of tensors $F_{g,1+n} \in \Hom(\Sym{n}{V},V)$ indexed by $g,n \geq 0$, we define new amplitudes $\tilde{F}_{g,1+n} \in \Hom(\Sym{n}{V} \otimes \Sym{}{W}, V)$ by the formula
\begin{equation}\label{eq:t:action}
	\tilde{F}_{g,1+n} ( v_1 \otimes \cdots \otimes v_n )
	\coloneqq
	\sum_{m \ge 0}
		\frac{1}{m!} \,
		F_{g,1+n+m} \bigl( v_1 \otimes \cdots \otimes v_n \otimes \iota^{\otimes m} \bigr)
	\,.
\end{equation}
Alternatively, we can consider $\tilde{F}_{g,1+n}$ as an element of $\Hom(\Sym{n}{V},V \otimes \Sym{}{(W^{\ast})})$, or instead as a $\Hom(\Sym{n}{V},V)$-valued formal function near $0$ on $W \cong V$.

The construction is perhaps more transparent in coordinates. For a fixed basis $(\mathrm{e}_i)_{i \in I}$ of $V$, let $(\mathrm{e}^i)_{i \in I}$ be the dual basis. The algebra $\Sym{}{(W^{\ast})}$ is identified through $\iota$ with $\CC\bbraket{(\mathrm{e}^i)_{i \in I}}$, graded by assigning degree $1$ to each $\mathrm{e}^{i}$. The translated amplitudes in coordinates are the following elements of $\CC\bbraket{(\mathrm{e}^i)_{i \in I}}$:
\begin{equation}\label{eq:translation:coords}
	\indF{\tilde{F}}{g}{i_0}{i_1, \dots, i_n}
	=
	\sum_{m \geq 0} \frac{1}{m!} \,
		\indF{F}{g}{i_0}{i_1, \dots, i_n, j_1, \dots, j_m}
		\,
		\mathrm{e}^{j_1} \cdots \mathrm{e}^{j_m} \,.
\end{equation}
Besides, it can be represented diagrammatically as
\begin{equation}
\begin{tikzpicture}[baseline,scale=.5]
	\draw[Red] (.4,1) -- (.4,2);
	\node at (1.25,1.5) {\small $\cdots$};
	\draw[Red] (2,1) -- (2,2);

	\draw[fill=mygray] (-2.2,-1) -- (2.2,-1) -- (2.2,1) -- (-2.2,1) -- cycle;
	\node at (0,0) {\small $g$};

	\draw (0,-1) -- (0,-2);
	\node at (0,-2) [left] {\small $i_0$};

	\draw (-2,1) -- (-2,2);
	\node at (-2,2) [above] {\small $i_1$};

	\node at (-1.15,1.5) {\small $\cdots$};

	\draw (-.4,1) -- (-.4,2);
	\node at (-.2,2) [above] {\small $i_n$};

	\node at (4.5,0) {$\displaystyle = \sum_{m \ge 0} \frac{1}{m!}$};

	\begin{scope}[xshift=9cm]
		\draw[fill=mygray] (-2.2,-1) -- (2.2,-1) -- (2.2,1) -- (-2.2,1) -- cycle;
		\node at (0,0) {\small $g$};

		\draw (0,-1) -- (0,-2);
		\node at (0,-2) [left] {\small $i_0$};

		\draw (-2,1) -- (-2,2);
		\node at (-2,2) [above] {\small $i_1$};

		\node at (-1.15,1.5) {\small $\cdots$};

		\draw (-.4,1) -- (-.4,2);
		\node at (-.2,2) [above] {\small $i_n$};

		\draw (.4,1) -- (.4,1.5);
		\draw[Red] (.4,1.5) -- (.4,2);
		\node[white] at (.4,1.5) {$\bullet$};
		\node at (.4,1.5) {$\circ$};
		\node at (1.25,1.5) {\small $\cdots$};
		\draw (2,1) -- (2,1.5);
		\draw[Red] (2,1.5) -- (2,2);
		\node[white] at (2,1.5) {$\bullet$};
		\node at (2,1.5) {$\circ$};

		\draw [decorate,decoration={calligraphic brace,amplitude=3pt},line width=1.25pt,xshift=0pt,yshift=4pt] (.3,2) -- (2.1,2) node [black,midway,yshift=9pt] {\footnotesize $m$};
	\end{scope}
\end{tikzpicture} \ .
\end{equation}
Here $\begin{tikzpicture}[baseline=-2pt]\draw[Red](0,.25)--(0,0);\draw(0,0)--(0,-.25);\node[white]at(0,0){$\bullet$};\node at(0,0){$\circ$};\end{tikzpicture}$ is the diagrammatic representation of the map $\iota$, which takes elements of $W$ as inputs (in red) and gives elements of $V$ as outputs (in black).

In order to realise the above quantities as F-TR amplitudes coming from some translated initial data, we consider the two auxiliary tensors $G \in \Hom(\Sym{}{W},V)$ and $H \in \Hom(V \otimes \Sym{}{W},V)$ in place of the `translated' $F_{0,1}$ and $F_{0,2}$, respectively:
\begin{equation}\label{eq:G:H:translation}
\begin{aligned}
	& G \coloneqq \sum_{m \ge 2} \frac{1}{m!} \, F_{0,1+m} \circ \iota^{\otimes m} \,,
	\qquad\quad
	&& G^i = \sum_{m \geq 2} \frac{1}{m!} \, \indF{F}{0}{i}{j_1,\dots,j_m} \, \mathrm{e}^{j_1} \cdots \mathrm{e}^{j_m} \,,
	\\
	& H \coloneqq \sum_{m \ge 1} \frac{1}{m!} \, F_{0,1+(1+m)} \circ (\id_V \otimes \iota^{\otimes m}) \,,
	\qquad\quad
	&& H^i_j = \sum_{m \geq 1} \frac{1}{m!} \, \indF{F}{0}{i}{j,j_1,\dots,j_m} \, \mathrm{e}^{j_1} \cdots \mathrm{e}^{j_m} \,.
\end{aligned}
\end{equation}

\begin{thm}\label{thm:transl:FAiry}
	If $(F_{g,1+n})_{g,n}$ are the amplitudes of an F-Airy structure $(A,B,\conn{C},\disc{C},D)$ on $V$, then $(\tilde{F}_{g,1+n})_{g,n}$ defined in \eqref{eq:t:action} coincide with the amplitudes of the family of F-Airy structures given by
	\begingroup
	\allowdisplaybreaks
	\begin{align}
		\nonumber
		& \tilde{A}
		=
			B \circ (\iota \otimes \tilde{A})
			+
			B \circ (\id \otimes H)
			+
			B \circ (\id \otimes H) \circ \sigma_{1,2}
			+
			\disc{C} \circ ( G \otimes \tilde{A} + H^{\otimes 2})
			\, , \\
		\nonumber
		& \tilde{B}
		=
			B \circ (\iota \otimes \tilde{B})
			+
			\disc{C} \circ ( H \otimes \id )
			+
			\disc{C} \circ ( G \otimes \tilde{B} )
			\, , \\
		& \conn{\tilde{C}}
		=
			B \circ (\iota \otimes \conn{\tilde{C}})
			+
			\disc{C} \circ ( G \otimes \conn{\tilde{C}} )
			\, , \\
		\nonumber
		& \disc{\tilde{C}}
		=
			B \circ (\iota \otimes \disc{\tilde{C}})
			+
			\disc{C} \circ ( G \otimes \disc{\tilde{C}} )
			\, , \\
		\nonumber
		& \tilde{D}
		=
			 B \circ (\iota \otimes \tilde{D})
			+
			\id \otimes \tr(\conn{C} \circ H)
			+
			\disc{C} \circ ( G \otimes \tilde{D} )
			\, ,
	\end{align}
	\endgroup
	where:
	\begin{itemize}
		\item we have implicitly moved all the $W$ tensor factors to the right and wrote relations between maps having $\prod_{m \geq 0} W^{\otimes m}$ in the source domain, which pass to the quotient by the symmetric group action that was our $\Sym{}{W}$;
		\item $\sigma_{1,2} \colon V^{\otimes 2} \otimes \Sym{}{W} \rightarrow V^{\otimes 2} \otimes \Sym{}{W}$ is the permutation of the first two tensor factors.
		\end{itemize}
	As $G$ and $H$ start in degree $2$ and $1$ respectively, the above formulae are recursive in the homogeneous components of the tensors. This can be easily seen from the equivalent formulation in coordinates: setting $X[m]$ for the $m$-th homogeneous component of a tensor $X$,
	\begingroup
	\allowdisplaybreaks
	\begin{align}\label{eq:transl:tensors}
		\nonumber
		& \ind{\tilde{A}[m]}{i}{j,k}
		=
			\ind{B}{i}{a,b} \, \mathrm{e}^{a} \, \ind{\tilde{A}[m-1]}{b}{j,k}
			+
			\ind{B}{i}{j,a} \, H[m]^a_k + \ind{B}{i}{k,a} \, H[m]^a_j \\
		\nonumber
		& \qquad \qquad \qquad
			+
			\indC{\disc{C}}{i}{a,b} \left(
				\sum_{\substack{m_1+m_2 = m \\ m_1 \geq 2, m_2 \ge 0}}
					G[m_1]^a \, \ind{\tilde{A}[m_2]}{b}{j,k}
				+
				\sum_{\substack{m_1+m_2 = m \\ m_1, m_2 \geq 1}}
					H[m_1]^a_j \, H[m_2]^b_k
				\right)
			\,, \\
		\nonumber
		& \ind{\tilde{B}[m]}{i}{j,k}
		=
			\ind{B}{i}{a,b} \, \mathrm{e}^{a} \, \ind{\tilde{B}[m-1]}{b}{j,k}
			+
			\indC{\disc{C}}{i}{a,k} \, H[m]^a_j
			+
			\indC{\disc{C}}{i}{a,b} \sum_{\substack{m_1+m_2 = m \\ m_1 \geq 2, m_2 \ge 0}}
				G[m_1]^a \, \ind{\tilde{B}[m_2]}{b}{j,k}
			\,, \\
		& \indC{\conn{\tilde{C}}[m]}{i,j}{k}
		=
			\ind{B}{i}{a,b} \, \mathrm{e}^{a} \, \indC{\conn{\tilde{C}}[m-1]}{b,j}{k}
			+
			\indC{\disc{C}}{i}{a,b} \sum_{\substack{m_1+m_2 = m \\ m_1 \geq 2, m_2 \ge 0}}
				G[m_1]^a \, \indC{\conn{\tilde{C}}[m_2]}{b,j}{k}
			\,, \\
		\nonumber
		& \indC{\disc{\tilde{C}}[m]}{i}{j,k}
		=
			\ind{B}{i}{a,b} \, \mathrm{e}^{a} \, \indC{\disc{\tilde{C}}[m-1]}{b}{j,k}
			+
			\indC{\disc{C}}{i}{a,b} \sum_{\substack{m_1+m_2 = m \\ m_1 \geq 2, m_2 \ge 0}}
				G[m_1]^a \, \indC{\disc{\tilde{C}}[m_2]}{b}{j,k}
			\,, \\
		\nonumber
		& \tilde{D}[m]^i
		=
			\ind{B}{i}{a,b} \, \mathrm{e}^{a} \, \tilde{D}[m-1]^b
			+
			\indC{\conn{C}}{i,k}{a} \, H[m]^a_k
			+
			\indC{\disc{C}}{i}{a,b} \sum_{\substack{m_1+m_2 = m \\ m_1 \geq 2, m_2 \ge 0}}
				G[m_1]^a \, \tilde{D}[m_2]^b \,,
	\end{align}
	\endgroup
	together with the initial conditions $\tilde{X}[0] = X$ for all $X \in \{ A,B,\conn{C},\disc{C},D \}$. For diagrammatic supporters, the modified tensors are pictured as follows (we omit the degree dependence from the notation, which can be recovered by `counting' the number of red legs).
	\begingroup
	\allowdisplaybreaks
	\begin{gather}\label{eq:t:data:dgrms}
		\nonumber
		\begin{tikzpicture}[baseline,scale=.4]
			% Atilde
			\draw (0,-1) -- (0,-2);
			\node at (0,-2) [left] {\small $i$};
			\draw (-.8,1) -- (-.8,2);
			\node at (-.9,2.5) {\small $j$};
			\draw (-.4,1) -- (-.4,2);
			\node at (-.3,2.5) {\small $\vphantom{j}k$};
			\draw[Red] (0,1) -- (0,2);
			\draw[Red] (.8,1) -- (.8,2);
			\node at (.2,1.5) {\tiny$\cdot$};
			\node at (.4,1.5) {\tiny$\cdot$};
			\node at (.6,1.5) {\tiny$\cdot$};
			\draw (-1,-1) -- (1,-1) -- (1,1) -- (-1,1) -- cycle;
			\node at (0,0) {\small $\tilde{A}$};
			\node at (2,0) {$=$};
			\begin{scope}[xshift=4cm,yshift=-1cm]
				% B
				\draw (-1,-1) -- (1,-1) -- (1,1) -- (-1,1) -- cycle;
				\node at (0,0) {\small $B$};
				\draw (0,-1) -- (0,-2);
				\node at (0,-2) [left] {\small $i$};
				\draw (-.7,1) -- (-.7,2);
				\draw[Red] (-.7,2) -- (-.7,3);
				\node[white] at (-.7,2) {\small$\bullet$};
				\node at (-.7,2) {\small$\circ$};
				%
				% Atilde
				\draw (.2,3.5) -- (.2,4);
				\node at (.1,4.5) {\small $j$};
				\draw (.6,3.5) -- (.6,4);
				\node at (.7,4.5) {\small $\vphantom{j}k$};
				\draw[Red] (1,3.5) -- (1,4);
				\draw[Red] (1.8,3.5) -- (1.8,4);
				\node at (1.2,3.75) {\tiny$\cdot$};
				\node at (1.4,3.75) {\tiny$\cdot$};
				\node at (1.6,3.75) {\tiny$\cdot$};
				\draw (.7,1) -- (.7,1.5);
				\draw (0,1.5) -- (2,1.5) -- (2,3.5) -- (0,3.5) -- cycle;
				\node at (1,2.5) {\small $\tilde{A}$};
				\node at (2.75,1) {$+$};
			\end{scope}
			\begin{scope}[xshift=9cm,yshift=-.5cm]
				% B
				\draw (-1,-1) -- (1,-1) -- (1,1) -- (-1,1) -- cycle;
				\node at (0,0) {\small $B$};
				\draw (0,-1) -- (0,-2);
				\node at (0,-2) [left] {\small $i$};
				\draw (-.5,1) -- (-.5,3.2);
				\node at (-.5,3.7) {\small $j$};
				%
				% H
				\draw (.2,2.5) -- (.2,3.2);
				\node at (.2,3.7) {\small $\vphantom{j}k$};
				\draw[Red] (.4,2.5) -- (.4,3.2);
				\draw[Red] (.8,2.5) -- (.8,3.2);
				\node at (.525,2.8) {\tiny$\cdot$};
				\node at (.675,2.8) {\tiny$\cdot$};
				\draw (.5,1) -- (.5,1.5);
				\draw (0,1.5) -- (1,1.5) -- (1,2.5) -- (0,2.5) -- cycle;
				\node at (.5,2) {\small $H$};
				\node at (2,.5) {$+$};
			\end{scope}
			\begin{scope}[xshift=13cm,yshift=-.5cm]
				% B
				\draw (-1,-1) -- (1,-1) -- (1,1) -- (-1,1) -- cycle;
				\node at (0,0) {\small $B$};
				\draw (0,-1) -- (0,-2);
				\node at (0,-2) [left] {\small $i$};
				\draw (-.5,1) -- (-.5,3.2);
				\node at (-.5,3.7) {\small $\vphantom{j}k$};
				%
				% H
				\draw (.2,2.5) -- (.2,3.2);
				\node at (.2,3.7) {\small $j$};
				\draw[Red] (.4,2.5) -- (.4,3.2);
				\draw[Red] (.8,2.5) -- (.8,3.2);
				\node at (.525,2.8) {\tiny$\cdot$};
				\node at (.675,2.8) {\tiny$\cdot$};
				\draw (.5,1) -- (.5,1.5);
				\draw (0,1.5) -- (1,1.5) -- (1,2.5) -- (0,2.5) -- cycle;
				\node at (.5,2) {\small $H$};
				\node at (2,.5) {$+$};
			\end{scope}
			\begin{scope}[xshift=17cm,yshift=-1cm]
				% C disc
				\draw (-1,-1) -- (1,-1) -- (1,1) -- (-1,1) -- cycle;
				\node at (0,0) {\small $\disc{C}$};
				\draw (0,-1) -- (0,-2);
				\node at (0,-2) [left] {\small $i$};
				%
				% G
				\draw (.2,3.5) -- (.2,4);
				\node at (.1,4.5) {\small $j$};
				\draw[Red] (-1,3) -- (-1,4);
				\draw[Red] (-.4,3) -- (-.4,4);
				\node at (-.6,3.5) {\tiny$\cdot$};
				\node at (-.8,3.5) {\tiny$\cdot$};
				\draw (-.7,1) -- (-.7,2);
				\draw (-1.2,2) -- (-.2,2) -- (-.2,3) -- (-1.2,3) -- cycle;
				\node at (-.7,2.5) {\small $G$};
				%
				% Atilde
				\draw (.6,3.5) -- (.6,4);
				\node at (.7,4.5) {\small $\vphantom{j}k$};
				\draw[Red] (1,3.5) -- (1,4);
				\draw[Red] (1.8,3.5) -- (1.8,4);
				\node at (1.2,3.75) {\tiny$\cdot$};
				\node at (1.4,3.75) {\tiny$\cdot$};
				\node at (1.6,3.75) {\tiny$\cdot$};
				\draw (.7,1) -- (.7,1.5);
				\draw (0,1.5) -- (2,1.5) -- (2,3.5) -- (0,3.5) -- cycle;
				\node at (1,2.5) {\small $\tilde{A}$};
				\node at (2.75,1) {$+$};
			\end{scope}
			\begin{scope}[xshift=22cm,yshift=-.5cm]
				% C disc
				\draw (-1,-1) -- (1,-1) -- (1,1) -- (-1,1) -- cycle;
				\node at (0,0) {\small $\disc{C}$};
				\draw (0,-1) -- (0,-2);
				\node at (0,-2) [left] {\small $i$};
				%
				% H left
				\draw[Red] (-.7,2.5) -- (-.7,3.2);
				\draw[Red] (-.3,2.5) -- (-.3,3.2);
				\node at (-.575,2.8) {\tiny$\cdot$};
				\node at (-.425,2.8) {\tiny$\cdot$};
				\draw (-.6,1) -- (-.6,1.5);
				\draw (-1.1,1.5) -- (-.1,1.5) -- (-.1,2.5) -- (-1.1,2.5) -- cycle;
				\draw (-.9,2.5) -- (-.9,3.2);
				\node at (-.6,2) {\small $H$};
				\node at (-.9,3.7) {\small $j$};
				%
				% H right
				\draw[Red] (.5,2.5) -- (.5,3.2);
				\draw[Red] (.9,2.5) -- (.9,3.2);
				\node at (.625,2.8) {\tiny$\cdot$};
				\node at (.775,2.8) {\tiny$\cdot$};
				\draw (.6,1) -- (.6,1.5);
				\draw (.1,1.5) -- (1.1,1.5) -- (1.1,2.5) -- (.1,2.5) -- cycle;
				\draw (.3,2.5) -- (.3,3.2);
				\node at (.6,2) {\small $H$};
				\node at (.3,3.7) {\small $\vphantom{j}k$};
			\end{scope}
		\end{tikzpicture}
		\\
		% 2nd row
		\begin{tikzpicture}[baseline,scale=.4]
			% Btilde
			\draw (0,-1) -- (0,-2);
			\node at (0,-2) [left] {\small $i$};
			\draw (-.8,1) -- (-.8,2);
			\node at (-.9,2.5) {\small $j$};
			\draw (-.4,1) -- (-.4,2);
			\node at (-.3,2.5) {\small $\vphantom{j}k$};
			\draw[Red] (0,1) -- (0,2);
			\draw[Red] (.8,1) -- (.8,2);
			\node at (.2,1.5) {\tiny$\cdot$};
			\node at (.4,1.5) {\tiny$\cdot$};
			\node at (.6,1.5) {\tiny$\cdot$};
			\draw (-1,-1) -- (1,-1) -- (1,1) -- (-1,1) -- cycle;
			\node at (0,0) {\small $\tilde{B}$};
			\node at (2,0) {$=$};
			\begin{scope}[xshift=4cm,yshift=-1cm]
				% B
				\draw (-1,-1) -- (1,-1) -- (1,1) -- (-1,1) -- cycle;
				\node at (0,0) {\small $B$};
				\draw (0,-1) -- (0,-2);
				\node at (0,-2) [left] {\small $i$};
				\draw (-.7,1) -- (-.7,2);
				\draw[Red] (-.7,2) -- (-.7,3);
				\node[white] at (-.7,2) {\small$\bullet$};
				\node at (-.7,2) {\small$\circ$};
				%
				% Btilde
				\draw (.2,3.5) -- (.2,4);
				\node at (.1,4.5) {\small $j$};
				\draw (.6,3.5) -- (.6,4);
				\node at (.7,4.5) {\small $\vphantom{j}k$};
				\draw[Red] (1,3.5) -- (1,4);
				\draw[Red] (1.8,3.5) -- (1.8,4);
				\node at (1.2,3.75) {\tiny$\cdot$};
				\node at (1.4,3.75) {\tiny$\cdot$};
				\node at (1.6,3.75) {\tiny$\cdot$};
				\draw (.7,1) -- (.7,1.5);
				\draw (0,1.5) -- (2,1.5) -- (2,3.5) -- (0,3.5) -- cycle;
				\node at (1,2.5) {\small $\tilde{B}$};
				\node at (2.75,1) {$+$};
			\end{scope}
			\begin{scope}[xshift=9cm,yshift=-.5cm]
				% C disc
				\draw (-1,-1) -- (1,-1) -- (1,1) -- (-1,1) -- cycle;
				\node at (0,0) {\small $\disc{C}$};
				\draw (0,-1) -- (0,-2);
				\node at (0,-2) [left] {\small $i$};
				\draw[Red] (-.6,2.5) -- (-.6,3.2);
				\draw[Red] (-.2,2.5) -- (-.2,3.2);
				\node at (-.475,2.8) {\tiny$\cdot$};
				\node at (-.325,2.8) {\tiny$\cdot$};
				%
				% H
				\draw (-.5,1) -- (-.5,1.5);
				\draw (-1,1.5) -- (0,1.5) -- (0,2.5) -- (-1,2.5) -- cycle;
				\draw (-.8,2.5) -- (-.8,3.2);
				\node at (-.5,2) {\small $H$};
				\node at (-.8,3.7) {\small $j$};
				\draw (.5,1) -- (.5,3.2);
				\node at (.5,3.7) {\small $\vphantom{j}k$};
				\node at (2,.5) {$+$};
			\end{scope}
			\begin{scope}[xshift=13cm,yshift=-1cm]
				% C disc
				\draw (-1,-1) -- (1,-1) -- (1,1) -- (-1,1) -- cycle;
				\node at (0,0) {\small $\disc{C}$};
				\draw (0,-1) -- (0,-2);
				\node at (0,-2) [left] {\small $i$};
				%
				% G
				\draw (.2,3.5) -- (.2,4);
				\node at (.1,4.5) {\small $j$};
				\draw[Red] (-1,3) -- (-1,4);
				\draw[Red] (-.4,3) -- (-.4,4);
				\node at (-.6,3.5) {\tiny$\cdot$};
				\node at (-.8,3.5) {\tiny$\cdot$};
				\draw (-.7,1) -- (-.7,2);
				\draw (-1.2,2) -- (-.2,2) -- (-.2,3) -- (-1.2,3) -- cycle;
				\node at (-.7,2.5) {\small $G$};
				%
				% Btilde
				\draw (.6,3.5) -- (.6,4);
				\node at (.7,4.5) {\small $\vphantom{j}k$};
				\draw[Red] (1,3.5) -- (1,4);
				\draw[Red] (1.8,3.5) -- (1.8,4);
				\node at (1.2,3.75) {\tiny$\cdot$};
				\node at (1.4,3.75) {\tiny$\cdot$};
				\node at (1.6,3.75) {\tiny$\cdot$};
				\draw (.7,1) -- (.7,1.5);
				\draw (0,1.5) -- (2,1.5) -- (2,3.5) -- (0,3.5) -- cycle;
				\node at (1,2.5) {\small $\tilde{B}$};
			\end{scope}
			\begin{scope}[xshift=20cm]
				% Ctilde conn
				\draw[Red] (.7,1) -- (.7,2);
				\draw[Red] (-.3,1) -- (-.3,2);
				\node at (0,1.5) {\tiny$\cdot$};
				\node at (.2,1.5) {\tiny$\cdot$};
				\node at (.4,1.5) {\tiny$\cdot$};
				\draw (-1,-1) -- (1,-1) -- (1,1) -- (-1,1) -- cycle;
				\node at (0,0) {\small $\conn{\tilde{C}}$};
				\draw (-.5,-1) -- (-.5,-2);
				\node at (-.5,-2) [left] {\small $\vphantom{j}i$};
				\draw (.5,-1) -- (.5,-2);
				\node at (.5,-2) [right] {\small $j$};
				\draw (-.7,1) -- (-.7,2);
				\node at (-.7,2.5) {\small $k$};
				\node at (2,0) {$=$};
				\begin{scope}[xshift=4cm,yshift=-1cm]
					% B
					\draw (-1,-1) -- (1,-1) -- (1,1) -- (-1,1) -- cycle;
					\node at (0,0) {\small $B$};
					\draw (0,-1) -- (0,-2);
					\node at (0,-2) [left] {\small $\vphantom{j}i$};
					\draw (-.7,1) -- (-.7,2);
					\draw[Red] (-.7,2) -- (-.7,3);
					\node[white] at (-.7,2) {\small$\bullet$};
					\node at (-.7,2) {\small$\circ$};
					%
					% Ctilde conn
					\draw[Red] (1.7,3.5) -- (1.7,4);
					\draw[Red] (.7,3.5) -- (.7,4);
					\node at (1,3.75) {\tiny$\cdot$};
					\node at (1.2,3.75) {\tiny$\cdot$};
					\node at (1.4,3.75) {\tiny$\cdot$};
					\draw (.5,1) -- (.5,1.5);
					\draw (0,1.5) -- (2,1.5) -- (2,3.5) -- (0,3.5) -- cycle;
					\node at (1,2.5) {\small $\conn{\tilde{C}}$};
					\draw (1.5,1.5) -- (1.5,-2);
					\node at (1.5,-2) [right] {\small $j$};
					\draw (.3,3.5) -- (.3,4);
					\node at (.3,4.5) {\small $k$};
					\node at (2.75,1) {$+$};
				\end{scope}
				\begin{scope}[xshift=9cm,yshift=-1cm]
					% C disc
					\draw (-1,-1) -- (1,-1) -- (1,1) -- (-1,1) -- cycle;
					\node at (0,0) {\small $\disc{C}$};
					\draw (0,-1) -- (0,-2);
					\node at (0,-2) [left] {\small $\vphantom{j}i$};
					%
					% G
					\draw[Red] (-1,3) -- (-1,4);
					\draw[Red] (-.4,3) -- (-.4,4);
					\node at (-.6,3.5) {\tiny$\cdot$};
					\node at (-.8,3.5) {\tiny$\cdot$};
					\draw (-.7,1) -- (-.7,2);
					\draw (-1.2,2) -- (-.2,2) -- (-.2,3) -- (-1.2,3) -- cycle;
					\node at (-.7,2.5) {\small $G$};
					%
					% Ctilde conn
					\draw[Red] (1.7,3.5) -- (1.7,4);
					\draw[Red] (.7,3.5) -- (.7,4);
					\node at (1,3.75) {\tiny$\cdot$};
					\node at (1.2,3.75) {\tiny$\cdot$};
					\node at (1.4,3.75) {\tiny$\cdot$};
					\draw (.5,1) -- (.5,1.5);
					\draw (0,1.5) -- (2,1.5) -- (2,3.5) -- (0,3.5) -- cycle;
					\node at (1,2.5) {\small $\conn{\tilde{C}}$};
					\draw (1.5,1.5) -- (1.5,-2);
					\node at (1.5,-2) [right] {\small $j$};
					\draw (.3,3.5) -- (.3,4);
					\node at (.3,4.5) {\small $k$};
				\end{scope}
			\end{scope}
		\end{tikzpicture} \\
		% 3nd row
		\nonumber
		\begin{tikzpicture}[baseline,scale=.4]
			% Ctilde disc
			\draw (0,-1) -- (0,-2);
			\node at (0,-2) [left] {\small $i$};
			\draw (-.8,1) -- (-.8,2);
			\node at (-.9,2.5) {\small $j$};
			\draw (-.4,1) -- (-.4,2);
			\node at (-.3,2.5) {\small $\vphantom{j}k$};
			\draw[Red] (0,1) -- (0,2);
			\draw[Red] (.8,1) -- (.8,2);
			\node at (.2,1.5) {\tiny$\cdot$};
			\node at (.4,1.5) {\tiny$\cdot$};
			\node at (.6,1.5) {\tiny$\cdot$};
			\draw (-1,-1) -- (1,-1) -- (1,1) -- (-1,1) -- cycle;
			\node at (0,0) {\small $\disc{\tilde{C}}$};
			\node at (2,0) {$=$};
			\begin{scope}[xshift=4cm,yshift=-1cm]
				% B
				\draw (-1,-1) -- (1,-1) -- (1,1) -- (-1,1) -- cycle;
				\node at (0,0) {\small $B$};
				\draw (0,-1) -- (0,-2);
				\node at (0,-2) [left] {\small $i$};
				\draw (-.7,1) -- (-.7,2);
				\draw[Red] (-.7,2) -- (-.7,3);
				\node[white] at (-.7,2) {\small$\bullet$};
				\node at (-.7,2) {\small$\circ$};
				%
				% Ctilde disc
				\draw (.2,3.5) -- (.2,4);
				\node at (.1,4.5) {\small $j$};
				\draw (.6,3.5) -- (.6,4);
				\node at (.7,4.5) {\small $\vphantom{j}k$};
				\draw[Red] (1,3.5) -- (1,4);
				\draw[Red] (1.8,3.5) -- (1.8,4);
				\node at (1.2,3.75) {\tiny$\cdot$};
				\node at (1.4,3.75) {\tiny$\cdot$};
				\node at (1.6,3.75) {\tiny$\cdot$};
				\draw (.7,1) -- (.7,1.5);
				\draw (0,1.5) -- (2,1.5) -- (2,3.5) -- (0,3.5) -- cycle;
				\node at (1,2.5) {\small $\disc{\tilde{C}}$};
				\node at (2.75,1) {$+$};
			\end{scope}
			\begin{scope}[xshift=9cm,yshift=-1cm]
				% C disc
				\draw (-1,-1) -- (1,-1) -- (1,1) -- (-1,1) -- cycle;
				\node at (0,0) {\small $\disc{C}$};
				\draw (0,-1) -- (0,-2);
				\node at (0,-2) [left] {\small $i$};
				%
				% G
				\draw (.2,3.5) -- (.2,4);
				\node at (.1,4.5) {\small $j$};
				\draw[Red] (-1,3) -- (-1,4);
				\draw[Red] (-.4,3) -- (-.4,4);
				\node at (-.6,3.5) {\tiny$\cdot$};
				\node at (-.8,3.5) {\tiny$\cdot$};
				\draw (-.7,1) -- (-.7,2);
				\draw (-1.2,2) -- (-.2,2) -- (-.2,3) -- (-1.2,3) -- cycle;
				\node at (-.7,2.5) {\small $G$};
				%
				% Ctilde disc
				\draw (.6,3.5) -- (.6,4);
				\node at (.7,4.5) {\small $\vphantom{j}k$};
				\draw[Red] (1,3.5) -- (1,4);
				\draw[Red] (1.8,3.5) -- (1.8,4);
				\node at (1.2,3.75) {\tiny$\cdot$};
				\node at (1.4,3.75) {\tiny$\cdot$};
				\node at (1.6,3.75) {\tiny$\cdot$};
				\draw (.7,1) -- (.7,1.5);
				\draw (0,1.5) -- (2,1.5) -- (2,3.5) -- (0,3.5) -- cycle;
				\node at (1,2.5) {\small $\disc{\tilde{C}}$};
			\end{scope}
			\begin{scope}[xshift=17cm]
				% Dtilde
				\draw[Red] (-.6,1) -- (-.6,2);
				\draw[Red] (.6,1) -- (.6,2);
				\node at (-.3,1.5) {\tiny$\cdot$};
				\node at (0,1.5) {\tiny$\cdot$};
				\node at (.3,1.5) {\tiny$\cdot$};
				\draw (-1,-1) -- (1,-1) -- (1,1) -- (-1,1) -- cycle;
				\node at (0,0) {\small $\tilde{D}$};
				\draw (0,-1) -- (0,-2);
				\node at (0,-2) [left] {\small $i$};
				\node at (2,0) {$=$};
				\begin{scope}[xshift=4cm,yshift=-1cm]
					% B
					\draw (-1,-1) -- (1,-1) -- (1,1) -- (-1,1) -- cycle;
					\node at (0,0) {\small $B$};
					\draw (0,-1) -- (0,-2);
					\node at (0,-2) [left] {\small $i$};
					\draw (-.7,1) -- (-.7,2);
					\draw[Red] (-.7,2) -- (-.7,3);
					\node[white] at (-.7,2) {\small$\bullet$};
					\node at (-.7,2) {\small$\circ$};
					%
					% Dtilde
					\draw[Red] (.3,3.5) -- (.3,4);
					\draw[Red] (1.7,3.5) -- (1.7,4);
					\node at (.7,3.75) {\tiny$\cdot$};
					\node at (1,3.75) {\tiny$\cdot$};
					\node at (1.3,3.75) {\tiny$\cdot$};
					\draw (.5,1) -- (.5,1.5);
					\draw (0,1.5) -- (2,1.5) -- (2,3.5) -- (0,3.5) -- cycle;
					\node at (1,2.5) {\small $\tilde{D}$};
					\node at (2.75,1) {$+$};
				\end{scope}
				\begin{scope}[xshift=9cm,yshift=-1cm]
					% C conn
					\draw (-1,-1) -- (1,-1) -- (1,1) -- (-1,1) -- cycle;
					\node at (0,0) {\small $\conn{C}$};
					\draw (-.5,-1) -- (-.5,-2);
					\node at (-.5,-2) [left] {\small $i$};
					%
					% H
					\draw[Red] (-.1,2.5) -- (-.1,3);
					\draw[Red] (.3,2.5) -- (.3,3);
					\node at (.025,2.75) {\tiny$\cdot$};
					\node at (.175,2.75) {\tiny$\cdot$};
					\draw (0,1) -- (0,1.5);
					\draw (.5,1.5) -- (-.5,1.5) -- (-.5,2.5) -- (.5,2.5) -- cycle;
					\node at (0,2) {\small $H$};
					\draw (.5,-1) -- (.5,-1.2);
					\draw (-.3,2.5) -- (-.3,3);
					\draw (1.5,3) arc (0:180:.9);
					\draw (1.5,-1.2) arc (0:-180:.5);
					\draw (1.5,-1.2) -- (1.5,3);
					\node at (2.75,1) {$+$};
				\end{scope}
				\begin{scope}[xshift=14cm,yshift=-1cm]
					% C disc
					\draw (-1,-1) -- (1,-1) -- (1,1) -- (-1,1) -- cycle;
					\node at (0,0) {\small $\disc{C}$};
					\draw (0,-1) -- (0,-2);
					\node at (0,-2) [left] {\small $i$};
					%
					% G
					\draw[Red] (-1,3) -- (-1,4);
					\draw[Red] (-.4,3) -- (-.4,4);
					\node at (-.6,3.5) {\tiny$\cdot$};
					\node at (-.8,3.5) {\tiny$\cdot$};
					\draw (-.7,1) -- (-.7,2);
					\draw (-1.2,2) -- (-.2,2) -- (-.2,3) -- (-1.2,3) -- cycle;
					\node at (-.7,2.5) {\small $G$};
					%
					% Dtilde
					\draw[Red] (.3,3.5) -- (.3,4);
					\draw[Red] (1.7,3.5) -- (1.7,4);
					\node at (.7,3.75) {\tiny$\cdot$};
					\node at (1,3.75) {\tiny$\cdot$};
					\node at (1.3,3.75) {\tiny$\cdot$};
					\draw (.5,1) -- (.5,1.5);
					\draw (0,1.5) -- (2,1.5) -- (2,3.5) -- (0,3.5) -- cycle;
					\node at (1,2.5) {\small $\tilde{D}$};
				\end{scope}
			\end{scope}
		\end{tikzpicture}
	\end{gather}
	\endgroup

	Moreover, if $G$ has a non-zero radius of convergence, so do $H$, $\tilde{F}_{g,1+n}$ and all the modified initial data $(\tilde{A},\tilde{B},\conn{\tilde{C}},\disc{\tilde{C}},\tilde{D})$.
\end{thm}

\begin{proof}
	We first examine the $(0,3)$ case. By applying F-TR to the translated $(0,3)$ amplitude in \eqref{eq:translation:coords} we get
	\begin{equation}
	\begin{split}
		\indF{\tilde{F}}{0}{i_0}{i_1,i_2}
		& =
		\sum_{m \geq 0} \frac{1}{m!} \,
			\indF{F}{0}{i_0}{i_1, i_2, j_1, \dots, j_m}
			\mathrm{e}^{j_1} \cdots \mathrm{e}^{j_m} \\
		& =
		\indF{F}{0}{i_0}{i_1,i_2} + \sum_{m \geq 1} \frac{1}{m!}
		\Bigg(
			\ind{B}{i_0}{i_1, a} \,
			\indF{F}{0}{a}{i_2, j_1, \dots, j_m}
			+
			\ind{B}{i_0}{i_2, a} \,
			\indF{F}{0}{a}{i_1, j_1, \dots, j_m}
			+
			\sum_{l=1}^{m}
				\ind{B}{i_0}{j_l, a} \,
				\indF{F}{0}{a}{i_1, i_2, j_1, \dots \widehat{j_l} \dots, j_m}
			 \\	
		& \qquad\qquad\qquad\qquad\quad
			+
			\indC{\disc{C}}{i_0}{a,b}
			\sum_{J \sqcup J' = \{j_1, \dots, j_m\} }
				\bigl(
					\indF{F}{0}{a}{J} \, \indF{F}{0}{b}{i_1,i_2,J'}
					+
					\indF{F}{0}{a}{i_1,J} \, \indF{F}{0}{b}{i_2,J'}
				\bigr)
		\Bigg) \mathrm{e}^{j_1} \cdots \mathrm{e}^{j_m} \,.
	\end{split}
	\end{equation}
	Collecting monomials in the linear forms $\mathrm{e}^{j_1} \mathrm{e}^{j_2} \cdots$ while taking into account the symmetry factor, we see that the degree $m$ component of $\tilde{F}[m]_{0,3}$ equals
	\begin{multline}
			\indF{\tilde{F}[m]}{0}{i_0}{i_1,i_2}
			=
			\delta_{m,0} \; A^{i_0}_{i_1,i_2}
			+
			\Bigl(
				\ind{B}{i_0}{i_1, a} \,
				H[m]^a_{i_2}
				+
				\ind{B}{i_0}{i_2, a} \,
				H[m]^{a}_{i_1}
			\Bigr)
			+
			\ind{B}{i_0}{a, b} \, \mathrm{e}^{a} \, \ind{\tilde{A}[m-1]}{b}{i_1, i_2}	\\
			+
			\indC{\disc{C}}{i_0}{a,b} \Biggr(
			\sum_{\substack{m_1+m_2 = m \\ m_1 \geq 2, m_2 \ge 0}}
				G[m_1]^a \,
				\ind{\tilde{F}[m_2]}{b}{i_1,i_2}
			+
			\sum_{\substack{m_1+m_2=m \\ m_1, m_2 \geq 1}}
				H[m_1]^a_{i_1} \,
				H[m_2]^b_{i_2}
			\Biggl) \,,
	\end{multline} 
	with the convention $\tilde{X}[-1] = 0$. The above recursion on $m$ uniquely characterises the sequence of tensors $(\tilde{F}[m]_{0,3})_{m \geq 0}$. Moreover, the recursive definition of the homogeneous components of $\tilde{A}$ in \eqref{eq:transl:tensors} coincides with the recursion for $\tilde{F}[m]_{0,3}$, hence the equality. The computation for the $(1,1)$ case is similar and omitted: it gives the recursion for $\tilde{D}[m]$.
	
	Suppose now that F-TR holds for all $(g_0,n_0)$ such that $2g_0-2+(1+n_0) < 2g-2+(1+n)$. From the definition of translated amplitudes and F-TR for the original amplitudes, we find
	\begin{multline}
		\indF{\tilde{F}}{g}{i_0}{i_1, \dots, i_n}
		=
		\sum_{m \geq 0} \frac{1}{m!} \Bigg(
			\sum_{p=1}^{n}
				\ind{B}{i_0}{i_p, a} \,
				\indF{F}{g}{a}{i_1, \dots \widehat{i_p} \dots, i_n, j_1, \dots, j_m}
			+
			\sum_{q=1}^{m}
				\ind{B}{i_0}{j_q, a} \,
				\indF{F}{g}{a}{i_1, \dots, i_n, j_1, \dots \widehat{j_q} \dots, j_m}
		\\
		+
		\frac{1}{2}	\indC{\conn{C}}{i_0,b}{a} \,
			\indF{F}{g-1}{a}{i_1, \dots, i_n, j_{1}, \dots, j_{m},b}
			+
			\frac{1}{2} \indC{\disc{C}}{i_0}{a,b}
			\sum_{\substack{
				h+h' = g \\ N \sqcup N' = \{ i_1, \dots, i_n \} \\ M \sqcup M' = \{ j_1, \dots, j_m \}
			}}
				\indF{F}{h}{a}{N \sqcup M} \, \indF{F}{h'}{b}{N' \sqcup M'}
		\Bigg) \mathrm{e}^{j_1} \cdots \mathrm{e}^{j_m} \,.
	\end{multline}
	The second $B$-sum and the terms $(h,N) = (0,\varnothing)$ or $(h',N') = (0,\varnothing)$ in the $\disc{C}$-sum involve $\tilde{F}_{g,1 + n}$, and we move these contributions to the left-hand side while exploiting the symmetry of $\disc{C}$ in its two lower indices. In all the other terms, redistributing the vectors $\mathrm{e}^j$ we recognise some $\tilde{F}_{g_0,1+n_0}$ with $2g_0 - 2 + (1 + n_0) < 2g - 2 + (1 + n)$. We thus arrive to
	\begin{multline}
		\big(
			\delta^{i_0}_a
			-
			\ind{B}{i_0}{j,a} \, \mathrm{e}^j
			-
			\indC{\disc{C}}{i_0}{j,a}
			G^j
		\big) \indF{\tilde{F}}{g}{a}{i_1, \dots, i_n}
		=
		\sum_{l=1}^{n}
			\ind{B}{i_0}{i_l,a} \,
			\indF{\tilde{F}}{g}{a}{i_1, \dots \widehat{i_l} \dots, i_n}
		+ 
		\frac{1}{2} \indC{\conn{C}}{i_0,b}{a} \,
		\indF{\tilde{F}}{g-1}{a}{i_1, \dots, i_n,b}
		\\
		+ \indC{\disc{C}}{i_0}{a,b} \Biggl(
			\sum_{l=1}^{n}
				H^a_{i_l} \,
				\indF{\tilde{F}}{g}{b}{i_1, \dots, \widehat{i_l}, \dots, i_n}
			+
			\frac{1}{2} \sum_{\substack{h+h' = g \\ J \sqcup J' = \{ i_1, \dots, i_n\} }}
				\indF{\tilde{F}}{h}{a}{J} \, \indF{\tilde{F}}{h'}{b}{J'}
			\Biggr) \,.
	\end{multline}
	Let us introduce the operator $K \in \End( V \otimes \CC\bbraket{(\mathrm{e}^i)_{i \in I}} )$, specified by its matrix of coefficients in the chosen basis as $K^i_a \coloneqq \delta^i_a - \ind{B}{i}{j,a}\mathrm{e}^j - \indC{\disc{C}}{i}{j,a} G^j$. Since $G^j$ contains only terms of positive degree, $K = \id_V + \bigO(\sum_{i \in I} \mathrm{e}^i)$ is invertible. We can then rewrite
	\begin{equation}
	\begin{split}
		\indF{\tilde{F}}{g}{i_0}{i_1,\dots,i_n}
		=
		\sum_{l = 1}^{n} \tilde{B}^{i_0}_{i_l,a} \indF{\tilde{F}}{g}{a}{i_1,\dots,\widehat{i_l},\dots,i_n}
		+
		\frac{1}{2} \indC{\conn{\tilde{C}}}{i_0,b}{a} \, \indF{\tilde{F}}{g-1}{a}{i_1,\dots,i_n,b}
		+
		\frac{1}{2} \indC{\disc{\tilde{C}}}{i_0}{a,b} \sum_{\substack{h + h' = g \\ J \sqcup J' = \{i_1,\dots,i_n\}}}
			\indF{\tilde{F}}{h}{a}{J}
			\indF{\tilde{F}}{h'}{b}{J'} \,,
	\end{split}
	\end{equation}
	where we have set
	\begin{equation}\label{eq:BC:tilde}
		\ind{\tilde{B}}{i}{j,k} = (K^{-1})^i_a \big( \ind{B}{a}{j,k} + \indC{\disc{C}}{a}{b,k} H^{b}_{j} \big) \,,
		\qquad
		\indC{\conn{\tilde{C}}}{i,j}{k} = (K^{-1})^i_a \indC{\conn{C}}{a,j}{k} \,,
		\qquad 
		\indC{\disc{\tilde{C}}}{i}{j,k} = (K^{-1})^i_a \indC{\disc{C}}{a}{j,k} \,.
	\end{equation}
	The recursive definition of $(\tilde{B}[m])_{m \geq 0}$ in \eqref{eq:transl:tensors} is tailored such that $K^{i}_{a} \ind{\tilde{B}}{a}{j,k} = \ind{B}{i}{j,k} + \indC{\disc{C}}{i}{a,k} H_j^a$, so it matches with the homogeneous components of $\tilde{B}$ defined by \eqref{eq:BC:tilde}. Similarly the recursive definitions of $(\conn{\tilde{C}}[m])_{m \geq 0}$ and $(\disc{\tilde{C}}[m])_{m \geq 0}$ in \eqref{eq:transl:tensors} match the homogeneous components of the tensors defined by \eqref{eq:BC:tilde}. This concludes the proof by induction.
		
	As for the convergence statement, when $G^i$ has a non-zero radius of convergence $r > 0$ for all $i \in I$, then $H^i_j = \partial_{\mathrm{e}^j} G^i$ has at least radius of convergence $r$ for all $i,j \in I$. Thus, the operator $K^{-1}$ also has a non-zero radius of convergence, say $r' > 0$. Therefore $\tilde{A}$, $\tilde{B}$, $\conn{\tilde{C}}$, $\disc{\tilde{C}}$, $\tilde{D}$ all have a radius of convergence larger or equal to $r'' = \min(r,r') > 0$. From the F-TR relation we deduce that $\tilde{F}_{g,1+n}$ has a radius of convergence larger or equal to $r''$.
\end{proof}

\begin{rem}\label{rem:translation}
	For a given vector $\tau \in V$, we can specialise the translated amplitudes by taking $W = V$, $\iota = \id_V$, and considering the partial function
	\begin{equation}
		\ev_{\tau} \colon \Sym{}{(V^{\ast})} \longrightarrow \CC \,,
		\qquad
		\sum_{m\ge 0} \lambda_{m,j_{1}} \cdots \lambda_{m,j_{m}}
		\longmapsto
		\sum_{m\ge 0} \braket{\lambda_{m,j_1}, \tau} \cdots \braket{\lambda_{m,j_m}, \tau} \,,
	\end{equation}
	where the value is defined if and only if the sum is absolutely convergent for some norm on $V$. In coordinates, setting $\tau = \tau^i \mathrm{e}_i$, we have
	\begin{equation}
		\ev_{\tau}
		\colon
		\sum_{m\ge 0} \mathrm{e}^{j_1} \cdots \mathrm{e}^{j_m}
		\longmapsto
		\sum_{m\ge 0} \tau^{j_1} \cdots \tau^{j_m} \,.
	\end{equation}
	If the evaluation $\preind{\tau}{G} \coloneqq \ev_{\tau} \circ G$ is convergent, so do $\preind{\tau}{H} \coloneqq \ev_{\tau} \circ H$ and the translated amplitudes
	\begin{equation}\label{eq:t:spec}
		\preind{\tau}{F}_{g,1+n}
		\coloneqq
		\ev_{\tau} \circ \tilde{F}_{g,1+n}
		=
		\sum_{m \ge 0}
			\frac{1}{m!} \,
			F_{g,1+n+m}\bigl( \id_{V}^{\otimes n} \otimes \tau^{\otimes m} \bigr) \,.
	\end{equation}
	In this case, the vector potential associated with the translated amplitudes is given by
	\begin{equation}\label{eq:Phi:trans}
		\preind{\tau}{\Phi}(x)
		=
		\Phi(x + \tau)
		-
		\hbar^{-1} \big( \preind{\tau}{G} + \preind{\tau}{H}(x) \big) \,.
	\end{equation}
	This formalises the heuristic argument presented at the beginning of the section.
\end{rem}

%–––––––––––––––––––––––––––––––––––––––––––%
\section{F-cohomological field theories}
\label{sec:FCohFTs}
%–––––––––––––––––––––––––––––––––––––––––––%
In this section we recall the definition of F-cohomological field theories, following \cite{BR21,ABLR23}, and study their symmetries. We work in cohomology with coefficients in $\CC$. There are obvious variants in cohomology with rational coefficients or in Chow.

%–––––––––––––––––––––––––––––––––––––––––––%
\subsection{F-cohomological field theories}
%–––––––––––––––––––––––––––––––––––––––––––%
Let $\Mbar_{g,1+n}$ be the Deligne--Mumford moduli space of stable curves of genus $g$ with $(1 + n)$ marked points labelled as $0, 1, \dots, n$. Given a splitting $h + h' = g$ of the genus and a splitting $J \sqcup J' = [n]$ of the marked points, we consider the glueing morphism of \emph{separating kind}:
\begin{equation}\label{eq:glueing:sep}
	\gl \colon \Mbar_{h,1+(1+|J|)} \times \Mbar_{h',1+|J'|} \longrightarrow \Mbar_{g,1 + n} \,.
\end{equation}
This morphism consists in glueing the first node from the left factor to the first node from the right factor, thus creating a stable curve. We also consider the morphism forgetting the last marked point and (if necessary) stabilising, i.e. contracting to a point the unstable components of the normalisation:
\begin{equation}\label{eq:forgetful}
	\fg \colon \Mbar_{g,1+(n+1)} \longrightarrow \Mbar_{g,1+n} \,.
\end{equation}

\begin{defn}
	An \emph{F-cohomological field theory} (F-CohFT for short) is the data of a vector space $V_0$, called the phase space, together with a collection of linear maps
	\begin{equation}
		\Omega_{g,1 + n} \colon V_0^{\otimes n} \longrightarrow H^{\rm even}(\Mbar_{g,1+n}) \otimes V_0
	\end{equation}
	indexed by integers $g,n \geq 0$ such that $2g - 2 + (1 + n) > 0$ and satisfying the following axioms:
	\begin{itemize}
		\item $\Omega_{g,1+n}$ is equivariant for the action of the symmetric group $\Sy_n$ permuting simultaneously the tensor factors of $V_0^{\otimes n}$ and the last $n$ marked points in $\Mbar_{g,1+n}$;

		\item whenever $g = h + h'$ and $J \sqcup J' = [n] \,$, pulling back by the corresponding morphism of separating kind yields
		\begin{equation}\label{eq:glueing:axiom}
			\gl^*\Omega_{g,1 + n}(v_1 \otimes\cdots\otimes v_n)
			=
			\Omega_{h,1+(1+|J|)}\big( \Omega_{h',1+|J'|}(v_{J'}) \otimes v_J \big) \,.
		\end{equation}
	\end{itemize}
	Furthermore, we say that the F-CohFT has a flat unit if we are provided with a distinguished element $\mathrm{e} \in V_0$ such that
	\begin{equation}
		\fg^* \Omega_{g,1+n}(v_1 \otimes\cdots\otimes v_n)
		=
		\Omega_{g,1+(n+1)}(v_1 \otimes\cdots\otimes v_n \otimes \mathrm{e})
		\qquad \text{and} \qquad
		\Omega_{0,3}(v \otimes \mathrm{e}) = v \,.
	\end{equation}
\end{defn}

\begin{ex}
	Every CohFT (see e.g. \cite{Pan19} for the definition) is an F-CohFT, since the former requires the additional data of a non-degenerate pairing on $V_0$ and compatibility with respect to the glueing morphism of non-separating kind. A less trivial class of examples are the ones constructed from the top Chern class of the Hodge bundle: $\lambda_g = c_g(\mathbb{E}) \in H^{2g}(\Mbar_{g, 1+n})$, where $\mathbb{E}$ is the vector bundle whose fibre over a smooth point $[C,p_0,\dots,p_n]$ is the space of holomorphic differentials on $C$. Given a CohFT $(\Omega_{g,n})_{g,n}$ on $V_0$, the collection of linear maps $(\lambda_g \cdot \Omega_{g,1+n})_{g,n}$ forms an F-CohFT after the appropriate identification of $V_0$ and $V_0^*$ through the given pairing. In particular, the class $\lambda_g$ itself is a prime example of an F-CohFT.
\end{ex}

As in the usual context, the cohomological degree-zero part of a given F-CohFT, that is
\begin{equation}
	\Omega^{0}_{g,1+n}
	\coloneqq
	\left[ \deg = 0 \right] \Omega_{g,1+n} \colon V_0^{\otimes n}
	\longrightarrow
	H^{0}(\Mbar_{g, 1+n}) \otimes V_0
	\cong V_0 \,,
\end{equation}
is uniquely characterised by a corresponding algebraic structure, that of an F-TFT (recall Definition~\ref{defn:F-TFT}). More precisely, the commutative and associative product on $V_0$ is given by $\Omega^{0}_{0,3}$, and the distinguished element is $w = \Omega^{0}_{1,1}$. From these data, the maps $\Omega^{0}_{g,1+n}$ are given by
\begin{equation}
	\Omega^{0}_{g,1+n}( v_1 \otimes \cdots \otimes v_n)
	=
	v_1 \bcdots v_n \bcdot w^g \,,
\end{equation}
i.e. they coincide with the F-TFT amplitudes defined in equation~\eqref{eq:FTFT}.

\begin{rem}
	There are three small differences in our definition compared to \cite{ABLR23}: they include the existence of a flat unit in the definition of an F-CohFT, we do not; they consider F-CohFTs as maps of the form $V_0^* \otimes V_0^{\otimes n} \rightarrow H^{\textup{even}}(\Mbar_{g,1+n})$, while we have moved the $V_0^*$ to the right by duality resulting in a slightly different (but equivalent) form for the axioms; they label marked points as $1,\dots, n+1$, while we labelled them $0,\dots,n$. 
\end{rem}

%–––––––––––––––––––––––––––––––––––––––––––%
\subsection{Known symmetries of F-CohFTs}
\label{subsec:F-Giv}
%–––––––––––––––––––––––––––––––––––––––––––%
Inspired by the Givental group action on CohFTs, in \cite{ABLR23} the authors describe how to act on F-CohFTs by means of changes of basis, R-actions, and translations. In this section, we collect the definition of such actions. Before proceeding, recall the definition of $\psi$-classes: for a given $i \in \set{0,1,\dots,n}$, set $\psi_i = c_1(\mathbb{L}_i) \in H^2(\Mbar_{g,1+n})$, where $\mathbb{L}_i$ is the line bundle whose fibre over a point $[C,p_0,\dots,p_n]$ is the cotangent line $T^*_{p_i}C$. 

%–––––––––––––––––––––––––––––––––––––––––––%
\paragraph{Change of basis.}
%–––––––––––––––––––––––––––––––––––––––––––%
Given $L \in \GL(V_0)$, define
\begin{equation}
	(\hat{L}\Omega)_{g,1+n}
	\coloneqq
	L \circ \Omega_{g,1+n} \circ (L^{-1})^{\otimes n} \,.
\end{equation}
The resulting collection of maps forms an F-CohFT, defining a left group action of $\GL(V_0)$.

%–––––––––––––––––––––––––––––––––––––––––––%
\paragraph{R-action.}
%–––––––––––––––––––––––––––––––––––––––––––%
Recall that boundary strata of $\Mbar_{g,1+n}$ are described by stable graphs (see for instance \cite{PPZ15}). Among all stable curves, those whose Jacobian variety is compact are called of \textit{compact type}. Boundary strata parametrising curves of compact type are in one-to-one correspondence with the stable graphs that only have separating edges. Such boundary strata are described precisely by the set of stable trees (Definition~\ref{def:stbl:tree}). For a given stable tree $\bm{T} \in \mathbb{T}_{g,1+n}$, the associated closed boundary stratum is $\Mbar_{\bm{T}} = \prod_{v \in \Vert(\bm{T})} \Mbar_{g(v),1 + n(v)}$, which comes with the inclusion map
\begin{equation}
	\xi_{\bm{T}} \colon \Mbar_{\bm{T}} \longrightarrow \Mbar_{g,1+n} \,.
\end{equation}
The glueing map of separating kind from \eqref{eq:glueing:sep} is an example of inclusion of boundary strata, corresponding to the stable tree with a single edge connecting two vertices of genera $h$ and $h'$ satisfying $g = h + h'$ and leaves labelled by $\set{0} \sqcup J$ and $J'$ respectively (satisfying $J \sqcup J' = [n]$):
\begin{equation}\label{eq:glueing:sep:tree}
\begin{tikzpicture}[baseline]
	\draw (0,0) -- (0,-1);
	\node at (0,-1) [right] {\footnotesize$0$};
	\draw (0,0) -- (145:1);
	\draw (0,0) -- (110:1);
	\draw [dotted, thick] (135:.8) -- (119:.8);
	\draw (0,0) -- (45:1);
	\draw (45:1) -- ($(45:1)+(110:1)$);
	\draw (45:1) -- ($(45:1)+(70:1)$);
	\draw [dotted, thick] ($(45:1) + (98:.85)$) -- ($(45:1) + (82:.85)$);

	\node at (127.5:1.2) {\footnotesize$J$};
	\node at ($(45:1) + (0,1.3)$) {\footnotesize$J'$};

	\draw[fill=white] (0,0) circle (.25cm);
	\node at (0,0) {\tiny$h$};

	\draw[fill=white] (45:1) circle (.25cm);
	\node at (45:1) {\tiny$h'$};
\end{tikzpicture}.
\end{equation}
Given $R(u) \in \mathfrak{Giv} \coloneqq \id_{V_0} + u \End(V_0)\bbraket{u}$, called the \emph{F-Givental group}, define
\begin{multline}\label{eq:Raction}
	(\hat{R}\Omega)_{g,1+n}
	\coloneqq
	\sum_{\bm{T} \in \mathbb{T}_{g,1+n}}
		\xi_{\bm{T},*} \Bigg[
			\Bigg(
				\bigotimes_{v \in \Vert(\bm{T})} \Omega_{g(v),1+n(v)}
			\Bigg)
		\\
			\underset{\bm{T}}{\circ}
			\Bigg(
				\bigotimes_{e \in \Edge(\bm{T})} \mathscr{E}_{R}(\psi_{e'},\psi_{e''})
			\Bigg)
			\underset{\bm{T}}{\circ}
			\Bigg(
				R(-\psi_0)
				\otimes
				\bigotimes_{i = 1}^{n} R^{-1}(\psi_i)
			\Bigg)
		\Bigg] \,.
\end{multline}
The operation $\circ_{\bm{T}}$ is the natural composition along edges and leaves of the rooted stable tree, precisely as in \eqref{eq:Bglbv:action}. The edge weight is defined as
\begin{equation}\label{eq:Bglbv:F-Giv}
	\mathscr{E}_{R}(u',u'')
	\coloneqq
	\frac{\id_{V_0} - R^{-1}(u') \circ R(-u'')}{u' + u''} \in \End(V_0)\bbraket{u',u''} \,.
\end{equation}
The inverse $R^{-1}$ is meant with respect to the product structure $\circ$ in $\End(V_0)\bbraket{u}$, that is composition on $\End(V_0)$ and multiplication on $\CC\bbraket{u}$. Since $R(u) = \id_{V_0} + \bigO(u)$, it is always invertible.

We remark that, contrary to the R-action on CohFTs, there is no symmetry factor in \eqref{eq:Raction} since stable trees do not have non-trivial automorphisms.

\begin{thm}[{R-action on F-CohFTs \cite{ABLR23}}]
	The collection of maps $\hat{R}\Omega$ forms an F-CohFT. The resulting action is a left group action of the F-Givental group $(\mathfrak{Giv},\circ)$.
\end{thm}

%–––––––––––––––––––––––––––––––––––––––––––%
\paragraph{Translation.}
%–––––––––––––––––––––––––––––––––––––––––––%
Given $T(u) \in u^2 V_0\bbraket{u}$, define
\begin{equation}\label{eq:translation:F-CohFT}
	(\hat{T}\Omega)_{g,1+n}
	\coloneqq
	\sum_{m \geq 0}
		\frac{1}{m!} \, 
		\fg_{m,*} \Big[
			\Omega_{g,1+n + m}
			\big(
				\id_{V_0^{\otimes n}} \otimes T(\psi_{n+1}) \otimes \dots \otimes T(\psi_{n+m})
			\big) \Big] \,,
\end{equation}
where $\fg_{m} \colon \Mbar_{g,1+n+m} \rightarrow \Mbar_{g,1+n}$ is the morphism forgetting the last $m$ marked points (and stabilising whenever necessary). For each stable $(g,1+n)$, the sum in equation~\eqref{eq:translation:F-CohFT} truncates to a finite sum as $T(u) = \bigO(u^2)$ and for cohomological degree reasons.

\begin{thm}[{Translation of F-CohFTs \cite{ABLR23}}] \label{thm:translation:F-CohFT}
	The collection of maps $\hat{T} \Omega$ forms an F-CohFT. The resulting action is an abelian group action of $(u^2 V_0\bbraket{u},+)$. Besides, suppose that $\Omega$ is an F-CohFT with flat unit $\mathrm{e}$. Given $R(u) \in \mathfrak{Giv}$, set 
	\begin{equation}
		T'_{R}(u)
		\coloneqq
		u\big(R(u) - \id_{V_0}\big)\mathrm{e}
		\qquad \text{and} \qquad
		T''_{R}(u)
		\coloneqq
		u\big(\id_{V_0} - R^{-1}(u)\big)\mathrm{e} \,.	
	\end{equation}
	Then $\hat{T}'_{R} \hat{R} \Omega$ and $\hat{R} \hat{T}''_{R} \Omega$ coincide, resulting in an F-CohFT with flat unit $\mathrm{e}$.
\end{thm}

%–––––––––––––––––––––––––––––––––––––––––––%
\subsection{New symmetries of F-CohFTs: the tick}
\label{subsec:new:symm}
%–––––––––––––––––––––––––––––––––––––––––––%
As F-CohFTs are subjected to a less restrictive set of axioms compared to CohFTs, one can expect that they admit a larger group of symmetries. We propose an additional action, the tick, that preserves F-CohFTs. Contrarily to the R-action, this is a linear action.
%–––––––––––––––––––––––––––––––––––––––––––%

%–––––––––––––––––––––––––––––––––––––––––––%
Consider the abelian group (for the addition)
\begin{equation}\label{eq:tick:space} 
	\mathfrak{tick}
	\coloneqq
	\prod_{k \geq 2} \big(H^{\textup{even}}(\Mbar_{0,k}) \otimes V_0\bbraket{u}^{\otimes k}\big)^{\Sy_k} \,.
\end{equation}
Here we took the invariants under the action of the symmetric group $\Sy_k$ by simultaneous permutation of the $V_0$ factors and the marked points on the moduli space side, and the unstable summand $k=2$ should be understood as $\Sym{2}{V_0\bbraket{u}}$. Given an F-CohFT $\Omega$ on $V_0$ and an element $\Sha \in \mathfrak{tick}$ written as
\begin{equation}
	\Sha = \sum_{k \geq 2} \Sha_{k}(u_1,\ldots,u_k) \,,
\end{equation}
we define
\begin{multline}\label{eq:tick:action} 
	(\hat{\Sha}\Omega)_{g,1+n}
	\coloneqq
	\sum_{m \geq 0} \frac{1}{m!}
	\sum_{\bm{k} \in \mathbb{Z}_{\geq 2}^m}
		\xi_{\bm{\Gamma}_{\bm{k}},*}
		\Bigg[
			\Omega_{g - |\bm{k}| + m,1 + n + |\bm{k}|}
			\bigg(
				\id_{V_0^{\otimes n}} \\
				\otimes
				\bigotimes_{\ell = 1}^{m}
					\Sha_{k_{\ell}}
					\bigl(
						\psi_{(\ell,1)},\ldots,\psi_{(\ell,k_{\ell})}
					\bigr)
			\bigg)
		\Bigg] \,,
\end{multline}
where the $m = 0$ term is just $\Omega_{g,1+n}$ and we have denoted $|\bm{c}| \coloneqq c_1 + \cdots + c_m$ for an $m$-tuple $\bm{c} = (c_1, \dots, c_m)$ of positive integers. The map $\xi_{\bm{\Gamma}} \colon \Mbar_{\bm{\Gamma}} \to \Mbar_{g,1+n}$ is again the inclusion of the closed boundary stratum defined by the stable graph $\bm{\Gamma}$, and $\bm{\Gamma}_{\bm{k}}$ is the stable graph defined as follows.
\begin{itemize}
	\item It has a central vertex of genus $g - |\bm{k}| + m$ and valency $1 + n + |\bm{k}|$; the half-edges consist of all the leaves, labelled by $0,1,\dots,n$, and $|\bm{k}|$ additional half-edges.

	\item It has $m$ additional vertices, called tick vertices, having genus $0$ and valencies $\bm{k} = (k_1,\dots,k_m)$; all half-edges are connected to the central vertex. In the unstable case $k_i = 2$, this should be amended: there is no vertex but rather a loop attached to the central vertex.
	
	\item The half-edges connecting the central vertex to the $\ell$-th tick vertex (with corresponding $\psi$-classes appearing in \eqref{eq:tick:action}) are labelled as $(\ell,1), (\ell,2), \dots, (\ell,k_{\ell})$.
\end{itemize}
The stable graph and the convention for the tick vertices are depicted below.
\begin{equation}\label{eq:Gamma:tick}
\begin{tikzpicture}[baseline, scale=.6]
		\node at (-4.2,0) {$\bm{\Gamma}_{\bm{h},\bm{k}} =$};

		\draw[rounded corners] (-3,-1) -- (3,-1) -- (3,1) -- (-3,1) -- cycle;
		\node at (0,-.25) {\small $g-|\bm{k}|+m$};

		\draw (0,-1) -- (0,-2);
		\node at (0,-2) [right] {\scriptsize$0$};

		\draw (-2.7,1) -- (-2.7,2);
		\node at (-2,1.5) {\small $\cdots$};
		\draw (-1.3,1) -- (-1.3,2);
		\draw [decorate,decoration={calligraphic brace,amplitude=3pt},line width=1.25pt,xshift=0pt,yshift=4pt] (-2.8,2) -- (-1.2,2) node [black,midway,yshift=9pt] {\scriptsize $n$};

		\draw (.1,2) to[out=-150,in=90] (-.4,1);
		\draw (.1,2) to[out=-120,in=90] (-.25,1);
		\node at (.2,1.25) {\tiny$\cdots$};
		\draw (.1,2) to[out=-30,in=90] (.6,1);
		\draw[fill=white] (.1,2) circle (.45cm);
		\node at (.1,2) {\tiny$0$};

		\draw [decorate,decoration={calligraphic brace,mirror,amplitude=3pt},line width=1.25pt,xshift=0pt,yshift=-3pt] (-.5,1) -- (.7,1) node [black,midway,yshift=-7pt] {\tiny$k_1$};

		\node at (1.2,2) {\small $\cdots$};
		
		\begin{scope}[xshift=2.1cm]
			\draw (.1,2) to[out=-150,in=90] (-.4,1);
			\draw (.1,2) to[out=-120,in=90] (-.25,1);
			\node at (.2,1.25) {\tiny$\cdots$};
			\draw (.1,2) to[out=-30,in=90] (.6,1);
			\draw[fill=white] (.1,2) circle (.45cm);
			\node at (.1,2) {\tiny$0$};

			\draw [decorate,decoration={calligraphic brace,mirror,amplitude=3pt},line width=1.25pt,xshift=0pt,yshift=-3pt] (-.5,1) -- (.7,1) node [black,midway,yshift=-7pt] {\tiny$k_{\!m}$};
		\end{scope}
		
		\begin{scope}[xshift=8cm,scale=1.5]
			\draw (.1,.5) to[out=-150,in=90] (-.4,-.5);
			\draw (.1,.5) to[out=-120,in=90] (-.25,-.5);
			\node at (.2,-.25) {\small$\cdots$};
			\draw (.1,.5) to[out=-30,in=90] (.6,-.5);
			\draw[fill=white] (.1,.5) circle (.45cm);
			\node at (.1,.5) {\small$0$};

			\node at (-.4,-.5) [left] {\tiny $(\ell,1)$};
			\node at (-.25,-.5) [below] {\tiny $(\ell,2)$};
			\node at (.6,-.5) [below] {\tiny $(\ell,k_{\!\ell})$};
		\end{scope}
\end{tikzpicture}
\end{equation}

\begin{thm}\label{thm:tick}
	The collection of maps $\hat{\Sha} \Omega$ forms an F-CohFT. The resulting action is an abelian group action of $(\mathfrak{tick},+)$. Besides, if $\Omega$ has a flat unit, so does $\hat{\Sha}\Omega$.
\end{thm}

\begin{proof}
	The $\Sy_n$-equivariance follows from the symmetry of $\Omega$ and the definition of the tick action. Preservation of the flat unit axiom is straightforward, as tick vertices have no leaves. That this defines an abelian group action on collections of maps is clear. The non-trivial claim is that it respects the glueing axiom of separating kind.

	Fix $v_{[n]} = v_1 \otimes \cdots \otimes v_n \in V_0^{\otimes n}$, and let $\gl = \xi_{\bm{T}}$ be a glueing morphism of separating kind corresponding to the stable tree from \eqref{eq:glueing:sep:tree}, splitting the genus as $g = h + h'$ and the marked points (excluding the root) as $[n] = J \sqcup J'$. We need to re-express
	\begin{equation}\label{eq:push:pull} 
		\xi_{\bm{T}}^* \,
		\xi_{\bm{\Gamma}_{\bm{k}},*}
		\left[
			\Omega_{g_0,1 + n + |\boldsymbol{k}|}\bigg(v_{[n]} \otimes \bigotimes_{\ell = 1}^{m} \Sha_{k_{\ell}}\bigg)
		\right] 
	\end{equation}
	in terms of two classes $\Omega$. Here $g_0 = g - |\bm{k}| + m$ is the genus of the central vertex and the $\psi$-classes in $\Sha$ are omitted from the notation. A strategy would be to move the pullback to the right of the pushforward, as we will then be able to use the F-CohFT axioms for $\Omega$. To this end, we have to understand the intersection of the closed boundary strata corresponding to $\bm{T}$ and $\bm{\Gamma}_{\bm{k}}$. This is the disjoint union of boundary strata corresponding to stable graphs $\bm{\Gamma}$ which map to both $\bm{T}$ and $\bm{\Gamma}_{\bm{k}}$ after contraction of edges. The two inclusion morphisms at the level of the corresponding moduli spaces are denoted $\eta_{\bm{T}}$ and $\eta_{\bm{\Gamma}_{\bm{k}}}$ (we omit the dependence on $\bm{\Gamma}$) and they are such that the following diagram commutes.
	\begin{equation}
	\begin{tikzcd}
		{\Mbar_{\bm{\Gamma}}} & {\Mbar_{\bm{T}}} \\
		{\Mbar_{\bm{\Gamma}_{\bm{k}}}} & {\Mbar_{g,1+n}}
		\arrow["{\eta_{\bm{\Gamma}_{\bm{k}}}}"', from=1-1, to=2-1]
		\arrow["{\xi_{\bm{\Gamma}_{\bm{k}}}}"', from=2-1, to=2-2]
		\arrow["{\xi_{\bm{T}}}", from=1-2, to=2-2]
		\arrow["{\eta_{\bm{T}}}", from=1-1, to=1-2]
	\end{tikzcd}
	\end{equation}
	It is important to notice that the pre-image (under the contraction map) in $\bm{\Gamma}$ of the unique edge in $\bm{T}$ must be separating. Together with the condition that the tick vertices have genus $0$, this implies that:
	\begin{enumerate}[label=(\roman*)]
		\item\label{tick:1} it cannot coincide with any pre-image of an edge in $\bm{\Gamma}_{\bm{k}}$ between the central vertex and one of the tick vertices;

		\item\label{tick:2} it cannot split one of the tick vertices.
	\end{enumerate}
	Therefore, the set of possible stable graphs $\bm{\Gamma}$ consists of exactly those stable graphs obtained from $\bm{\Gamma}_{\bm{k}}$ by separating the central vertex into two vertices $\nu$ and $\nu'$ and distributing among them the genus and the tick vertices:
	\begin{equation}\label{eq:gamma}
	\begin{tikzpicture}[baseline, scale=.6]
			\node at (-4,0) {$\bm{\Gamma} =$};

			\draw (0,-1) -- (0,-2);
			\node at (0,-2) [right] {\scriptsize$0$};

			\draw (2,0) -- (5,3);

			\node at (3.3,-1.2) {\small$\nu\vphantom{\nu'}$};

			\draw[rounded corners,fill=white] (-3,-1) -- (3,-1) -- (3,1) -- (-3,1) -- cycle;
			\node at (0,0) {\small $h-|\bm{k}_{L}|+|L|$};

			\draw (-2.7,1) -- (-2.7,2);
			\node at (-2,1.5) {\small $\cdots$};
			\draw (-1.3,1) -- (-1.3,2);
			\node at (-2,2.5) {\scriptsize $J$};

			\draw (.1,2) to[out=-150,in=90] (-.4,1);
			\draw (.1,2) to[out=-120,in=90] (-.25,1);
			\node at (.2,1.25) {\tiny$\cdots$};
			\draw (.1,2) to[out=-30,in=90] (.6,1);
			\draw[fill=white] (.1,2) circle (.45cm);

			\node at (1.2,2) {\small $\cdots$};
			
			\begin{scope}[xshift=2.1cm]
				\draw (.1,2) to[out=-150,in=90] (-.4,1);
				\draw (.1,2) to[out=-120,in=90] (-.25,1);
				\node at (.2,1.25) {\tiny$\cdots$};
				\draw (.1,2) to[out=-30,in=90] (.6,1);
				\draw[fill=white] (.1,2) circle (.45cm);
			\end{scope}

			\draw [decorate,decoration={calligraphic brace,amplitude=3pt},line width=1.25pt,xshift=0pt,yshift=4pt] (-.2,2.5) -- (2.5,2.5) node [black,midway,yshift=9pt] {\scriptsize $L$};

			\begin{scope}[xshift=7cm,yshift=3cm]
				\node at (3.3,-1.2) {\small$\nu'$};

				\draw[rounded corners,fill=white] (-3,-1) -- (3,-1) -- (3,1) -- (-3,1) -- cycle;
				\node at (0,0) {\small $h'-|\bm{k}_{L\!'}|+|L'|$};

				\draw (-2.7,1) -- (-2.7,2);
				\node at (-2,1.5) {\small $\cdots$};
				\draw (-1.3,1) -- (-1.3,2);
				\node at (-2,2.5) {\scriptsize $J'$};

				\draw (.1,2) to[out=-150,in=90] (-.4,1);
				\draw (.1,2) to[out=-120,in=90] (-.25,1);
				\node at (.2,1.25) {\tiny$\cdots$};
				\draw (.1,2) to[out=-30,in=90] (.6,1);
				\draw[fill=white] (.1,2) circle (.45cm);

				\node at (1.2,2) {\small $\cdots$};
				
				\begin{scope}[xshift=2.1cm]
					\draw (.1,2) to[out=-150,in=90] (-.4,1);
					\draw (.1,2) to[out=-120,in=90] (-.25,1);
					\node at (.2,1.25) {\tiny$\cdots$};
					\draw (.1,2) to[out=-30,in=90] (.6,1);
					\draw[fill=white] (.1,2) circle (.45cm);
				\end{scope}

				\draw [decorate,decoration={calligraphic brace,amplitude=3pt},line width=1.25pt,xshift=0pt,yshift=4pt] (-.2,2.5) -- (2.5,2.5) node [black,midway,yshift=9pt] {\scriptsize $L'$};
			\end{scope}
	\end{tikzpicture}
	\end{equation}
	with $L \sqcup L' = [m]$ and similarly for $\bm{k}_{L}$, and $\bm{k}_{L'}$. In cohomology, the commuting diagram yields
	\begin{equation}\label{eq:push:pull:comm}
		\xi_{\bm{T}}^* \, \xi_{\bm{\Gamma}_{\bm{k}},*}
		=
		\sum_{\bm{\Gamma}} \epsilon_{\bm{\Gamma}} \cdot \eta_{\bm{\Gamma}_{\bm{k}},*} \, \eta_{\bm{T}}^* \, ,
	\end{equation}
	where $\epsilon_{\bm{\Gamma}}$ is the excess class, i.e. the Euler class of the normal bundle of the intersection of the two boundary strata under consideration. As explained in \cite{PPZ15}, it is equal to
	\begin{equation}
		\epsilon_{\bm{\Gamma}}
		=
		\prod_{e} (- \psi_{e'} - \psi_{e''}) \,,
	\end{equation}
	where the product ranges over all edges of $\bm{\Gamma}$ that are common to $\bm{T}$ and $\bm{\Gamma}_{\bm{k}}$, and $(\psi_{e'}$, $\psi_{e''})$ are the $\psi$-classes attached to the marked points joined by $e$. Condition \ref{tick:1} implies that there are no such edges, hence $\epsilon_{\bm{\Gamma}} = 1$ for all $\bm{\Gamma}$. We then employ \eqref{eq:push:pull:comm} in \eqref{eq:push:pull}: since $\eta_{\bm{T}}$ is a glueing morphism of separating kind and $\Omega$ is an F-CohFT, we have
\begin{equation}
\begin{split}
\label{222sk}
	& \quad 	\eta_{\bm{T}}^* \Omega_{g_0,1+n + |\boldsymbol{k}|}\bigg(
			v_{[n]}
			\otimes
			\bigotimes_{\ell = 1}^m \Sha_{k_{\ell}}
		\bigg) \\
		& =
		\Omega_{h_0,1+(1+|J|+|\bm{k}_{L}|)}
		\left(
			\Omega_{h_0',1+|J'|+|\bm{k}_{L'}|}
			\bigg(
				v_{J'} \otimes \bigotimes_{\ell \in L'} \Sha_{k_{\ell}}
			\bigg)
			\otimes
			v_{J} \otimes \bigotimes_{\ell \in L} \Sha_{k_{\ell}}
		\right), 
\end{split} 
\end{equation}
	where $h_0 = h-|\bm{k}_{L}|+|L|$ and $h_0' = h'-|\bm{k}_{L'}|+|L'|$. Applying further $\eta_{\bm{\Gamma}_{\bm{k}},*}$ to the above equation means pushing forward by the map contracting all edges of the tick vertices. This can be achieved by contracting the edges connecting the tick vertices to $\nu'$ first, and the edges connecting the tick vertices to $\nu$ second:
	\begin{equation}\label{eq:split:eta}
		\eta_{\bm{\Gamma}_{\bm{k}},*}
		=
		\eta_{\nu,*} \circ \eta_{\nu',*} \,.
	\end{equation}
	We conclude the computation by summing over all stable graphs $\bm{\Gamma}_{\bm{k}}$ as above, and then over all compatible stable graphs $\bm{\Gamma}$. In view of \eqref{eq:split:eta} and the fact that ticks are distributed on each of the factor in \eqref{222sk} independently, this re-constructs two independent tick actions on the two F-CohFTs placed at $\nu$ and $\nu'$, namely
	\begin{equation}
		\xi_{\bm{T}}^*(\hat{\Sha}\Omega)_{g,1+n}(v_{[n]})
		=
		(\hat{\Sha}\Omega)_{h,1+(1+|J|)}\left(
			(\hat{\Sha}\Omega)_{h',1+|J'|}(v_{J'})
			\otimes
			v_{J}
		\right) .
	\end{equation}
	This concludes the proof.
\end{proof}

\begin{rem}
	The tick action commutes with the translation but does not commute with the change of basis. Instead, we have $\hat{L}\hat{\Sha} \Omega = \hat{\Sha}{}_{L} \hat{L} \Omega$ with $[\Sha_L]_{k} \coloneqq L^{\otimes k} \circ \Sha_{k}$. The tick action does not commute with the R-action, and generates new operations when combined with it. It would be interesting to have a global description (i.e. not just by generators) of the group generated by $\hat{L},\hat{R},\hat{\Sha}$ and of its action on F-CohFTs.
\end{rem}

\begin{rem}
	One may wonder if it is possible to have symmetries coming from attaching ``ticks with genus $h > 0$''. However, in this case condition \ref{tick:2} would be violated: the tick vertices could have been split as
	\begin{equation}
	\begin{tikzpicture}[baseline,scale=.8]
		\draw (.1,.5) -- (.1,2);

		\draw (.1,.5) to[out=-150,in=90] (-.4,-.5);
		\draw (.1,.5) to[out=-120,in=90] (-.25,-.5);
		\node at (.2,-.25) {\small$\cdots$};
		\draw (.1,.5) to[out=-30,in=90] (.6,-.5);
		\draw[fill=white] (.1,.5) circle (.45cm);
		\node at (.1,.5) {\small$h'$};

		%\node at (-.4,-.5) [left] {\tiny $1$};
		%\node at (-.25,-.5) [below] {\tiny $2$};
		%\node at (.6,-.5) [below] {\tiny $k$};
		\draw [decorate,decoration={calligraphic brace,amplitude=3pt},line width=1.25pt,xshift=0pt,yshift=-4pt] (.7,-.5) -- (-.5,-.5) node [black,midway,yshift=-9pt] {\scriptsize $k$};

		\draw[fill=white] (.1,2) circle (.45cm);
		\node at (.1,2) {\small$h''$};
	\end{tikzpicture}
	\end{equation}
	whenever $k > 1$ and for $h = h' + h''$ with $h'' > 0$.
\end{rem}

\begin{rem}
	One may wonder if it is possible to have symmetries coming from attaching ``ticks with heads'', i.e. elements of
	\begin{equation}
		\big(
			\Hom(V^{\odot l},V^{\odot k})
			\otimes
			H^{\text{even}}(\Mbar_{0,l+k})
		\big)^{\mathfrak{S}_l \times \mathfrak{S}_{k}},
		\qquad
		\text{with $l > 0$.}
	\end{equation}
	Although \ref{tick:1} and \ref{tick:2} are respected for $k \geq 2$, such an action does not preserve F-CohFTs. Indeed, let $\tilde{\Omega}$ be the result of acting on an F-CohFT by attaching an arbitrary number of ticks with heads. The $l > 0$ head(s) of each tick are inputs of $\tilde{\Omega}$. Pulling back such terms by a glueing morphism of separating kind, the heads still correspond to inputs of the pulled-back class.

	In comparison, the right-hand side of the F-CohFT axiom \eqref{eq:glueing:axiom} for $\tilde{\Omega}$ contains all terms where the output of the innermost $\tilde{\Omega}$ is connected to \emph{some} input of the outermost $\tilde{\Omega}$. The latter input may originate from the head of a tick, and such terms can never occur in the pulled-back class (see \eqref{eq:tick:head}), causing $\tilde{\Omega}$ to violate the required glueing property for an F-CohFT.
	\begin{equation}\label{eq:tick:head}
	\begin{tikzpicture}[baseline, scale=.6]
			\draw (0,-1) -- (0,-2);
			%\node at (0,-2) [right] {\scriptsize$0$};

			\draw[rounded corners,fill=white] (-3,-1) -- (3,-1) -- (3,1) -- (-3,1) -- cycle;
			\node at (0,0) {$\tilde{\Omega}$};

			\draw (-2.7,1) -- (-2.7,2);
			\node at (-2,1.5) {\small $\cdots$};
			\draw (-1.3,1) -- (-1.3,2);
			%\node at (-2,2.5) {\scriptsize $J$};

			\draw (.1,2) to[out=-150,in=90] (-.4,1);
			\draw (.1,2) to[out=-120,in=90] (-.25,1);
			\node at (.2,1.25) {\tiny$\cdots$};
			\draw (.1,2) to[out=-30,in=90] (.6,1);
			\draw (.1,2) to[out=150,in=-90] (-.4,3);
			\draw (.1,2) to[out=120,in=-90] (-.25,3);
			\node at (.2,2.75) {\tiny$\cdots$};
			\draw (.1,2) to[out=30,in=-90] (.6,3);
			\draw[fill=white] (.1,2) circle (.45cm);

			\node at (1.2,2) {\small $\cdots$};
			
			\begin{scope}[xshift=2.1cm]
				\draw (.1,2) to[out=-150,in=90] (-.4,1);
				\draw (.1,2) to[out=-120,in=90] (-.25,1);
				\node at (.2,1.25) {\tiny$\cdots$};
				\draw (.1,2) to[out=-30,in=90] (.6,1);
				\draw (.1,2) to[out=150,in=-90] (-.4,3);
				\draw (.1,2) to[out=120,in=-90] (-.25,3);
				\node at (.2,2.75) {\tiny$\cdots$};
				\draw (.1,2) to[out=30,in=-105] (.6,2.5) to[out=75,in=-90] (5,4);
				\draw[fill=white] (.1,2) circle (.45cm);
			\end{scope}

			\begin{scope}[xshift=7cm,yshift=5cm]

				\draw[rounded corners,fill=white] (-3,-1) -- (3,-1) -- (3,1) -- (-3,1) -- cycle;
				\node at (0,0) {$\tilde{\Omega}$};

				\draw (-2.7,1) -- (-2.7,2);
				\node at (-2,1.5) {\small $\cdots$};
				\draw (-1.3,1) -- (-1.3,2);
				%\node at (-2,2.5) {\scriptsize $J'$};

				\draw (.1,2) to[out=-150,in=90] (-.4,1);
				\draw (.1,2) to[out=-120,in=90] (-.25,1);
				\node at (.2,1.25) {\tiny$\cdots$};
				\draw (.1,2) to[out=-30,in=90] (.6,1);
				\draw (.1,2) to[out=150,in=-90] (-.4,3);
				\draw (.1,2) to[out=120,in=-90] (-.25,3);
				\node at (.2,2.75) {\tiny$\cdots$};
				\draw (.1,2) to[out=30,in=-90] (.6,3);
				\draw[fill=white] (.1,2) circle (.45cm);

				\node at (1.2,2) {\small $\cdots$};
				
				\begin{scope}[xshift=2.1cm]
					\draw (.1,2) to[out=-150,in=90] (-.4,1);
					\draw (.1,2) to[out=-120,in=90] (-.25,1);
					\node at (.2,1.25) {\tiny$\cdots$};
					\draw (.1,2) to[out=-30,in=90] (.6,1);
					\draw (.1,2) to[out=150,in=-90] (-.4,3);
					\draw (.1,2) to[out=120,in=-90] (-.25,3);
					\node at (.2,2.75) {\tiny$\cdots$};
					\draw (.1,2) to[out=30,in=-90] (.6,3);
					\draw[fill=white] (.1,2) circle (.45cm);
				\end{scope}
			\end{scope}
	\end{tikzpicture}
	\end{equation}
	We do not exclude, but find it unlikely, that universal symmetries (acting on any F-CohFTs) other than the ones presented exist.
\end{rem} 

%–––––––––––––––––––––––––––––––––––––––––––%
\section{Identification of the two theories}
 \label{sec:identification}
%–––––––––––––––––––––––––––––––––––––––––––%

%–––––––––––––––––––––––––––––––––––––––––––%
\subsection{F-CohFT amplitudes}
%–––––––––––––––––––––––––––––––––––––––––––%
Our main goal is to associate to a given F-CohFT $\Omega$ on $V_0$ a collection of linear maps, called amplitudes, of the form
\begin{equation}
	F_{g,1+n} \in \Hom(\Sym{n}{V_+},V_+)
	\qquad\text{on}\qquad
	V_+ \coloneqq V_0[u]
\end{equation}
that capture all intersections of $\Omega$ with $\psi$-classes. The space $V_+$ of $V_0$-valued polynomials in $u$ is called the \emph{loop space}, and the variable $u$ is responsible for controlling all possible powers of $\psi$-classes.

Unlike the usual setting, however, the definition of amplitudes associated with F-CohFTs involves the choice of an element $\mathscr{U} \in \End(V_0)[u_0]\bbraket{u}$ that keeps track of $\psi_0$, the class coupled to the output vector of the given F-CohFT. This is because, while F-CohFTs naturally treat input and output vectors differently, $\psi_0$ is treated on the same footing as all other $\psi$-classes. Thus, we are forced to `dualise' the loop variable to obtain an element of $V_+$ as output.

To this end, it will prove useful to introduce the space
\begin{equation}
	V_- \coloneqq V_0[u^{-1}]\frac{\dd u}{u}
\end{equation}
of $V_0$-valued polynomial differential forms in $u^{-1}$. The spaces $V_+$ and $V_-$ can be considered as `partially dual' to each other, with the duality taking place on the loop variable but not on $V_0$. Indeed, interpreting $V_+ = V_0 \otimes \CC[u]$ and $V_- = V_0 \otimes \CC[u^{-1}]\frac{\dd u}{u}$, we have the following identification at the level of loop variables: $\CC[u]^* \cong \CC\bbraket{u^{-1}}\frac{\dd u}{u}$, realised by the residue pairing
\begin{equation}
	\Braket{ \chi, f } = \Res_{u = 0} \ f(u) \chi(u)
\end{equation}
for $f \in \CC[u]$ and $\chi \in \CC\bbraket{u^{-1}}\frac{\dd u}{u}$. This interpretation will perhaps make the following definitions more natural.

\begin{defn}\label{def:up:down}
	An element $\mathscr{U}(u_0,u) \in \End(V_0)[u_0]\bbraket{u}$ is called non-degenerate if there exists $\mathscr{D}(u_0,u) \in \End(V_0)[u_0^{-1}]\bbraket{u^{-1}} \frac{\dd u_0 \, \dd u}{u_0 u}$ such that
	\begin{equation}
		\Res_{u,u_1=0} \
			\mathscr{U}(u_0,u) \mathscr{D}(u,u_1) f(u_1)
		=
		f(u_0) \,,
		\qquad
		\Res_{u,u_1=0} \
			\mathscr{D}(u_0,u) \mathscr{U}(u,u_1) \chi(u_1) 
		=
		\chi(u_0) \,,
	\end{equation}
	for all $f \in V_+$ and all $\chi \in V_-$. We call $\mathscr{U}$ an \emph{up-morphism}, and $\mathscr{D}$ a \emph{down-morphism}.
\end{defn}

Considering the `partial duality' between $V_+$ and $V_-$, it is natural to consider $\mathscr{U}$ and $\mathscr{D}$ as linear operators:
\begin{equation}\label{eq:U:D:maps}
\begin{aligned}
	\mathscr{U} \colon V_- & \longrightarrow V_+ 
	&\qquad&
	\mathscr{U}[\chi](u_0)
	\coloneqq
	\Res_{u = 0} \ \mathscr{U}(u_0,u)\chi(u) \,, \\
	\mathscr{D} \colon V_+ & \longrightarrow V_-
	&\qquad&
	\mathscr{D}[f](u_0)
	\coloneqq
	\Res_{u = 0} \ \mathscr{D}(u_0,u)f(u) \,.
\end{aligned}
\end{equation}
Abusing notations, we denote these maps with the symbols $\mathscr{U}$ and $\mathscr{D}$ respectively. The non-degeneracy condition simply asserts that $\mathscr{U}$ and $\mathscr{D}$ are inverses of each other as operators: $\mathscr{U} \circ \mathscr{D} = \id_{V_+}$ and $\mathscr{D} \circ \mathscr{U} = \id_{V_-}$.

For a fixed basis $(\mathrm{e}_{\alpha})_{\alpha \in \mathfrak{a}}$ of $V_0$, we have the natural bases
\begin{equation}\label{eq:basis:loop}
	\mathrm{e}_{(\alpha,k)} \coloneqq \mathrm{e}_{\alpha} u^k \in V_+ \,,
	\qquad\qquad
	\mathrm{e}_{\alpha}^k \coloneqq \mathrm{e}_{\alpha} \frac{\dd u}{u^{k+1}} \in V_- \,,
\end{equation}
indexed by $(\alpha,k) \in I = \mathfrak{a} \times \NN$. The positioning of the indices has been chosen to maintain a consistent use of Einstein's convention. In this case, the expression for $\mathscr{U}$ and $\mathscr{D}$ is given as
\begin{equation}\label{eq:U:D:coords}
	\mathscr{U}\bigl[ \mathrm{e}_{\beta}^{j} \bigr]
	=
	\mathscr{U}_{\beta}^{\alpha;i,j} \, \mathrm{e}_{(\alpha,i)} \,,
	\qquad\qquad
	\mathscr{D}\bigl[ \mathrm{e}_{(\beta,j)} \bigr]
	=
	\mathscr{D}_{\beta;i,j}^{\alpha} \, \mathrm{e}_{\alpha}^i \,,
\end{equation}
(which explains the choice of names for $\mathscr{U}$ and $\mathscr{D}$ as up- and down-morphisms respectively) and the non-degeneracy condition is $\mathscr{U}_{\mu}^{\beta;j,m} \mathscr{D}_{\alpha;m,i}^{\mu} = \delta_{\alpha}^{\beta} \, \delta_{i}^{j}$ and $\mathscr{D}_{\mu;j,m}^{\beta} \mathscr{U}_{\alpha}^{\mu;m,i} = \delta_{\alpha}^{\beta} \, \delta_{j}^{i}$.

\begin{ex}\label{ex:stndrd:up:down}
	A standard choice of up/down-morphisms is
	\begin{equation}
	\begin{split}
		\mathscr{U}(u_0,u) &= \id_{V_0} \frac{1}{1 - u_0u} = \id_{V_0} \sum_{k \geq 0} (u_0 u)^k \,, \\
		\mathscr{D}(u_0,u) &= \id_{V_0} \frac{\dd u_0 \, \dd u}{u_0 u - 1} = \id_{V_0} \sum_{k \geq 0} \frac{\dd u_0 \, \dd u}{(u_0 u)^{k + 1}} \,.
	\end{split}
	\end{equation}
	In this case, we have $\mathscr{U}_{\beta}^{\alpha;i,j} = \delta^{\alpha}_{\beta} \, \delta^{i,j}$ and $\mathscr{D}^{\alpha}_{\beta;i,j} = \delta^{\alpha}_{\beta} \, \delta_{i,j}$.
\end{ex}

\begin{defn}
	Let $\Omega$ be an F-CohFT and $(\mathscr{U},\mathscr{D})$ a choice of up/down-morphisms. Define the \emph{amplitudes} associated with $\Omega$ as the collection of linear maps $F_{g,1+n} \in \Hom(\Sym{n}{V_+},V_+)$ given by
	\begin{equation}\label{eq:ampl:F-CohFT}
		F_{g,1+n}(f_1 \otimes \cdots \otimes f_n)(u_0)
		\coloneqq
		\int_{\Mbar_{g,1+n}}
		\mathscr{U}(u_0,\psi_0)
		\Bigl[
			\Omega_{g,1+n} \bigl( f_1(\psi_1) \otimes \cdots \otimes f_n(\psi_n) \bigr)
		\Bigr]
		\,,
	\end{equation}
	where $f_i = f_i(u) \in V_+$ and the F-CohFT is extended from $V_0$ to $V_+$ by linearity. The \emph{ancestor (vector) potential} associated with $\Omega$ is the $\hbar^{-1}V_+\bbraket{\hbar}$-valued formal function
	\begin{equation}
		\Phi(x)
		\coloneqq
		\sum_{\substack{g,n \geq 0 \\ 2g-2+(1+n)>0}} \frac{\hbar^{g - 1}}{n!} \, F_{g,1+n}(x^{\otimes n}) \,,
	\end{equation}
	where $x = x(u)$ is the formal variable in $V_+$. In other words:
	\begin{equation}
		\Phi \in \mathrm{Fun}_{V_+} \coloneqq \prod_{g,n \geq 0} \hbar^{g - 1} \, \Hom(\Sym{n}{V_+},V_+)\,.
	\end{equation}
\end{defn}

We emphasise that, contrary to the usual setting, the amplitudes associated with a given F-CohFT depend on the choice of up/down-morphisms. Abusing notations, we omit this dependence. The necessity of such a choice is perhaps more transparent in coordinates: with the notation from \eqref{eq:basis:loop} and \eqref{eq:U:D:coords}, set
\begin{equation}
	\Braket{
		\tau^{\beta}_{\ell}\tau_{(\alpha_1,k_1)} \cdots \tau_{(\alpha_n,k_n)}
	}^{\Omega}_{g}
	\coloneqq
	\int_{\Mbar_{g,1+n}}
		\Braket{
			\mathrm{e}^{\beta},
			\Omega_{g,1+n}(\mathrm{e}_{\alpha_1} \otimes \cdots \otimes \mathrm{e}_{\alpha_n})
		}
		\psi_{0}^{\ell} \prod_{i=1}^n \psi_{i}^{k_i} \,.
\end{equation}
Here $\braket{ \mathrm{e}^{\beta}, \mathrm{e}_{\alpha} } = \delta_{\alpha}^{\beta}$ is the canonical pairing between $V_0^*$ and $V_0$. Then the amplitudes in coordinates read
\begin{equation}
	\indF{F}{g}{(\alpha_0,k_0)}{(\alpha_1,k_1),\dots,(\alpha_n,k_n)}
	=
	\mathscr{U}_{\beta}^{\alpha_0;k_0,\ell}
	\Braket{
		\tau^{\beta}_{\ell}\tau_{(\alpha_1,k_1)} \cdots \tau_{(\alpha_n,k_n)}
	}^{\Omega}_{g}
\end{equation}
and the ancestor potential is given by
\begin{equation}
	\Phi(x)
	=
	\sum_{\substack{g,n \geq 0 \\ 2g-2+(1+n)>0}} \frac{\hbar^{g - 1}}{n!} \,
		\indF{F}{g}{(\alpha_0,k_0)}{(\alpha_1,k_1),\dots,(\alpha_n,k_n)} \,
		\mathrm{e}_{(\alpha_0,k_0)}
		\prod_{i=1}^n x^{(\alpha_i,k_i)} \,,
\end{equation}
where $x = x^{(\alpha,k)} \mathrm{e}_{(\alpha,k)}$ denotes the formal variable on $V_+$.

\begin{rem}\label{rem:finite:F}
	Since $\overline{\mathcal{M}}_{g,1+n}$ has complex dimension $3g - 2 + n$, the evaluation of the tensors $F_{g,1+n} \in \Hom(\Sym{n}{V_+},V_+)$ on monomials $v_1 u^{d_1} \otimes \cdots \otimes v_n u^{d_n}$ vanishes whenever $d_1 + \cdots + d_n > 3g - 3 + (n+1)$. In particular, the tensor can be extended to $\Hom(\Sym{n}{\widehat{V}_+},V_+)$ for
	\begin{equation}
		\widehat{V}_+
		\coloneqq
		V_0\bbraket{u} \,,
	\end{equation}
	called the \emph{completed loop space}. By composition with the natural inclusion $V_+ \hookrightarrow \widehat{V}_+$, it can also be considered as an element of $\Hom(\Sym{n}{\widehat{V}_+},\widehat{V}_+)$. In the following, it will prove useful to introduce the `partial dual' completed space
	\begin{equation}
		\widehat{V}_-
		\coloneqq
		V_0\bbraket{u^{-1}} \frac{\dd u}{u} \,.
	\end{equation}
\end{rem}

%–––––––––––––––––––––––––––––––––––––––––––%
\subsection{Actions on F-CohFT amplitudes}
\label{subsec:actions:FCohFT:ampl}
%–––––––––––––––––––––––––––––––––––––––––––%
We can now describe the result of the different actions on F-CohFT at the level of amplitudes. For changes of basis and R-actions, this requires a concomitant transformation of the up/down-morphisms used to define the transformed amplitudes. Throughout the rest of the section, we fix an F-CohFT $\Omega$ on $V_0$ together with a choice $(\mathscr{U},\mathscr{D})$ of up/down-morphisms.

%–––––––––––––––––––––––––––––––––––––––––––%
\paragraph{Change of basis.}
%–––––––––––––––––––––––––––––––––––––––––––%
Given $L \in \GL(V_0)$, the amplitudes associated with $\hat{L}\Omega$ are given by
\begin{equation}\label{eq:change:basis:CohFT}
	(\hat{L}F)_{g,1+n}
	\coloneqq L
	\circ F_{g,1+n} \circ (L^{-1})^{\otimes n} \,,
\end{equation}
considering $L$ and $L^{-1}$ as elements in $\GL(V_+)$, provided we use the new up/down-morphisms
\begin{equation}\label{eq:new:up:down:change:basis}
	\preind{L}{\mathscr{U}}(u_0,u) \coloneqq L \circ \mathscr{U}(u_0,u) \circ L^{-1} \,,
	\qquad\qquad
	\preind{L}{\mathscr{D}}(u_0,u) \coloneqq L \circ \mathscr{D}(u_0,u) \circ L^{-1} \,
\end{equation}
to define the amplitudes $\hat{L}F$ of $\hat{L}\Omega$. Indeed:
\begin{equation}\label{eq:new:up:change:basis}
\begin{split}
	(\hat{L}F)_{g,1+n}&(f_1 \otimes \cdots \otimes f_n)
	=
	\int_{\Mbar_{g,1+n}}
		\preind{L}{\mathscr{U}}(u_0,\psi_0)
		\Bigl[
			\hat{L}\Omega_{g,1+n}\bigl( f_1(\psi_1) \otimes \cdots \otimes f_n(\psi_n) \bigr)
		\Bigr] \\
	&=
	\int_{\Mbar_{g,1+n}}
		\bigl( L \circ \mathscr{U}(u_0,\psi) \circ L^{-1} \bigr)
		\Bigl[
			L \circ \Omega_{g,1+n}\bigl( L^{-1}f_1(\psi_1) \otimes \cdots \otimes L^{-1}f_n(\psi_n) \bigr)
		\Bigr] \\
	&=
	L \circ
	\int_{\Mbar_{g,1+n}}
		\mathscr{U}(u_0,\psi)
		\Bigl[
			\Omega_{g,1+n}\bigl( L^{-1}f_1(\psi_1) \otimes \cdots \otimes L^{-1}f_n(\psi_n) \bigr)
		\Bigr] \\
	&=
	\bigl( L \circ F_{g,1+n} \circ (L^{-1})^{\otimes n} \bigr)(f_1 \otimes \cdots \otimes f_n) \,.
\end{split}
\end{equation}
Notice that the above computation fixes the new up-morphism $\preind{L}{\mathscr{U}}$. The non-degeneracy condition satisfied by $\mathscr{U}$ immediately implies that $\preind{L}{\mathscr{U}}$ is non-degenerate too, with associated down-morphism given as in \eqref{eq:new:up:down:change:basis}.

To sum up, the transformed amplitudes $\hat{L}F_{g,1+n}$ are of the form $\preind{L}{F}_{g,1+n}$, where the notation is in accordance with the one introduced in Section~\ref{subsec:change:bases} for changes of bases in the context of F-Airy structures. In particular, the ancestor potential transforms as
\begin{equation}
	\preind{L}{\Phi} = L \circ \Phi \circ L^{-1} \,.
\end{equation}
The changes of basis for F-CohFTs realise only special changes of bases for the amplitudes, namely those corresponding to $L_{\textup{t}} = L_{\textup{s}} = L$ induced by $\mathrm{GL}(V_0) \subset \mathrm{GL}(V_+)$.

%–––––––––––––––––––––––––––––––––––––––––––%
\paragraph{R-action.}
%–––––––––––––––––––––––––––––––––––––––––––%
Let $R \in \mathfrak{Giv}$ be an R-matrix. In order to analyse the amplitudes associated with $\hat{R}\Omega$, we need to introduce three operators $B_{R}\in \Hom(V_+,\widehat{V}_+)$, $L_{R,\textup{s}}, \ L_{R,\textup{t}} \in \GL(\widehat{V}_+)$ associated with the R-matrix as follows.
\begin{itemize}
	\item First, recall the definition of the edge weight \eqref{eq:Bglbv:F-Giv}:
	\begin{equation}
		\mathscr{E}_{R}(u_0,u)
		=
		\frac{\id_{V_0} - R^{-1}(u_0) \circ R(-u)}{u_0 + u} \in \End(V_0)\bbraket{u_0,u} \,,
	\end{equation}
	Notice that $\mathscr{E}_{R}$ can be equivalently considered as a linear operator:
	\begin{equation}
		\mathscr{E}_{R} \colon V_- \longrightarrow \widehat{V}_+ \,,
		\qquad
		\mathscr{E}_{R}[\chi](u)
		\coloneqq
		\Res_{u' = 0} \ \mathscr{E}_{R}(u,u')\chi(u') \,.
	\end{equation}
	Then, we have a well-defined map $B_{R} \coloneqq \mathscr{E}_{R} \circ \mathscr{D} \in \Hom(V_+,\widehat{V}_+)$, explicitly given by
	\begin{equation}\label{eq:BR}
		B_{R}[f](u_0) = \Res_{u,u' = 0} \ \mathscr{E}_{R}(u_0,u) \, \mathscr{D}(u,u') \, f(u') \,.
	\end{equation}

	\item Second, let $L_{R,\textup{s}}$ and $L_{R,\textup{t}}$ be the elements in $\GL(\widehat{V}_+)$ acting as multiplication by $R(u)$ and $R(-u)$ respectively:
	\begin{equation}
		L_{R,\textup{s}}[f](u) \coloneqq R(u)f(u) \,,
		\qquad
		L_{R,\textup{t}}[f](u) \coloneqq R(-u)f(u) \,.
	\end{equation}
	Notice that $L_{R,\textup{s}}$ and $L_{R,\textup{t}}$ are indeed invertible, with inverse being the multiplication by $R^{-1}(u)$ and $R^{-1}(-u)$ respectively.

	\item Third, let $V_{R,+} \coloneqq L_{R,\textup{t}}(V_+)$. This is a subspace of $\widehat{V}_+$ isomorphic to $V_+$. We introduce a new up-morphism $\preind{R}{\mathscr{U}} : V_- \rightarrow V_{R,+}$ by the formula
	\begin{equation}\label{eq:R:up}
		\preind{R}{\mathscr{U}}(u_0,u)
		\coloneqq
		R(-u_0) \circ \mathscr{U}(u_0,u) \circ R^{-1}(-u) \,.
	\end{equation}
\end{itemize}
We stress that the new up-morphism does not take value in $V_+$ but in the isomorphic space $V_{R,+}$. Unlike the case above, its invertibility is not immediately clear. To justify it we notice that the new up-morphism $\preind{R}{\mathscr{U}}$ can be written as the composition
\begin{equation}
	\preind{R}{\mathscr{U}}
	=
	L_{R,\textup{t}} \circ \mathscr{U} \circ M_{R}^{-1}
	\colon
	V_- \longrightarrow V_{R,+} \,,
\end{equation}
where $L_{R,\textup{t}} \in \GL(\widehat{V}_+)$ is the multiplication by $R(-u)$ as above, $\mathscr{U}$ is the old up-morphism, and $M_{R}^{-1} \in \GL(V_-)$ is defined as
\begin{equation}
	M_{R}^{-1}[\chi](u) \coloneqq \Bigl[ R^{-1}(-u) \chi(u) \Bigr]_- \,.
\end{equation}
Here $[ \ph ]_-$ is the projection from $\dd \widehat{V}_+ \oplus V_-$ to the $V_-$ summand, and we have silently used the natural inclusion $V_{R,+} \hookrightarrow \widehat{V}_+$. It is easy to see that $M_{R}^{-1}$ is indeed invertible with inverse
\begin{equation}
	M_{R}[\chi](u) = \Bigl[ R(-u) \chi(u) \Bigr]_- \,.
\end{equation}
The notation has been fixed so that it matches with the R-matrix. It is then clear that $\preind{R}{\mathscr{U}}$ is indeed invertible, with inverse
\begin{equation}
	\preind{R}{\mathscr{D}}
	\coloneqq
	M_R \circ \mathscr{D} \circ L_{R,\textup{t}}^{-1} \colon V_{R,+} \longrightarrow V_- \,.
\end{equation}
With these conventions, using the non-degeneracy condition $\mathscr{D} \circ \mathscr{U} = \id_{V_-}$, it is easy to check that the amplitudes associated with $\hat{R}\Omega$ and the up-morphism $\preind{R}{\mathscr{U}}$ are the elements $(\hat{R}F)_{g,1+n} \in \Hom(\Sym{n}{V_{R,+}},V_{R,+})$ given by
\begin{equation}\label{RFtrans}
	(\hat{R}F)_{g,1+n}
	=
	L_{R,\textup{t}}
	\circ
	\left(
		\sum_{\bm{T} \in \mathbb{T}_{g,1+n}}
			\bigg( \bigotimes_{v \in \Vert(\bm{T})} F_{g(v),1+n(v)} \bigg)
			\underset{\bm{T}}{\circ}
			\bigg( \bigotimes_{e \in \Edge(\bm{T})} B_{R} \Bigg)
	\right)
	\circ
	(L_{R,\textup{s}}^{-1})^{\otimes n}
	\,.
\end{equation}
On the right-hand side, Remark~\ref{rem:finite:F} was used to upgrade the definition $F_{g,1+n}$ from an element of $ \Hom(\Sym{n}{V_+},V_+)$ to an element of $\Hom(\Sym{n}{V_{R,+}},V_+)$. For the same finiteness reason, $(\hat{R}F)_{g,1+n}$ also extend as elements of $\Hom(\Sym{n}{\widehat{V}_+},\widehat{V}_+)$. 

The justification of \eqref{RFtrans} together with the formula for the new up-morphism \eqref{eq:R:up} is completely analogous to \eqref{eq:new:up:change:basis} and omitted. To sum up, the transformed amplitudes are of the form $\preind{L_{R}}{\vphantom{F}}(\preind{B_{R}\vphantom{L}}{F})$, following the notation introduced in Section~\ref{sec:action} for changes of bases and Bogoliubov transformations in the context of F-Airy structures. The corresponding transformation of the ancestor potential is uniquely characterised by the following fixed point equation:
\begin{equation}
	\bigl(
		L_{R,\textup{t}}^{-1} \circ \preind{R}{\Phi}
	\bigr) (x)
	=
	\Phi\left(
		L_{R,\textup{s}}^{-1}(x)
		+
		\hbar
		\bigl( B_{R} \circ L_{R,\textup{t}}^{-1} \circ \preind{R}{\Phi} \bigr)(x)
	\right) .
\end{equation}

%–––––––––––––––––––––––––––––––––––––––––––%
\paragraph{Translation.}
%–––––––––––––––––––––––––––––––––––––––––––%
Given $T(u) \in u^2 V_0\bbraket{u}$, the amplitudes associated with $\hat{T}\Omega$ and the same up/down-morphisms are simply
\begin{equation}
	(\hat{T}F)_{g,1+n}
	=
	\sum_{m \geq 0}
		\frac{1}{m!} F_{g,1+n+m} \bigl( \id_{V_+}^{\otimes n} \otimes T^{\otimes m} \bigr) \,.
\end{equation}
In particular the translated amplitudes take the form $\preind{T}{F}$, following the notation introduced in Section~\ref{subsec:translation} for translations in the context of F-Airy structures (cf. Remark~\ref{rem:translation}). In particular, the ancestor potential transforms as
\begin{equation}
	\preind{T}{\Phi}(x)
	=
	\Phi(x + T)
	-
	\hbar^{-1} \big( G + H(x) \big)
\end{equation}
where $G \coloneqq \sum_{m \ge 2} \frac{1}{m!} F_{0,1+m}(T^{\otimes m})$ and $H(x) \coloneqq \sum_{m \ge 1} \frac{1}{m!} F_{0,2+m}(x \otimes T^{\otimes m})$.

%–––––––––––––––––––––––––––––––––––––––––––%
\paragraph{Tick action.}
%–––––––––––––––––––––––––––––––––––––––––––%
Ticks act linearly on F-CohFTs and they act as differential operators at the level of ancestor potentials (we omit the action at the level of amplitudes, since it is easy to extract it from the definition). To formulate it precisely, we first introduce the following variant of the Weyl algebra of differential operators on $V_+$. As a graded vector space, it is
\begin{equation}
	\label{534eq} \mathfrak{W}_{V_+}
	\coloneqq
	\prod_{g,n,m \geq 0} \mathfrak{W}_{V_+}[g,n,m] \,,
	\qquad
	\mathfrak{W}_{V_+}[g,n,m]
	\coloneqq
	\hbar^{g - 1} \, \Hom(\Sym{n}{V_+},\Sym{m}{V_+}) \,,
\end{equation}
where the respective summands have weight $2g - 2 + n + m$. As an associative algebra it is the quotient of the free complete associative algebra generated by $V_+$ and $V_+^*$ modulo the relations
\begin{equation}
	\forall v,\,w \in V_+
	\quad
	\forall \lambda,\,\mu \in V_+^*
	\qquad
	[v,w] = 0 \,,
	\quad
	[v,\lambda] = \hbar \lambda(v) \,,
	\quad
	[\lambda,\mu] = 0 \,.
\end{equation}
We consider its subalgebra $\mathfrak{W}_{V_+}^{\geq 0}$ keeping only components of non-negative weight. 

We let $\mathfrak{W}_{V_+}^{\geq 0}$ act in a $\CC\bbraket{\hbar}$-linear and natural way on the space of vector-valued formal functions on $V_+$, that is on 
\begin{equation}
	\mathrm{Fun}_{V_+} \coloneqq \prod_{g,n \geq 0} \hbar^{g - 1} \, \Hom(\Sym{n}{V_+},V_+) \,.
\end{equation}
On the generators of $\mathfrak{W}_V$, this means:
\begin{itemize}
	\item elements of $\Hom(V_+,\CC)$ act by multiplication in the completed symmetric algebra of $V_+^*$;
	
	\item an element $v \in V_+$ acts on $\lambda \in V_+^*$ as $v.\lambda \coloneqq \hbar \lambda(v)$ and this action is extended to a derivation on $\mathrm{Fun}_{V_+}$.
\end{itemize}
The corresponding representation is denoted $\mathrm{Op} \colon \mathfrak{W}_{V_+}^{\ge 0} \to \End(\mathrm{Fun}_{V_+})$. Ticks determine elements of the Weyl algebra upon integration over the moduli space of curves, which are pure differential operators (i.e. we only see the summands $n = 0$ of \eqref{534eq}). More precisely, the following hold, without any change of up/down-morphisms.

\begin{lem}
Given $\Sha \in \mathfrak{tick}$, the ancestor potential associated with $\hat{\Sha}\Omega$ is
	\begin{equation}
		\preind{\Sha}{\Phi}
		=
		\exp\Biggl[
				\sum_{k \geq 2} \mathrm{Op}\left(
				\int_{\Mbar_{0,k}} \Sha_{k}
			\right)
		\Biggr].\Phi 
	\end{equation}
	where it is understood that in the $(0,2)$-summand the integration is omitted: there is no moduli space and $\Sha_{2}$ is already an element of $\Sym{2}{V_+}$.
	\end{lem}
Note that the operators appearing in the exponential for the tick action do commute (they are of the form $\partial_x^k$).

%–––––––––––––––––––––––––––––––––––––––––––%
\subsection{The identification}
\label{subsec:identif}
%–––––––––––––––––––––––––––––––––––––––––––%
The striking similarity between F-Airy structures and F-CohFTs is not a coincidence, and is parallel to the one explored in \cite{DOSS14} in the ordinary case. The identification consists of two steps. Firstly, we identify the two theories for a `small' set of cases, namely topological F-CohFTs on $V_0$ (i.e. F-CohFT obtained by coupling the fundamental class of $\Mbar_{g,1+n}$ to an F-TFT) with certain F-Airy structures on the loop space $V_+ = V_0[u]$. Secondly, we identify the action on the two theories.
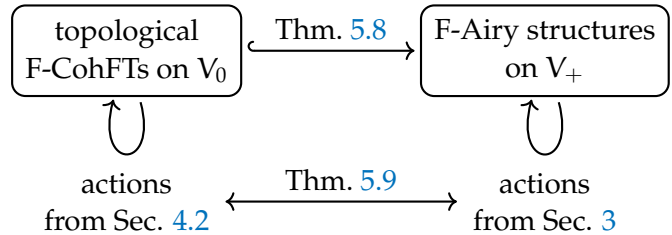
\begin{center}
\begin{tikzpicture}[
	baseline, thick, shorten >=2pt, shorten <=2pt,%
	box/.style = {draw,inner sep=5pt,rounded corners=5pt}%
]
	\node[box,align=center] (FCohFT) at (-2.75,0) {topological \\ F-CohFTs on $V_0$};
	\node[box,align=center] (FAiry) at (2.75,0) {F-Airy structures \\ on $V_+$};

	\draw[right hook->] (FCohFT) --node[midway,above] {Thm.~\ref{thm:top:F-CohFT:F-TR}} (FAiry);

	\path (FCohFT) edge[loop below] (FCohFT);
	\path (FAiry) edge[loop below] (FAiry);

	\node[align=center] (FCohFTaction) at (-2.75,-2) {actions \\ from Sec. \ref{subsec:F-Giv}};
	\node[align=center] (FAiryaction) at (2.75,-2) {actions \\ from Sec. \ref{sec:action}};

	\draw[<->] (FCohFTaction) --node[midway,below] {Thm.~\ref{thm:ident:actions}} (FAiryaction);
\end{tikzpicture}
\end{center}
The only actions we are going to consider at the level of F-CohFT are the changes of basis, R-actions and translations. Handling tick actions would bring us to the world of ``higher F-Airy structures'', in the same flavour as the higher Airy structures considered in \cite{BBCCN24}. The corresponding higher F-topological recursion would have higher-degree terms instead of just quadratic terms in \eqref{eq:F-TR:coord:free}. There is no difficulty in posing a definition of higher F-Airy structure and introduce actions on their defining tensors that would then implement the tick actions on F-CohFTs at the level of amplitudes. As this would demand another level of notational complexity without any surprise in the logic, we refrain from discussing it further.

%–––––––––––––––––––––––––––––––––––––––––––%
\paragraph{Topological F-CohFTs.}
%–––––––––––––––––––––––––––––––––––––––––––%
Let us start by defining topological F-CohFTs.

\begin{lem}
	Let $(V_0, \bcdot, w)$ be an F-TFT. The collection of maps
	\begin{equation}\label{eq:F-TFT:fund:class}
		\Omega_{g,1+n} \colon V_0^{\otimes n} \longrightarrow H^0(\Mbar_{g,1+n}) \otimes V_0 \,,
		\qquad
		v_1 \otimes \cdots \otimes v_n
		\longmapsto
		[1] \otimes \big( v_1 \bcdots v_n \bcdot w^g \big) \,,
	\end{equation}
	where $[1] \in H^0(\Mbar_{g,1+n})$ is the fundamental class, forms an F-CohFT on $V_0$. F-CohFTs of this type are called \emph{topological F-CohFT}. If the original F-TFT is unital, the associated F-CohFT has a flat unit.
\end{lem}

\begin{proof}
	The fact that the collection of maps $(\Omega_{g,1+n})_{g,n}$ forms an F-CohFT is a straightforward consequence of the compatibility of the fundamental class with glueing pullbacks and the commutativity/associativity of the F-TFT product. If the F-TFT is unital, the flat unit axiom follows from the fundamental class being compatible with the forgetful pullback.
\end{proof}

The amplitudes associated with the trivial CohFT (that is, the fundamental class) are recursively computed by topological recursion after Laplace transform. This is nothing but a reformulation of Witten's conjecture/Kontsevich's theorem \cite{Wit91,DVV91,Kon92} in terms of Virasoro constraints. As a consequence, we find that the amplitudes associated with a topological F-CohFT are again computed by F-topological recursion after Laplace transform. In order to state the precise result, preliminary considerations are due.

\begin{defn}
	First, we introduce a second loop space:
	\begin{equation}
		\Upsilon_+ \coloneqq \zeta V_0[\zeta^2] \,.
	\end{equation}
	The space $\Upsilon_+$ is related to $V_+$ through the \emph{Laplace isomorphism}:
	\begin{equation}
		\mathscr{L} \colon \Upsilon_+ \longrightarrow V_+ \,,
		\qquad\qquad
		\mathscr{L}[f](u)
		\coloneqq
		\frac{1}{\sqrt{2\pi u}} \int_{-\infty}^{+\infty}
			e^{-\frac{\zeta^2}{2u}} \,
			\dd f(\zeta)
		\,.
	\end{equation}
	Concretely, $\mathscr{L}^{-1}$ maps the basis vector $\mathrm{e}_{(\alpha,k)} \coloneqq \mathrm{e}_{\alpha} u^k$ to the basis vector $\epsilon_{(\alpha,k)} \coloneqq \mathrm{e}_{\alpha} \frac{\zeta^{2k+1}}{(2k+1)!!}$.

	We can also consider the `partial dual' space and the associated Laplace isomorphism:
	\begin{equation}
		\Upsilon_- \coloneqq V_0[\zeta^{-2}] \frac{\dd\zeta}{\zeta^2} \,,
		\qquad\qquad
		\mathscr{L}^{\ast} \colon V_- \longrightarrow \Upsilon_- \,.
	\end{equation}
	The duality in the loop variable is defined analogously through the formal residue. Concretely, the dual isomorphism maps the basis vector $\mathrm{e}_{\alpha}^k \coloneqq \mathrm{e}_{\alpha} \frac{\dd u}{u^{k+1}}$ to the basis vector $\epsilon_{\alpha}^k \coloneqq \mathrm{e}_{\alpha} \frac{(2k+1)!!}{\zeta^{2k+2}}\dd\zeta$. The Laplace isomorphisms also extend to isomorphisms involving $\widehat{V}_{\pm}$ and the completed loop spaces
	\begin{equation}
		\widehat{\Upsilon}_+
		\coloneqq
		\zeta V_0\bbraket{\zeta^2} \,,
		\qquad\qquad
		\widehat{\Upsilon}_-
		\coloneqq
		V_0\bbraket{\zeta^{-2}}\frac{\dd \zeta}{\zeta^2} \,.
	\end{equation}
\end{defn}

After applying the Laplace isomorphisms $(\mathscr{L}^{-1})^{\otimes 2}$ to $\mathscr{U}$ and $(\mathscr{L}^{\ast})^{\otimes 2}$ to $\mathscr{D}$, we obtain the elements $\Up \in \zeta_0 \zeta\,\End(V_0)[\zeta_0^2]\bbraket{\zeta^2}$ and $\Delta \in \End(V_0)[\zeta_0^{-2}]\bbraket{\zeta^{-2}}\frac{\dd\zeta_0\dd\zeta}{\zeta_0^2 \zeta^2}$. The non-degeneracy condition is equivalent to
\begin{equation}
	\Res_{\zeta,\zeta_1=0} \ \Up(\zeta_0,\zeta) \Delta(\zeta,\zeta_1) f(\zeta_1)
	=
	f(\zeta_0) \,,
	\qquad
	\Res_{\zeta,\zeta_1=0} \ \Delta(\zeta_0,\zeta) \Up(\zeta,\zeta_1) \chi(\zeta_1) 
	=
	\chi(\zeta_0)
\end{equation}
for any $\phi \in \Upsilon_+$ and $\chi \in \Upsilon_-$. Again, we can interpret both $\Up$ and $\Delta$ as linear operators:
\begin{equation}
\begin{aligned}
	& \Up = \mathscr{L}^{-1} \circ \mathscr{U} \circ (\mathscr{L}^{\ast})^{-1} \colon \Upsilon_- \longrightarrow \Upsilon_+ 
	&\qquad&
	\Up[\chi](\zeta_0)
	\coloneqq
	\Res_{\zeta = 0} \ \Up(\zeta_0,\zeta)\chi(\zeta) \,, \\
	& \Delta = \mathscr{L}^{\ast} \circ \mathscr{D} \circ \mathscr{L} \colon \Upsilon_+ \longrightarrow \Upsilon_-
	&\qquad&
	\Delta[f](\zeta_0)
	\coloneqq
	\Res_{\zeta = 0} \ \Delta(\zeta_0,\zeta) f(\zeta) \,,
\end{aligned}
\end{equation}
and the non-degeneracy condition simply states that $\Up$ and $\Delta$ are inverses of each other.

\begin{ex}\label{ex:stndrd:up:down:L}
	The standard up/down-morphisms of Example~\ref{ex:stndrd:up:down} yield, after application of the Laplace isomorphism
	\begin{equation}
	\begin{split}
		\Up(\zeta_0,\zeta)
			&=
			\id_{V_0} \sum_{k \geq 0} \frac{(\zeta_0\zeta)^{2k + 1}}{(2k + 1)!!^2}
			=
			\id_{V_0} \, \frac{1}{\zeta_0 \zeta} \biggl(
				{}_{1}F_2 \Bigl[
					\begin{smallmatrix} & \hspace*{-2pt}1\hspace*{-2pt} & \\[1pt] \frac{1}{2} && \frac{1}{2} \end{smallmatrix}
				\Bigr] \bigl( \tfrac{\zeta_0^{2}\zeta^{2}}{4} \bigr)
				-1
			\biggr) \,, \\
			%\tfrac{\pi}{2} \, \bm{L}_0(\zeta_0\zeta) \,, \\
		\Delta(\zeta_0,\zeta)
			&=
			\id_{V_0} \sum_{k \geq 0} \frac{(2k + 1)!!^2}{(\zeta_0 \zeta)^{2k + 2}} \, \dd\zeta_0 \, \dd\zeta
			=
			\id_{V_0} \,
			\dd_{\zeta_0} \, \dd_{\zeta}
			\biggl(
				\frac{1}{\zeta_0 \zeta} \,
				{}_{3}F_0 \Bigl[
					\begin{smallmatrix} \frac{1}{2} & \frac{1}{2} & 1 \\[1pt] & \emptyset & \end{smallmatrix}
				\Bigr] \bigl( 4\zeta_0^{-2}\zeta^{-2} \bigr)
			\biggr) \,.
	\end{split}
	\end{equation}
	Here ${}_{p}F_q$ is a generalised hypergeometric series. The up-morphism can be alternatively written as $\Up(\zeta_0,\zeta) = \id_{V_0} \, \tfrac{\pi}{2} \, \bm{L}_0(\zeta_0\zeta)$, where $\bm{L}_0$ is the modified Struve function of order $0$. In coordinates, it yields $\Up[\epsilon_{\alpha}^{k}] = \epsilon_{(\alpha,k)}$ and $\Delta[\epsilon_{(\alpha,k)}] = \epsilon_{\alpha}^{k}$.
\end{ex}

\begin{thm}\label{thm:top:F-CohFT:F-TR}
	Let $(V_0, \bcdot, w)$ be an F-TFT. For a fixed choice of up/down-morphisms $(\mathscr{U},\mathscr{D})$, denote by $(F_{g,1+n})_{g,n}$ the associated amplitudes on $V_+$. Their Laplace transform, that is
	\begin{equation}
		\mathscr{L}^{-1} \circ F_{g,1+n} \circ \mathscr{L}^{\otimes n} \in \Hom(\Sym{n}{\Upsilon_+},\Upsilon_+) \,,
	\end{equation}
	are computed by F-topological recursion from the following F-Airy structure:
	\begin{equation}\label{eq:FAiry:loop}
	\begin{aligned}
		&
		A \in \Hom(\Sym{2}{\Upsilon_+}, \Upsilon_+)
		&& \qquad
		A(f_1 \otimes f_2)
		=
		\overline{\Up}\bigl[
			\dd f_1 \bcdot_{\scriptscriptstyle 0} \dd f_2
		\bigr]
		\in \Upsilon_+
		\,, \\[1ex]
		&
		B \in \Hom(\Upsilon_+^{\otimes 2}, \Upsilon_+) 
		&& \qquad
		B(f_1 \otimes f_2)
		=
		\overline{\Up}\bigl[
			\dd f_1 \bcdot_{\scriptscriptstyle 0} \Delta f_2
		\bigr]
		\in \Upsilon_+
		\,, \\[1ex]
		&
		\conn{C} \in \Hom(\Upsilon_+, \Upsilon_+^{\otimes 2}) 
		&& \qquad
		\conn{C}(f)
		=
		\bigl( (\overline{\Up} \otimes \id_{\Upsilon_+}) \circ \kappa_{\scriptscriptstyle 0} \bigr) \bigl[
			\Delta f
		\bigr]
		\in \Upsilon_+\, \widehat{\otimes} \,\Upsilon_+
		\,, \\[1ex]
		&
		\disc{C} \in \Hom(\Sym{2}{\Upsilon_+}, \Upsilon_+) 
		&& \qquad
		\disc{C}(f_1 \otimes f_2)
		=
		\overline{\Up}\bigl[
			\Delta f_1 \bcdot_{\scriptscriptstyle 0} \Delta f_2
		\bigr]
		\in \Upsilon_+
		\,, \\[1ex]
		&
		D \in \Upsilon_+
		&& \qquad
		D = 
		\tfrac{1}{2} \, \overline{\Up}\bigl[ \varpi_{\scriptscriptstyle 0} \bigr]
		\in \Upsilon_+
		\,. 
	\end{aligned}
	\end{equation}
	Here the following notations/conventions have been used.
	\begin{itemize}
		\item The linear map $\overline{\Up}$ is the extension of $\Up$ to $\Upsilon \coloneqq V_0(\!( \zeta^2 )\!) \dd\zeta = \dd \widehat{\Upsilon}_+ \oplus \Upsilon_-$ which is zero on $\dd\widehat{\Upsilon}_+$.

		\item The product $\bcdot_{\scriptscriptstyle 0}$ is the one induced by the F-TFT on $V_0$-valued $1$-forms (that is, the product $\bcdot$ on $V_0$ and the usual product on $1$-forms), twisted by $\theta_{\scriptscriptstyle 0} \coloneqq \frac{1}{\zeta^2 \dd\zeta}$:
		\begin{equation}
			\chi_1(\zeta)
			\bcdot_{\scriptscriptstyle 0}
			\chi_2(\zeta)
			\coloneqq
			\bigl( \chi_1(\zeta) \bcdot \chi_2(\zeta) \bigr) \theta_{\scriptscriptstyle 0}(\zeta) \,.
		\end{equation}

		\item The map $\kappa_{\scriptscriptstyle 0} \colon \Upsilon_- \to \Upsilon_-\, \widehat{\otimes}\, \Upsilon_+$ is defined as
		\begin{equation}
			\chi(\zeta)
			\longmapsto
			\sum_{k \ge 0}
				\left(
					\chi(\zeta_1) \biggl( \theta_{\scriptscriptstyle 0}(\zeta_1) \frac{\dd\zeta_1}{\zeta_1^{2k+2}} \biggr)
				\right)
				\otimes
				\left(
					w \, \zeta_2^{2k+1}
				\right)
		\end{equation}
		In other words, $\kappa_{\scriptscriptstyle 0}$ is the multiplication of $\chi(\zeta_1)\theta_{\scriptscriptstyle 0}(\zeta_1) \otimes w$ by $\frac{\zeta_2 \dd\zeta_1}{\zeta_1^2 - \zeta_2^2}$, expanded in geometric series in the regime $|\zeta_1| > |\zeta_2|$.

		\item $\varpi_{\scriptscriptstyle 0} \coloneqq w \, \theta_{\scriptscriptstyle 0}(\zeta) \, \frac{(\dd\zeta)^2}{(2\zeta)^2} \in \Upsilon_-$.
	\end{itemize}
	In coordinates, for a fixed basis $(\mathrm{e}_{(\alpha,k)} = \mathrm{e}_{\alpha}u^k)_{\alpha,k}$ of $V_+$ with structure constants $\mathrm{e}_{\beta} \bcdot \mathrm{e}_{\gamma} = P^{\alpha}_{\beta,\gamma} \mathrm{e}_{\alpha}$ and distinguished vector $w = w^{\alpha} \mathrm{e}_{\alpha}$, we have
	\begin{equation}\label{eq:FAiry:loop:coord}
	\begin{split}
		& A_{(\beta,j),(\gamma,k)}^{(\alpha,i)}
		=
			\mathscr{U}_{\lambda}^{\alpha;i,0} \,
			P_{\beta,\gamma}^{\lambda} \,
			\delta_{j,k,0} \,, \\
		& B_{(\beta,j),(\gamma,k)}^{(\alpha,i)}
		=
			\mathscr{U}^{\alpha;i,\ell}_{\lambda} \,
			P_{\beta,\mu}^{\lambda} \,
			\mathscr{D}^{\mu}_{\gamma;m,k} \,
			\delta^{m+1}_{\ell+j} \,
			\frac{(2m+1)!!}{(2\ell+1)!!(2j-1)!!} \,
			\,, \\
		& \indC{\conn{C}}{(\alpha,i),(\beta,j)}{(\gamma,k)}
		=
			\mathscr{U}^{\alpha;i,\ell}_{\lambda} \,
			w^{\beta} \,
			\mathscr{D}^{\lambda}_{\gamma;m,k} \,
			\delta_{\ell}^{m+j+2} \,
			\frac{(2m+1)!! (2j+1)!!}{(2\ell + 1)!!}
		\,, \\
		& \indC{\disc{C}}{(\alpha,i)}{(\beta,j),(\gamma,k)}
		=
			\mathscr{U}^{\alpha;i,\ell}_{\lambda} \,
			P_{\rho,\sigma}^{\lambda} \,
			\mathscr{D}^{\rho}_{\beta;r,j} \,
			\mathscr{D}^{\sigma}_{\gamma;s,k} \,
			\delta_{\ell}^{r+s+2} \,
			\frac{(2r+1)!! (2s+1)!!}{(2\ell + 1)!!}
		\,, \\
		& D^{(\alpha,i)}
		=
			\mathscr{U}_{\lambda}^{\alpha;i,1} \,
			\frac{w^{\lambda}}{24} \,.
	\end{split}
	\end{equation}
\end{thm}

\begin{proof}
	We proceed by induction on $2g-2+(1+n) > 0$, after setting up some notations. In coordinates, denote the F-CohFT associated with $(V_0,\bcdot,w)$ defined in equation~\eqref{eq:F-TFT:fund:class} as
	\begin{equation}
		\Omega_{g,1+n}(\mathrm{e}_{\alpha_1} \otimes \cdots \otimes \mathrm{e}_{\alpha_n})
		=
		[1] \otimes \mathrm{e}_{\alpha_1} \bcdots \mathrm{e}_{\alpha_n} \bcdot w^g
		=
		\mathscr{F}_{g;\alpha_1,\dots,\alpha_n}^{\alpha_0} [1] \otimes \mathrm{e}_{\alpha_0} \,.
	\end{equation}
	Thus, the associated F-CohFT amplitudes read
	\begin{equation}\label{eq:top:FCohFT:ampl}
		F_{g;(\alpha_1,k_1),\dots,(\alpha_n,k_n)}^{(\alpha_0,k_0)}
		=
		\mathscr{U}_{\beta}^{\alpha_0;k_0,\ell} \,
		\mathscr{F}_{g;\alpha_1,\dots,\alpha_n}^{\beta}
		\braket{\tau_{\ell} \tau_{k_1} \cdots \tau_{k_n} }_g \,,
	\end{equation}
	where $\braket{\tau_{\ell} \tau_{k_1} \cdots \tau_{k_n} }_g \coloneqq \int_{\Mbar_{g,1+n}} \psi_0^{\ell} \psi^{k_1}_1 \cdots \psi^{k_n}_n$ is Witten's notation for $\psi$-class intersection numbers. Using $\braket{\tau_i \tau_j \tau_k}_0 = \delta_{i,j,k,0}$ and $\braket{\tau_i}_1 = \delta_{i,1}/24$, we deduce the thesis for the basic topologies. In order to compute recursively the F-CohFT amplitudes, we notice that the amplitudes in \eqref{eq:top:FCohFT:ampl} `decouple' as F-TFT correlators multiplied by Witten's correlators. We can then employ the Virasoro constraints for $\psi$-class intersection numbers
	\begin{multline}\label{eq:Virasoro}
		\braket{ \tau_{\ell} \tau_{k_1} \cdots \tau_{k_n} }_g
		=
		\sum_{m=1}^n
			\frac{(2(\ell+k_m-1)+1)!!}{(2\ell+1)!! (2k_m-1)!!}
			\braket{ \tau_{\ell+k_m-1} \tau_{k_1} \cdots \widehat{\tau_{k_m}} \cdots \tau_{k_n} }_g
			\\
		+
		\frac{1}{2} \sum_{\substack{a,a' \ge 0 \\ a + a' = \ell-2}} \frac{(2a+1)!! (2a'+1)!!}{(2\ell+1)!!} \Bigg(
			\braket{ \tau_{a} \tau_{a'} \tau_{k_1} \cdots \tau_{k_n} }_{g-1}
			\\
			+
			\sum_{ \substack{h+h' = g \\ J \sqcup J' = [n] }}
				\braket{ \tau_{a} \tau_{K_{J}} }_{h}
				\braket{ \tau_{a'} \tau_{K_{J'}} }_{h'}
		\Bigg)
	\end{multline}
	together with the recursive structure of the F-TFT amplitudes:
	\begin{equation}\label{eq:FTFT:rec}
	\begin{aligned}
		& \mathscr{F}_{g;\alpha_1,\dots,\alpha_n}^{\lambda}
		=
		P_{\mu,\alpha_m}^{\lambda} \, \mathscr{F}_{g;\alpha_1,\dots,\widehat{\alpha_m},\dots,\alpha_n}^{\mu}
		&& \qquad m \in [n] \,, \\
		& \mathscr{F}_{g;\alpha_1,\dots,\alpha_n}^{\lambda}
		=
		w^{\mu} \, \mathscr{F}_{g-1;\mu,\alpha_1,\dots,\alpha_n}^{\lambda \vphantom{\mu'}} \,,
		&& \\
		& \mathscr{F}_{g;\alpha_1,\dots,\alpha_n}^{\lambda}
		=
		P_{\mu,\mu'}^{\lambda} \, \mathscr{F}_{h;\alpha_J}^{\mu} \, \mathscr{F}_{h';\alpha_{J'}}^{\mu'}
		&& \qquad h + h' = g , \ J \sqcup J' = [n] \,.
	\end{aligned}
	\end{equation}
	Coupling the three different terms in the Virasoro recursion (corresponding to the three lines in \eqref{eq:Virasoro}) with the three different relations satisfied by the F-TFT amplitudes (corresponding to the three lines in \eqref{eq:FTFT:rec}), we deduce that $F_{g,1+n}$ is indeed computed by F-topological recursion with data given by \eqref{eq:FAiry:loop:coord}. For instance, for the B-term we find
	\begin{equation}
	\begin{split}
		&
		\mathscr{U}_{\lambda}^{\alpha_0;k_0,\ell} \,
		P_{\mu,\alpha_m}^{\lambda} \,
		\delta^{p+1}_{\ell+k_m} \,
		\frac{(2p+1)!!}{(2\ell+1)!! (2k_m-1)!!} \,
		\mathscr{F}_{g;\alpha_1,\dots,\widehat{\alpha_m},\dots,\alpha_n}^{\mu}
		\braket{ \tau_{p} \tau_{k_1} \cdots \widehat{\tau_{k_m}} \cdots \tau_{k_n} }_g \\
		& \quad =
		\mathscr{U}_{\lambda}^{\alpha_0;k_0,\ell} \,
		P_{\mu,\alpha_m}^{\lambda} \,
		\delta^{p+1}_{\ell+k_m} \,
		\frac{(2p+1)!!}{(2\ell+1)!! (2k_m-1)!!} \,
		\delta^{\mu}_{\nu} \, \delta_{p}^{q} \,
		\mathscr{F}_{g;\alpha_1,\dots,\widehat{\alpha_m},\dots,\alpha_n}^{\nu}
		\braket{ \tau_{q} \tau_{k_1} \cdots \widehat{\tau_{k_m}} \cdots \tau_{k_n} }_g \\
		& \quad =
		\underbrace{
			\mathscr{U}_{\lambda}^{\alpha_0;k_0,\ell} \,
			P_{\mu,\alpha_m}^{\lambda} \,
			\mathscr{D}_{\beta;p,j}^{\mu} \,
			\delta^{p+1}_{\ell+k_m} \,
			\frac{(2p+1)!!}{(2\ell+1)!! (2k_m-1)!!}
		}_{
			= B^{(\alpha_0,k_0)}_{(\beta,j),(\alpha_m,k_m)}
		} \\
		&
		\qquad\qquad\qquad \times
		\underbrace{
			\mathscr{U}_{\nu}^{\beta;j,q} \,
			\mathscr{F}_{g;\alpha_1,\dots,\widehat{\alpha_m},\dots,\alpha_n}^{\nu}
			\braket{ \tau_{q} \tau_{k_1} \cdots \widehat{\tau_{k_m}} \cdots \tau_{k_n} }_g
			\vphantom{\frac{(2p+1)!!}{(2\ell+1)!! (2k_m-1)!!}}
		}_{
			= \indF{F}{g}{(\beta,j)}{(\alpha_1,k_1),\dots,\widehat{(\alpha_m,k_m)},\dots,(\alpha_n,k_n)}
		} .
	\end{split}
	\end{equation}
	Notice that a crucial role is played by the non-degeneracy conditions $\mathscr{D}_{\beta;p,j}^{\mu} \mathscr{U}_{\nu}^{\beta;j,q} = \delta^{\mu}_{\nu} \, \delta_{p}^{q}$. It is now easy to check that the expression for $B$ in \eqref{eq:FAiry:loop} is the coordinate-free versions of \eqref{eq:FAiry:loop:coord}:
	\begin{equation}
	\begin{split}
		B\bigl( \epsilon_{(\beta,j)} \otimes \epsilon_{(\gamma,k)} \bigr)
		&=
		\overline{\Up}\bigl[
			\bigl( \dd\epsilon_{(\beta,j)} \bcdot \Delta\epsilon_{(\gamma,k)} \bigr) \theta_{\scriptscriptstyle 0}
		\bigr] \\
		&=
		\overline{\Up}\Biggl[
			\left(
				\mathrm{e}_{\beta} \frac{\zeta^{2j}}{(2j-1)!!} \dd\zeta
				\bcdot
				\mathscr{D}^{\mu}_{\gamma;m,k} \, \mathrm{e}_{\mu} \, \frac{(2m+1)!!}{\zeta^{2m+2}} \dd\zeta
			\right)
			\frac{1}{\zeta^2 \dd\zeta}
		\Biggr] \\
		&=
		P^{\lambda}_{\beta,\mu} \, \mathscr{D}^{\mu}_{\gamma;m,k} \, \frac{(2m+1)!!}{(2j-1)!!} \,
		\overline{\Up}\Biggl[
			\mathrm{e}_{\lambda} \frac{\dd \zeta}{\zeta^{2(m-j+1)+2}}
		\Biggr] \\
		&=
		\mathscr{U}^{\alpha;i,\ell}_{\lambda} \,
		P_{\beta,\mu}^{\lambda} \,
		\mathscr{D}^{\mu}_{\gamma;m,k} \,
		\delta^{m+1}_{\ell+j} \,
		\frac{(2m+1)!!}{(2\ell+1)!!(2j-1)!!} \,
		\epsilon_{(\alpha,i)} \,.
	\end{split}
	\end{equation}
	A similar computation can be carried out for $\conn{C}$ and $\disc{C}$, thus completing the proof.
\end{proof}

%–––––––––––––––––––––––––––––––––––––––––––%
\paragraph{Identification of the orbits.}
%–––––––––––––––––––––––––––––––––––––––––––%
We are now ready to state the main result of this section: the identification of the actions on F-Airy structures and F-CohFTs described in Sections~\ref{sec:action} and \ref{subsec:F-Giv} respectively. The proof is a simple consequence of the analysis carried out in Section~\ref{subsec:actions:FCohFT:ampl}.

\begin{thm}\label{thm:ident:actions}
	Let $\Omega$ be an F-CohFT on $V_0$. Suppose that, after a choice of up/down-morphisms $(\mathscr{U}, \mathscr{D})$, the associated amplitudes $(F_{g,1+n})_{g,n}$ are computed by F-TR on $V_+$.
		\begin{itemize}
		\item \emph{Change of basis}. For a given $L \in \GL(V_0)$, the amplitudes associated with the F-CohFT $\hat{L}\Omega$
		are computed by F-TR and coincide with $\preind{L}{F}$ for
		\begin{equation}
			L_{\textup{t}} = L_{\textup{s}} \coloneqq L \in \GL(V_+) \,,
		\end{equation}
		provided that the transformed amplitudes are computed with respect to the up/down-morphisms
		\begin{equation}
			\preind{L}{\mathscr{U}}
			\coloneqq
			L \circ \mathscr{U} \circ L^{-1}
			\qquad\text{and}\qquad
			\preind{L}{\mathscr{D}}
			\coloneqq
			L \circ \mathscr{D} \circ L^{-1} \,.
		\end{equation}

		\item \emph{R-action}. For a given $R(u) \in \mathfrak{Giv}$, the amplitudes associated with the F-CohFT $\hat{R}\Omega$ are computed by F-TR  (with underlying F-Airy structure based on $V_{R,+} \coloneqq L_{R,\textup{t}}(V_+) \subset \widehat{V}_+$, see \ref{subsec:actions:FCohFT:ampl}) and coincide with $\preind{L_{R}}{\vphantom{F}}(\preind{B_{R}}{F})$ for
		\begin{equation}
		\begin{aligned}
			\qquad
			& L_{{R},\textup{s}} \in \GL(\widehat{V}_+)
			& \quad &
			L_{{R},\textup{s}}[f](u)
			\coloneqq
			R(u)f(u) \,, \\
			\qquad
			& L_{{R},\textup{t}} \in \GL(\widehat{V}_+)
			& \quad &
			L_{{R},\textup{t}}[f](u)
			\coloneqq
			R(-u)f(u) \,, \\
			\qquad
			& B_{{R}} \in \Hom(V_+,\widehat{V}_+)
			& \quad & 
			B_{{R}}[f](u)
			\coloneqq
			\Res_{u',u'' = 0} \
				\frac{\id_{V_0} - R^{-1}(u) \circ R(-u')}{u + u'} \,
				\mathscr{D}(u',u'') \, f(u'') \,,
		\end{aligned}
		\end{equation}
		provided that the transformed amplitudes are computed with respect to the up/down-morphisms
		\begin{equation}
			\preind{R}{\mathscr{U}} 
			\coloneqq
			L_{{R},\textup{t}} \circ \mathscr{U} \circ M_R^{-1} \,,
			\qquad\text{and}\qquad
			\preind{R}{\mathscr{D}} 
			\coloneqq
			M_R \circ \mathscr{D} \circ L_{{R},\textup{t}}^{-1} \,,
		\end{equation}
		where $L_{{R},\textup{t}}$ is the multiplication by $R(-u)$ as above, its inverse $L_{{R},\textup{t}}^{-1}$ is the multiplication by $R^{-1}(-u)$, and $M_R, M_R^{-1} \in \GL(V_-)$ are defined as
		\begin{equation}
			M_{R}[\chi](u) \coloneqq \Bigl[ R(-u) \chi(u) \Bigr]_- \,,
			\qquad
			M_{R}^{-1}[\chi](u) \coloneqq \Bigl[ R^{-1}(-u) \chi(u) \Bigr]_- \,.
		\end{equation}

		\item \emph{Translation}. For a given $T(u) \in u^2 V_0\bbraket{u} \subset \widehat{V}_+$, the amplitudes associated with the F-CohFT $\hat{T}\Omega$ are computed by F-TR and coincide with the amplitudes $\preind{T}{F}$ (cf. Remark~\ref{rem:translation}).
	\end{itemize}
\end{thm}

Keeping in mind possible applications, let us consider the case of F-CohFTs of the form
\begin{equation}
	\Omega = \hat{L}\hat{R}\hat{T}\Omega^0 \,,
\end{equation}
where $\Omega^0$ is a topological F-CohFT, $T(u) \in u^2V\bbraket{u}$, $R(u) \in \mathfrak{Giv}$, and $L \in \GL(V_0)$. Our goal is to write down the initial data $(A,B,\conn{C},\disc{C},D)$ explicitly in terms of the F-TFT structure and the data of $T$, $R$, and $L$. We proceed step by step, modifying the initial data for $\Omega^0$ accordingly.

\textit{Step 1.} Theorem~\ref{thm:top:F-CohFT:F-TR} provides the initial data for $\Omega^0$ in terms of the associated F-TFT data.

\textit{Step 2.} As for $\hat{T}\Omega^0$, the transformation of the initial data is provided by Theorem~\ref{thm:transl:FAiry}, and the formulae can be simplified as follows. Firstly notice that, as $T$ starts in degree $2$, for cohomological degree reasons the tensors $\preind{T}{G}$ and $\preind{T}{H}$ identically vanish. Hence, the equations defining the translated initial data drastically simplify. We claim that, upon identification of the underlying loop spaces via Laplace isomorphisms, the translated initial data are given by
\begin{equation}\label{eq:FAiry:loop:transl}
\begin{aligned}
	&
	\preind{T}{A}(f_1 \otimes f_2)
	=
	\overline{\Up}\bigl[
		\dd f_1 \bcdot_{\scriptscriptstyle T} \dd f_2
	\bigr]
	\in \Upsilon_+
	\,, \\[1ex]
	&
	\preind{T}{B}(f_1 \otimes f_2)
	=
	\overline{\Up}\bigl[
		\dd f_1 \bcdot_{\scriptscriptstyle T} \Delta f_2
	\bigr]
	\in \Upsilon_+
	\,, \\[1ex]
	&
	\preind{T}{\vphantom{C}}\conn{C}(f)
	=
	\bigl( (\overline{\Up} \otimes \id_{\Upsilon_+}) \circ \kappa_{\scriptscriptstyle T} \bigr) \bigl[
			\Delta f
	\bigr]
	\in \Upsilon_+\, \widehat{\otimes}\,\Upsilon_+
	\,, \\[1ex]
	&
	\preind{T}{\vphantom{C}}\disc{C}(f_1 \otimes f_2)
	=
	\overline{\Up}\bigl[
		\Delta f_1 \bcdot_{\scriptscriptstyle T} \Delta f_2
	\bigr]
	\in \Upsilon_+
	\,, \\[1ex]
	&
	\preind{T}{D} = 
	\tfrac{1}{2} \, \overline{\Up}\bigl[ \varpi_{T} \bigr]
	\in \Upsilon_+
	\,,
\end{aligned}
\end{equation}
where now the following ($T$-dependent) notations/conventions have been introduced.
\begin{itemize}
	\item The product $\bcdot_{\scriptscriptstyle T}$ is now twisted using $\tau \coloneqq \mathscr{L}^{-1}[T]$. Namely, we introduce
	\begin{equation}
		\theta_{\scriptscriptstyle T}
		\coloneqq
		\frac{1}{\zeta^2 \dd\zeta- \dd\tau}
		=
		\frac{1}{\zeta^2\dd\zeta} \left(
			1 + \sum_{m \geq 1} \left( \frac{\dd \tau}{\zeta^2\dd \zeta} \right)^m
		\right) ,
	\end{equation}
	where the powers live in the algebra of $V_0$-valued formal power series. The twisted product is then defined as
	\begin{equation}
		\chi_1(\zeta)
		\bcdot_{\scriptscriptstyle T}
		\chi_2(\zeta)
		\coloneqq
		\bigl( \chi_1(\zeta) \bcdot \chi_2(\zeta) \bigr) \bcdot \theta_{\scriptscriptstyle T}(\zeta) \,.
	\end{equation}
	Strictly speaking, the $m = 0$ term does not live in the same space as the series over $m \geq 1$, but the above formula still makes sense if interpreted as
	\begin{equation}
		\chi_1(\zeta)
		\bcdot_{\scriptscriptstyle T}
		\chi_2(\zeta)
		=
		\bigl( \chi_1(\zeta) \bcdot \chi_2(\zeta) \bigr) \frac{1}{\zeta^2 \dd \zeta}
		+
		\sum_{m \geq 1} \chi_1(\zeta) \bcdot \chi_2(\zeta) \bcdot \left(\frac{\dd \tau}{\zeta^2\dd \zeta}\right)^m \frac{1}{\zeta^2\dd \zeta} \,.
	\end{equation}
	If the F-TFT comes with a unit $\mathrm{e}$, one can safely substitute the $m = 0$ term with $\mathrm{e} \frac{1}{\zeta^2 \dd\zeta}$.

	\item The map $\kappa_{\scriptscriptstyle T} \colon \Upsilon_- \to \Upsilon \,\widehat{\otimes}\, \Upsilon_+$ is defined as
	\begin{equation}
		\chi(\zeta)
		\longmapsto
		\sum_{k \ge 0}
			\left(
				\chi(\zeta_1) \bcdot \theta_{\scriptscriptstyle T}(\zeta_1) \biggl( \frac{\dd\zeta_1}{\zeta_1^{2k+2}} \biggr)
			\right)
			\otimes
			\bigl(
				w \, \zeta_2^{2k+1}
			\bigr) \,.
	\end{equation}

	\item $\varpi_{T} \coloneqq w \bcdot \theta_{\scriptscriptstyle T}(\zeta) \, \frac{(\dd\zeta)^2}{(2\zeta)^2} \in \Upsilon$.
\end{itemize}
Notice that the result of the $T$-twisted product $\chi_1 \bcdot_{\scriptscriptstyle T} \chi_2$, the first tensor factor of $\kappa_{\scriptscriptstyle T}[\chi]$, and $\varpi_{\scriptscriptstyle T}$ belong to $\Upsilon = \dd\widehat{\Upsilon}_+ \oplus \Upsilon_-$ rather than simply $\Upsilon_-$. This is due to the fact that $\theta_{\scriptscriptstyle T}(\zeta)$ contains (arbitrarily large) positive powers of $\zeta$. However, this is not a problem: the application of $\overline{\Up}$ in \eqref{eq:FAiry:loop:transl} annihilates all the terms from $\dd\widehat{\Upsilon}_+$, providing well-defined elements of $\widehat{\Upsilon}_+$ at the end of the computation.

Let us check \eqref{eq:FAiry:loop:transl} for the $B$-tensor. The translated initial data are uniquely characterised by equation~\eqref{eq:BC:tilde}, which in the case of vanishing $\preind{T}{G}$ and $\preind{T}{H}$ simplifies to $B = K \circ \preind{T}{B}$ with $K \coloneqq \id_{V_+} - B(\tau \otimes \id_{V_+}) \in \End(V_+)$. Notice that the definition of $K$ involves the inverse Laplace-transformed translation, since all computations are performed on $\Upsilon_+$ rather than $V_+$. We can now check that the translated initial data indeed satisfy the equation:
\begin{equation}\label{eq:transl:B:check}
\begin{split}
	(K \circ \preind{T}{B})(f_1 \otimes f_2)
	& =
	\preind{T}{B}(f_1 \otimes f_2)
	-
	B( \tau \otimes \preind{T}{B}(f_1 \otimes f_2) ) \\
	& =
	\overline{\Up}\bigl[
		\dd f_1\bcdot \Delta f_2 \bcdot \theta_{\scriptscriptstyle T}
	\bigr]
	-
	\overline{\Up}\Bigl[
		\theta_{\scriptscriptstyle 0} \, \dd {\tau} \bcdot (\Delta \circ \overline{\Up}) \bigl[ \dd f_1 \bcdot \Delta f_2 \bcdot \theta_{\scriptscriptstyle T} \bigr]
	\Bigr] \\
	& =
	\overline{\Up}\Bigl[
		\bigl( \theta_{\scriptscriptstyle T} - \theta_{\scriptscriptstyle 0} \, \dd {\tau} \bcdot \theta_{\scriptscriptstyle T} \bigr) \bcdot \dd f_1 \bcdot \Delta f_2
	\Bigr] .
\end{split}
\end{equation}
To go from the second to the last line, we recall that $\overline{\Up}$ is the operator $\Up$ extended by zero on $\dd\widehat{\Upsilon}_+$. We then decompose
\begin{equation}
	\chi = \dd f_1 \bcdot \Delta f_2 \bcdot \theta_{\scriptscriptstyle T}
	=
	\chi_- + \dd \chi_+
	\qquad\text{with}\qquad
	\chi_{-} \in \Upsilon_{-}\,, \ \chi_+ \in \widehat{\Upsilon}_+
\end{equation}
and employ the non-degeneracy condition $\Delta \circ \Up = \id_{\Upsilon_-}$ to get $\Delta \circ \overline{\Up}[\chi] = \chi_- = \chi - \dd\chi_+$. Nevertheless, since $T(u) = \bigO(u^2)$, we have $\theta_{\scriptscriptstyle 0} \, \dd\tau = \bigO(\zeta^2)$, so that $\theta_{\scriptscriptstyle 0} \, \dd\tau \bcdot \dd\chi_+ \in \dd\widehat{\Upsilon}_+$ is annihilated by the outermost $\overline{\Up}$ and yields the last line of \eqref{eq:transl:B:check}. Now, recalling that $\theta_{\scriptscriptstyle 0} = (\zeta^2 \dd\zeta)^{-1}$, we find $(\theta_{\scriptscriptstyle T} - \theta_{\scriptscriptstyle 0} \, \dd\tau \bcdot \theta_{\scriptscriptstyle T}) = \theta_{\scriptscriptstyle 0}$, hence the claim $K \circ \preind{T}{B} = B$. Similar computations hold for the other tensors.

\textit{Step 3.} For $\hat{R}\hat{T}\Omega^0$, ignoring the change of bases induced by the R-action (see Step 4), the transformation of the initial data is provided by Theorem~\ref{thm:Bglbv:FAiry} and reads (upon identification of the underlying loop spaces via Laplace isomorphisms)
\begin{equation}\label{eq:FAiry:loop:RT}
\begin{aligned}
	&
	\preind{RT}{A}(f_1 \otimes f_2)
	=
	\overline{\Up}\bigl[
		\dd f_1 \bcdot_{\scriptscriptstyle T} \dd f_2
	\bigr]
	\in \Upsilon_+
	\,, \\[1ex]
	&
	\preind{RT}{B}(f_1 \otimes f_2)
	=
	\overline{\Up}\bigl[
		\dd f_1
		\bcdot_{\scriptscriptstyle T}
		\bigl(
			(\id_{\Upsilon_-} + \dd \circ \Epsilon_{R}) \circ \Delta
		\bigr) f_2
	\bigr]
	\in \Upsilon_+
	\,, \\[1ex]
	&
	\preind{RT}{\vphantom{C}}\conn{C}(f)
	=
	\bigl( (\overline{\Up} \otimes \id_{\Upsilon_+}) \circ \kappa_{\scriptscriptstyle T} \bigr) \bigl[
			\Delta f
	\bigr]
	\in \Upsilon_+ \,\widehat{\otimes}\,\Upsilon_+
	\,, \\[1ex]
	&
	\preind{RT}{\vphantom{C}}\disc{C}(f_1 \otimes f_2)
	=
	\overline{\Up}\bigl[
		\bigl(
			(\id_{\Upsilon_-} + \dd \circ \Epsilon_{R}) \circ \Delta
		\bigr) f_1
		\bcdot_{\scriptscriptstyle T}
		\bigl(
			(\id_{\Upsilon_-} + \dd \circ \Epsilon_{R}) \circ \Delta
		\bigr) f_2
	\bigr]
	\in \Upsilon_+
	\,, \\[1ex]
	&
	\preind{RT}{D} = 
	\tfrac{1}{2} \, \overline{\Up}\bigl[ \varpi_{T} \bigr]
	\in \Upsilon_+
	\,,
\end{aligned}
\end{equation}
where $\Epsilon_{R} \colon \Upsilon_- \to \Upsilon_+$ is the Laplace transform of the edge weight operator $\mathscr{E}_{R} \colon V_- \to V_+$ defined in terms of the R-matrix in equation~\eqref{eq:Bglbv:F-Giv}. In other words, $\Epsilon_{R} \coloneqq \mathscr{L}^{-1} \circ \mathscr{E}_{R} \circ (\mathscr{L}^{\ast})^{-1}$.

\textit{Step 4.} To conclude, following Section~\ref{subsec:change:bases}, the change of bases induced by $R$ and $L$ giving the amplitudes associated with $\hat{L}\hat{R}\hat{T}\Omega^0$ (computed with respect to the properly modified up/down-morphisms) produces the following transformed initial data (again, upon identification of the underlying loop spaces via Laplace isomorphisms):
\begin{equation}\label{eq:FAiry:loop:LRT}
\begin{aligned}
	&
	\preind{LRT}{A}(f_1 \otimes f_2)
	=
	\overline{\Up}_{LR}\bigl[
		\dd_{LR} f_1 \bcdot_{\scriptscriptstyle T} \dd_{LR} f_2
	\bigr]
	\in \Upsilon_+
	\,, \\[1ex]
	&
	\preind{LRT}{B}(f_1 \otimes f_2)
	=
	\overline{\Up}_{LR}\bigl[
		\dd_{LR} f_1
		\bcdot_{\scriptscriptstyle T}
		\bigl(
			(\id_{\Upsilon_-} + \dd \circ \Epsilon_{R}) \circ \Delta_{LR}
		\bigr) f_2
	\bigr]
	\in \Upsilon_+
	\,, \\[1ex]
	&
	\preind{LRT}{\vphantom{C}}\conn{C}(f)
	=
	\bigl(
		(\overline{\Up}_{LR} \otimes \lambda_{{LR},\textup{t}})
		\circ
		\kappa_{\scriptscriptstyle T}
	\bigr) \bigl[
			\Delta_{LR} f
	\bigr]
	\in \Upsilon_+ \,\widehat{\otimes}\,\Upsilon_+
	\,, \\[1ex]
	&
	\preind{LRT}{\vphantom{C}}\disc{C}(f_1 \otimes f_2)
	=
	\overline{\Up}_{LR}\bigl[
		\bigl(
			(\id_{\Upsilon_-} + \dd \circ \Epsilon_{R}) \circ \Delta_{LR}
		\bigr) f_1
		\bcdot_{\scriptscriptstyle T}
		\bigl(
			(\id_{\Upsilon_-} + \dd \circ \Epsilon_{R}) \circ \Delta_{LR}
		\bigr) f_2
	\bigr]
	\in \Upsilon_+
	\,, \\[1ex]
	&
	\preind{LRT}{D} = 
	\tfrac{1}{2} \, \overline{\Up}_{LR}\bigl[ \varpi_{T} \bigr]
	\in \Upsilon_+
	\,,
\end{aligned}
\end{equation}
where:
\begin{itemize}
	\item $\lambda_{{LR},\textup{s}} \coloneqq \mathscr{L}^{-1} \circ L_{{LR},\textup{s}} \circ \mathscr{L}$ and $\lambda_{{LR},\textup{t}} \coloneqq \mathscr{L}^{-1} \circ L_{{LR},\textup{t}} \circ \mathscr{L}$ are the automorphisms of $\Upsilon_+$ defined as the Laplace transform of
	\begin{equation}\label{eq:L:LR}
		L_{{LR},\textup{s}} \coloneqq L \circ R(u) \,,
		\qquad
		L_{{LR},\textup{t}} \coloneqq L \circ R(-u) \,,
	\end{equation}
	that is the automorphisms of $V_+$ acting as $L$ composed with the multiplication by $R(u)$ and $R(-u)$ respectively.

	\item $\Up_{LR}$, $\dd_{LR}$, and $\Delta_{LR}$ are the following compositions of linear maps:
	\begin{equation}\label{eq:UpLR:DownLR}
		\Up_{LR} \coloneqq \lambda_{{LR},\textup{t}} \circ \Up  \,,
		\qquad
		\dd_{LR} \coloneqq \dd \circ \lambda_{{LR},\textup{s}}^{-1} \,,
		\qquad
		\Delta_{LR} \coloneqq \Delta \circ \lambda_{{LR},\textup{t}}^{-1} \,.
	\end{equation}
	The operator $\overline{\Up}_{LR}$ is again the extension of $\Up_{LR}$ to $\Upsilon = \dd \widehat{\Upsilon}_+ \oplus \Upsilon_-$ which is zero on $\dd\widehat{\Upsilon}_+$.
\end{itemize}
Diagrammatically, the final formulae for the F-Airy structure computing the Laplace transformed amplitudes associated with $\Omega = \hat{L}\hat{R}\hat{T}\Omega^0$ can be represented as follows (we omit the superscript $LRT$ from the tensors).
\begin{equation}\label{eq:F-CohFT:diagr}
	\begin{tikzpicture}[%
		baseline,%
		xscale=.7, yscale=.5
	]
	% A
		\node at (-1.8,0) {$A =$};
		\draw (-1,-1) -- (1,-1) -- (1,1) -- (-1,1) -- cycle;
		\node at (0,0) {$\bcdot_{{\scriptscriptstyle T}}$};

		\draw (-.6,1) -- (-.6,1.5);
		\draw (-.1,1.5) -- (-1.1,1.5) -- (-1.1,2.5) -- (-.1,2.5) -- cycle;
		\node at (-.6,2) {\tiny $\dd_{\scriptscriptstyle L \!\! R}$};
		\draw (-.6,2.5) -- (-.6,3);

		\draw (.6,1) -- (.6,1.5);
		\draw (.1,1.5) -- (1.1,1.5) -- (1.1,2.5) -- (.1,2.5) -- cycle;
		\node at (.6,2) {\tiny $\dd_{\scriptscriptstyle L \!\! R}$};
		\draw (.6,2.5) -- (.6,3);

		\draw (0,-1) -- (0,-1.5);
		\draw (-.5,-1.5) -- (.5,-1.5) -- (.5,-2.5) -- (-.5,-2.5) -- cycle;
		\node at (0,-2) {\tiny $\overline{\Up}_{\scriptscriptstyle L \!\! R}$};
		\draw (0,-2.5) -- (0,-3);

	% B
	\begin{scope}[xshift=4.5cm]
		\node at (-1.8,0) {$B =$};
		\draw (-1,-1) -- (1,-1) -- (1,1) -- (-1,1) -- cycle;
		\node at (0,0) {$\bcdot_{{\scriptscriptstyle T}}$};

		\draw (-.6,1) -- (-.6,3);
		\draw (-.1,3) -- (-1.1,3) -- (-1.1,4) -- (-.1,4) -- cycle;
		\node at (-.6,3.5) {\tiny $\dd_{\scriptscriptstyle L \!\! R}$};
		\draw (-.6,4) -- (-.6,4.5);

		\draw (.6,1) -- (.6,1.5);
		\draw (-.2,1.5) -- (1.4,1.5) -- (1.4,2.5) -- (-.2,2.5) -- cycle;
		\node at (.6,2) {\tiny $\id \!+\! \dd\Epsilon_{\scriptscriptstyle R}$};
		\draw (.6,2.5) -- (.6,3);
		\draw (.1,3) -- (1.1,3) -- (1.1,4) -- (.1,4) -- cycle;
		\node at (.6,3.5) {\tiny $\Delta_{\scriptscriptstyle L \!\! R}$};
		\draw (.6,4) -- (.6,4.5);

		\draw (0,-1) -- (0,-1.5);
		\draw (-.5,-1.5) -- (.5,-1.5) -- (.5,-2.5) -- (-.5,-2.5) -- cycle;
		\node at (0,-2) {\tiny $\overline{\Up}_{\scriptscriptstyle L \!\! R}$};
		\draw (0,-2.5) -- (0,-3);
	\end{scope}

	% Cconn
	\begin{scope}[xshift=9cm]
		\node at (-1.8,0) {$\conn{C} =$};
		\draw (-1,-1) -- (1,-1) -- (1,1) -- (-1,1) -- cycle;
		\node at (0,0) {$\gamma_{{\scriptscriptstyle T}}$};

		\draw (-.6,-1) -- (-.6,-1.5);
		\draw (-.1,-1.5) -- (-1.1,-1.5) -- (-1.1,-2.5) -- (-.1,-2.5) -- cycle;
		\node at (-.6,-2) {\tiny $\overline{\Up}_{\scriptscriptstyle L \!\! R}$};
		\draw (-.6,-2.5) -- (-.6,-3);

		\draw (.6,-1) -- (.6,-1.5);
		\draw (.1,-1.5) -- (1.1,-1.5) -- (1.1,-2.5) -- (.1,-2.5) -- cycle;
		\node at (.6,-2) {\tiny ${\vphantom{\Up}\lambda}_{\scriptscriptstyle L \!\! R,\textup{t}}$};
		\draw (.6,-2.5) -- (.6,-3);

		\draw (0,1) -- (0,1.5);
		\draw (-.5,1.5) -- (.5,1.5) -- (.5,2.5) -- (-.5,2.5) -- cycle;
		\node at (0,2) {\tiny $\Delta_{\scriptscriptstyle L \!\! R}$};
		\draw (0,2.5) -- (0,3);
	\end{scope}

	% Cdisc
	\begin{scope}[xshift=13.5cm]
		\node at (-1.8,0) {$\disc{C} =$};
		\draw (-1,-1) -- (1,-1) -- (1,1) -- (-1,1) -- cycle;
		\node at (0,0) {$\bcdot_{{\scriptscriptstyle T}}$};

		\draw (-.6,1) -- (-.6,1.5);
		\draw (-.1,1.5) -- (-1.7,1.5) -- (-1.7,2.5) -- (-.1,2.5) -- cycle;
		\node at (-.9,2) {\tiny $\id \!+\! \dd\Epsilon_{\scriptscriptstyle R}$};
		\draw (-.6,2.5) -- (-.6,3);
		\draw (-.1,3) -- (-1.1,3) -- (-1.1,4) -- (-.1,4) -- cycle;
		\node at (-.6,3.5) {\tiny $\Delta_{\scriptscriptstyle L \!\! R}$};
		\draw (-.6,4) -- (-.6,4.5);

		\draw (.6,1) -- (.6,1.5);
		\draw (.1,1.5) -- (1.7,1.5) -- (1.7,2.5) -- (.1,2.5) -- cycle;
		\node at (.9,2) {\tiny $\id \!+\! \dd\Epsilon_{\scriptscriptstyle R}$};
		\draw (.6,2.5) -- (.6,3);
		\draw (.1,3) -- (1.1,3) -- (1.1,4) -- (.1,4) -- cycle;
		\node at (.6,3.5) {\tiny $\Delta_{\scriptscriptstyle L \!\! R}$};
		\draw (.6,4) -- (.6,4.5);

		\draw (0,-1) -- (0,-1.5);
		\draw (-.5,-1.5) -- (.5,-1.5) -- (.5,-2.5) -- (-.5,-2.5) -- cycle;
		\node at (0,-2) {\tiny $\overline{\Up}_{\scriptscriptstyle L \!\! R}$};
		\draw (0,-2.5) -- (0,-3);
	\end{scope}

	% D
	\begin{scope}[xshift=18cm]
		\node at (-1.6,0) {$D =$};
		\draw (-.8,-1) -- (.8,-1) -- (.8,1) -- (-.8,1) -- cycle;
		\node at (0,0) {$\tfrac{1}{2}\varpi_{{\scriptscriptstyle T}}$};

		\draw (0,-1) -- (0,-1.5);
		\draw (-.5,-1.5) -- (.5,-1.5) -- (.5,-2.5) -- (-.5,-2.5) -- cycle;
		\node at (0,-2) {\tiny $\overline{\Up}_{\scriptscriptstyle L \!\! R}$};
		\draw (0,-2.5) -- (0,-3);
	\end{scope}

	\end{tikzpicture}
\end{equation}

\begin{rem}\label{rem:UpLR:DownLR:new}
	Formulae \eqref{eq:UpLR:DownLR} express the operator $\Up_{LR}$ and $\Delta_{LR}$ in terms of the up/down-morphisms $(\Up,\Delta)$. It would be more natural though to express them in terms of (the Laplace transform of) the new up/down-morphisms $(\tilde{\mathscr{U}},\tilde{\mathscr{D}})$ provided by Theorem~\ref{thm:ident:actions}, that is
	\begin{equation}
		\tilde{\Up}
		=
		\lambda_{{LR},\textup{t}} \circ \Up \circ \mu_{LR}^{-1}
		\qquad\text{and}\qquad
		\tilde{\Delta}
		=
		\mu_{LR} \circ \Delta \circ \lambda_{{LR},\textup{t}}^{-1} \,.
	\end{equation}
	Here $\mu_{LR} \coloneqq \mathscr{L}^{*} \circ M_{LR} \circ (\mathscr{L}^{*})^{-1}$ and its inverse $\mu_{LR}^{-1} = \mathscr{L}^{*} \circ M_{LR}^{-1} \circ (\mathscr{L}^{*})^{-1}$ are the automorphisms of $\Upsilon_-$ defined as the Laplace transforms of
	\begin{equation}
		M_{LR}[\chi](u) \coloneqq \Bigl[ L \circ R(-u) \chi(u) \Bigr]_{-} \,,
		\qquad
		M_{LR}^{-1}[\chi](u) = \Bigl[ (L \circ R(-u))^{-1} \chi(u) \Bigr]_{-} \,.
	\end{equation}
	The up/down-morphisms $(\tilde{\mathscr{U}},\tilde{\mathscr{D}})$ are more natural, since they are the ones used to compute the transformed amplitudes. This is easily achieved as
	\begin{equation}\label{eq:UpLR:DownLR:new}
		\Up_{LR}
		=
		\tilde{\Up} \circ \mu_{LR}
		\qquad\text{and}\qquad
		\Delta_{LR}
		=
		\mu_{LR}^{-1} \circ \tilde{\Delta} \,.
	\end{equation}
\end{rem}

\begin{rem}
	We conclude with a remark on the symmetries of the F-Airy structure tensors associated with an F-CohFT of the form $\Omega = \hat{L}\hat{R}\hat{T}\Omega^0$. If $\Omega^0$ is a topological field theory (rather than merely an F-topological one), then the associated tensors satisfy the so-called IHX relations (cf.~\cite{ABCO24}). In this case, the actions of $\hat{L}$, $\hat{R}$, and $\hat{T}$ preserve these symmetries for the tensors $A$, $B$, and $\disc{C}$, since their transformation laws (up to index raising/lowering) match those of the corresponding classical Airy structures. However, this argument does not extend to the tensors $\conn{C}$ and $D$, which transform differently due to the presence of sums over stable trees (rather than stable graphs).
\end{rem}

%–––––––––––––––––––––––––––––––––––––––––––%
\subsection{Example: the extended 2-spin F-CohFT}
\label{subsec:ext:2spin:class}
%–––––––––––––––––––––––––––––––––––––––––––%
An example of F-CohFT is given by the extended $r$-spin class. The underlying F-manifold was constructed in \cite{JKV01} and further studied in \cite{BCT19,Bur20,BR21,ABLR23}. In this section, we focus on the $r = 2$ case.

From the F-manifold of the extended $2$-spin theory, we can associate two families of F-CohFTs depending on a parameter $s \in \mathbb{C}^*$ and both defined over the vector space $V_0 \coloneqq \CC\mathrm{e}_1 \oplus \CC\mathrm{e}_2$. The first one is the extended $2$-spin CohFT shifted along $(0,s)$ \cite{BR21}:
\begin{equation}
	\Omega_{g,1+n}^{s} \colon
	V_0^{\otimes n} \longrightarrow H^{\textup{even}}(\Mbar_{g,1+n}) \otimes V_0 \,.
\end{equation}
In the original reference, it is denoted as $c^{2,\textup{ext},(0,s)}_{g,1+n}$. The second family is obtained from the F-Givental group action \cite{ABLR23}:
\begin{equation}
	\overline{\Omega}_{g,1+n}^{s} \colon
	V_0^{\otimes n} \longrightarrow H^{\textup{even}}(\Mbar_{g,1+n}) \otimes V_0 \,.
\end{equation}
In the original reference, it is denoted as $c^{\overline{F}_{(0,s)},(0,-s^2)}_{g,1+n}$ and its construction is recalled below.

As pointed out in \cite{ABLR23}, the two F-CohFTs do not coincide, but it is reasonable to expect that they are related on the moduli space of stable curves of compact type (recall that a stable curve is of compact type if its dual graph is a stable tree). Indeed, on the one hand $\overline{\Omega}^{s}$ is constructed through the F-Givental action, and as such it is supported on compact type. On the other hand $\Omega^{s}$ is non-zero outside compact type, but after multiplication by $\lambda_g$ (the top Chern class of the Hodge bundle) we get a class supported on compact type. It is conjectured that
\begin{equation}\label{eq:conj:ext:2spin}
	\overline{\Omega}_{g,1+n}^{s} 
	\overset{?}{=}
	\lambda_g \, \Omega_{g,1+n}^{s} \,.
\end{equation}
In support of this conjecture, notice that $\overline{\Omega}_{g,1+n}^{s}(\mathrm{e}_{1}^{\otimes n}) = 0$, while $\Omega_{g,1+n}^{s}(\mathrm{e}_{1}^{\otimes n}) = \lambda_g \, \mathrm{e}_1$. The latter restricts to zero on the moduli of compact type as $\lambda_g^2 = 0$. A proof of equation~\eqref{eq:conj:ext:2spin} would be particularly interesting from the point of view of the double ramification hierarchy \cite{Bur15,BR16,BR21}, where only $\lambda_g \, \Omega_{g,1+n}^{s}$ is relevant. This last point motivates our interest in the intersection indices of $\overline{\Omega}_{g,1+n}^{s}$ and $\psi$-classes: thanks to the identification discussed in Theorems~\ref{thm:top:F-CohFT:F-TR} and \ref{thm:ident:actions}, such intersection indices are recursively computed by F-TR.

We start by recalling from \cite{ABLR23} the construction of $\overline{\Omega}^{s}$. The underlying F-TFT, denoted $\overline{\Omega}^{s,0}$ is identified by the algebra $(V_0,\bcdot)$ and distinguished vector $w$ given as
\begin{equation}
	V_0 \coloneqq \CC\mathrm{e}_1 \oplus \CC\mathrm{e}_2 \,,
	\qquad
	\mathrm{e}_{\beta} \bcdot \mathrm{e}_{\gamma} \coloneqq \delta^{\alpha}_{\beta,\gamma} \mathrm{e}_{\alpha} \,,
	\qquad
	w \coloneqq -s^2 \, \mathrm{e}_2 \,.
\end{equation}
In particular, the unit is $\mathrm{e} = \mathrm{e}_1 + \mathrm{e}_2$. Notice that the F-TFT is semisimple. Now consider $L \in \GL(V_0)$, $R(u) \in \mathfrak{Giv}$, and $T(u) \in u^2 V_0\bbraket{u}$ given by
\begin{equation}
	L
	\coloneqq
	\begin{pmatrix}
		1 & 0 \\
		\frac{1}{s} & -\frac{1}{s}
	\end{pmatrix} ,
	\quad
	R(u)
	\coloneqq
	\id_{V_0}
	-
	\sum_{m \ge 1}
		\begin{pmatrix}
			0 & 0 \\
			\frac{(2m-1)!!}{s^{2m}} & 0
		\end{pmatrix} u^m \,,
	\quad
	T(u)
	\coloneqq
	- \sum_{m \geq 2} \frac{(2m-3)!!}{s^{2m - 2}} \mathrm{e}_2 \, u^{m} \,.
\end{equation}
Notice that that the translation is the one induced by the R-matrix as in Theorem~\ref{thm:translation:F-CohFT}, that is $T(u) = u(\id_{V_0} - R^{-1}(u))\mathrm{e}$. Then $\overline{\Omega}^{s}$ is defined as
\begin{equation}
	\overline{\Omega}^{s}
	\coloneqq
	\hat{L}\hat{R}\hat{T}\overline{\Omega}^{s,0} \,.
\end{equation}
Our goal is to compute the correlators associated with the above F-CohFT. To this end, we choose the standard up/down-morphism of Example~\ref{ex:stndrd:up:down}, so that in the basis $(\mathrm{e}_{\alpha,k} = \mathrm{e}_{\alpha} u^k)_{\alpha,k}$ of $V_+$ the correlators read
\begin{equation}
	\indF{F}{g}{(\alpha_0,k_0)}{(\alpha_1,k_1),\dots,(\alpha_n,k_n)}
	=
	\int_{\Mbar_{g,1+n}}
		\Braket{
			\mathrm{e}^{\alpha_0},
			\overline{\Omega}^{s}_{g,1+n}(\mathrm{e}_{\alpha_1} \otimes \cdots \otimes \mathrm{e}_{\alpha_n})
		}
		\psi_{0}^{k_0} \prod_{i=1}^n \psi_{i}^{k_i} \,.
\end{equation}
We proceed by computing all the ingredients appearing in \eqref{eq:FAiry:loop:LRT}. The computations are performed on the natural bases of $\Upsilon_+$ and $\Upsilon_-$, that is
\begin{equation}
	\epsilon_{(\alpha,k)} \coloneqq \mathrm{e}_{\alpha} \frac{\zeta^{2k+1}}{(2k+1)!!} \in \Upsilon_+ \,,
	\qquad\qquad
	\epsilon_{\alpha}^k \coloneqq \mathrm{e}_{\alpha} \frac{(2k+1)!!}{\zeta^{2k+2}} \dd\zeta \in \Upsilon_- \,.
\end{equation}
It is also convenient to introduce a basis of $\dd\Upsilon_+$ by extending that of $\Upsilon_-$ to negative indices:
\begin{equation}
	\epsilon_{\alpha}^{-k}
	\coloneqq
	\mathrm{e}_{\alpha} \frac{(-2k+1)!!}{\zeta^{-2k+2}} \dd\zeta
	=
	(-1)^{k-1} \, \mathrm{e}_{\alpha} \frac{\zeta^{2k-2}}{(2k-3)!!} \dd\zeta \,.
\end{equation}
The last equation follows the convention $(-2k+1)!! \coloneqq (-1)^{k-1} \frac{1}{(2k-3)!!}$, which is the natural extension of the double factorial deduced from its relation with the Gamma function. With this convention,
\begin{equation}
	\dd\epsilon_{(\alpha,k)} = (-1)^k \epsilon_{\alpha}^{-k-1} \,.
\end{equation}
We will use double factorials of odd negative integers throughout the rest of this section.

%–––––––––––––––––––––––
\textbf{Change of bases.} 
The automorphisms $\lambda_{LR,\textup{s}}^{-1}$ and $\lambda_{LR,\textup{t}}$ of $\Upsilon_+$ responsible for the change of bases are
\begin{equation}
\begin{split}
	% & \lambda_{LR,\textup{s}}[\epsilon_{(\alpha,k)}]
	% =
	% \delta_{\alpha}^1 \biggl(
	% 	\epsilon_{(1,k)} + \sum_{m \geq 0} \tfrac{(2m-1)!!}{s^{2m + 1}} \epsilon_{(2,k+m)}
	% \biggr)
	% - \tfrac{1}{s} \delta_{\alpha}^2 \epsilon_{(2,k) \,,
	% \\
	& \lambda_{LR,\textup{s}}^{-1}[\epsilon_{(\alpha,k)}]
	=
	\delta_{\alpha}^1 \biggl(
		\epsilon_{(1,k)}
		+
		\sum_{m \geq 0}
			\tfrac{(2m-1)!!}{s^{2m}} \epsilon_{(2,k+m)}
	\biggr)
	- s \delta_{\alpha}^2 \epsilon_{(2,k)} \,,
	\\
	& \lambda_{LR,\textup{t}}[\epsilon_{(\alpha,k)}]
	=
	\delta_{\alpha}^1 \biggl(
		\epsilon_{(1,k)}
		+
		\tfrac{1}{s}
		\sum_{m \geq 0}
			\tfrac{1}{(-2m-1)!! \, s^{2m}} \epsilon_{(2,k+m)}
	\biggr)
	- \tfrac{1}{s} \delta_{\alpha}^2 \epsilon_{(2,k)} \,.
\end{split}
\end{equation}
In particular, the twisted differential $\dd_{LR}$ reads
\begin{equation}
	\dd_{LR}\epsilon_{(\alpha,k)}
	=
	(-1)^k
	\biggl[
		\delta_{\alpha}^1
		\biggl(
			\epsilon_{1}^{-k-1}
			+
			\sum_{m \geq 0} \tfrac{1}{(-2m-1)!! \, s^{2m}} \epsilon_{2}^{-k-1-m}
		\biggr)
		- s \delta_{\alpha}^2 \epsilon_{2}^{-k-1}
	\biggr] \,.
\end{equation}
% Besides, the automorphisms $\mu$ of $\Upsilon_-$ and its inverse appearing in Remark~\ref{rem:UpLR:DownLR:new} are given as
% \begin{equation}
% \begin{split}
% 	& \mu[\epsilon_{\alpha}^k]
% 	=
% 	\delta_{\alpha}^1 \biggl(
% 		\epsilon_{1}^k
% 		+
% 		\tfrac{1}{s}
% 		\sum_{m = 0}^{k}
% 			\tfrac{1}{(-2m-1)!! \, s^{2m}} \epsilon_{2}^{k-m}
% 	\biggr)
% 	- \tfrac{1}{s} \delta_{\alpha}^2 \epsilon_{2}^k \,,
% 	\\
% 	& \mu^{-1}[\epsilon_{\alpha}^k]
% 	=
% 	\delta_{\alpha}^1 \biggl(
% 		\epsilon_{1}^k
% 		+
% 		\sum_{m = 0}^{k}
% 			\tfrac{1}{(-2m-1)!! \, s^{2m}} \epsilon_{2}^{k-m}
% 	\biggr)
% 	- s \delta_{\alpha}^2 \epsilon_{2}^k \,.
% \end{split}
% \end{equation}
As the final up/down-morphisms are chosen to be the standard ones, we find that the isomorphisms $\Up_{LR}$ and $\Delta_{LR}$ (computed via \eqref{eq:UpLR:DownLR:new}) are simply given by
\begin{equation}
\begin{split}
	& \Up_{LR}[\epsilon_{\alpha}^k]
	=
	\delta_{\alpha}^1 \biggl(
		\epsilon_{(1,k)}
		+
		\tfrac{1}{s}
		\sum_{m = 0}^{k}
			\tfrac{1}{(-2m-1)!! \, s^{2m}} \epsilon_{(2,k-m)}
	\biggr)
	- \tfrac{1}{s} \delta_{\alpha}^2 \epsilon_{(2,k)} \,,
	\\
	& \Delta_{LR}[\epsilon_{(\alpha,k)}]
	=
	\delta_{\alpha}^1 \biggl(
		\epsilon_{1}^k
		+
		\sum_{m = 0}^{k}
			\tfrac{1}{(-2m-1)!! \, s^{2m}} \epsilon_{2}^{k-m}
	\biggr)
	- s \delta_{\alpha}^2 \epsilon_{2}^k \,.
\end{split}
\end{equation}

%–––––––––––––––––––––––
\textbf{R-action.} 
The Laplace transform of the differential of the edge weight, as a linear operator $\dd \circ \Epsilon_R \colon \Upsilon_- \to \dd\Upsilon_+$, reads
\begin{equation}
	(\dd \circ \Epsilon_R)[\epsilon_{\alpha}^k]
	=
 	\delta_{\alpha}^1 \sum_{m \geq k+1} \frac{1}{(-2m-1)!! \, s^{2m}} \epsilon_{2}^{k-m} \,.
\end{equation}

%–––––––––––––––––––––––
\textbf{Translation.}
% The differential of the inverse Laplace transform of the translation, normalised by $\zeta^2 \dd\zeta$ is
% \begin{equation}
% 	\frac{\dd\tau(\zeta)}{\zeta^2 \dd\zeta}
% 	=
% 	- \mathrm{e}_2 \sum_{m \ge 1} \frac{1}{2m+1} \Bigl( \frac{\zeta}{s} \Bigr)^{2m}
% 	=
% 	\mathrm{e}_2
% 	\left(
% 		1 + \frac{s}{2\zeta} \ln{\frac{s - \zeta}{s + \zeta}}
% 	\right),
% \end{equation}
% so that $\theta = \frac{1}{\zeta^2 \dd\zeta - \dd\tau}$ reads
The element $\theta_{\scriptscriptstyle T}(\zeta) = \frac{1}{\zeta^2 \dd\zeta - \dd\tau}$ is easily computed from the Laplace transform of the translation as
\begin{equation}
	\theta_{\scriptscriptstyle T}(\zeta)
	=
	\frac{1}{\zeta^2 \dd\zeta} \,\frac{1}{\mathrm{e} - \mathrm{e}_2
	\frac{s}{2\zeta} \ln{\frac{s - \zeta}{s + \zeta}}
	}
	=
	\frac{1}{\zeta^2 \dd\zeta}
	\biggl(
		\mathrm{e}_1 + \mathrm{e}_2 \sum_{m \ge 0} \frac{\vartheta_m}{s^{2m}} \zeta^{2m}
	\biggr) ,
\end{equation}
with the convention $\vartheta_0 \coloneqq 1$. The last equation follows from the fact that $\mathrm{e}_2$ is idempotent, and the expansion coefficients are given by
\begin{equation}
	\vartheta_m
	=
	\sum_{\substack{m_1 + m_2 + \cdots = m \\ m_1,m_2,\ldots \ge 1}}
		\prod_{i \ge 1} \frac{-1}{2m_i + 1} \,.
\end{equation}
Notice that the sum is finite, since $m_i \ge 1$ are required to sum up to $m$. The first elements of the sequence $(\vartheta_m)_{m \ge 0}$ are $(1,\frac{1}{3},\frac{4}{45},\frac{44}{945},\frac{428}{14175},\frac{10196}{467775},\frac{10719068}{638512875},\ldots)$.

It is now possible to compute the twisted product $\bcdot_{\scriptscriptstyle T}$, the map $\kappa_{\scriptscriptstyle T}$, and the $V_0$-valued form $\varpi_{\scriptscriptstyle T}$. The computations are performed modulo $\dd\Upsilon_+$, since the subsequent application of $\overline{\Up}_{LR}$ would annihilate all such terms. The twisted product on elements of $\dd\Upsilon_+ \oplus \Upsilon_-$ is given by
\begin{equation}
	\epsilon_{\beta}^j
	\bcdot_{\scriptscriptstyle T}
	\epsilon_{\gamma}^k
	=
	\delta^{1}_{\beta,\gamma} \, \dfsymb{j,k}{j+k+2} \, \epsilon_{1}^{j+k+2}
	+
	\delta^{2}_{\beta,\gamma}
	\sum_{m = 0}^{\max\{0,j+k+2\}}
		\dfsymb{j,k}{j+k+2-m} \,
		\tfrac{\vartheta_m}{s^{2m}} \,
		\epsilon_{2}^{j+k+2-m}
	+
	\dd\Upsilon_+
\end{equation}
for any $j,k \in \ZZ$. Here, and in the rest of this section, we use the following short-hand notation for ratio of double factorials:
\begin{equation}
	\Bigl\{ \! \genfrac..{0pt}{1}{a_1 \, ,\ldots, \, a_M}{b_1 \, ,\ldots, \, b_N} \! \Bigr\}
	\coloneqq
	\frac{\prod_{i=1}^M (2a_i+1)!!}{\prod_{j=1}^N (2b_j+1)!!}
	\qquad\text{for}\qquad
	a_i,b_j \in \ZZ \,.
\end{equation}
The double factorial of negative odd integers is assumed as above. Notice that expressions of this form are the main combinatorial factors appearing in the Virasoro constrains for the Witten--Kontsevich correlators.

The map $\kappa_{\scriptscriptstyle T} \colon \Upsilon_- \to \Upsilon \otimes \Upsilon_+$ is
\begin{multline}
	\kappa_{\scriptscriptstyle T}[\epsilon_{\gamma}^k]
	=
	- s^2
	\sum_{\ell \ge 0} \biggl(
		\delta_{\gamma}^1 \,
		\dfsymb{\ell,k}{k+\ell+2} \,
		%\tfrac{(2\ell+1)!!(2k+1)!!}{(2(k+\ell+2)+1)!!} 
		\epsilon_{1}^{k+\ell+2} \\
		+
		\delta_{\gamma}^2
		\sum_{m = 0}^{k+\ell+2} 
			\dfsymb{\ell,k}{k+\ell+2-m} \,
			%\tfrac{(2\ell+1)!!(2k+1)!!}{(2(k+\ell+2-m)+1)!!} 
			\tfrac{\vartheta_m}{s^{2m}}
			\epsilon_{2}^{k+\ell+2-m}
	\biggr) \otimes \epsilon_{(2,\ell)}
	+
	\dd\Upsilon_+ \otimes \Upsilon_+ \,.
\end{multline}

Finally the $V_0$-valued form is
\begin{equation}
	\varpi_T
	=
	- \frac{1}{12}
	\left(
		s^2 \epsilon_{2}^{1} - \epsilon_{2}^{0}
	\right)
	+ \dd\Upsilon_+ \,.
\end{equation}

%–––––––––––––––––––––––
\textbf{The $(A,B,\conn{C},\disc{C},D)$ tensors.}
Using all the necessary ingredients, we obtain the following expressions for the tensors of the extended $2$-spin F-CohFT $\overline{\Omega}^s$.
{\allowdisplaybreaks
\begin{align}
	% –––––––––––– %
	% A tensor
	% –––––––––––– %
	\notag
	& \ind{A}{(\alpha,i)}{(\beta,j),(\gamma,k)}
		=
		\Bigl[
			\delta^{\alpha}_1
			\delta_{\beta,\gamma}^1
			+
			\delta^{\alpha}_2
			\bigl(
				\delta_{(\beta,\gamma)}^{(1,2)}
				+
				\delta_{(\beta,\gamma)}^{(2,1)}
				-
				s
				\delta_{\beta,\gamma}^{2}
			\bigr)
		\Bigr]
		\delta_{0}^{i}
		\delta_{j,k}^{0}
		\\[1ex]
	% –––––––––––– %
	% B tensor
	% –––––––––––– %
	\notag
	& \ind{B}{(\alpha,i)}{(\beta,j),(\gamma,k)}
		=
		\delta^{\alpha}_1
		\delta_{\beta,\gamma}^1
		\delta^{i}_{k-j+1}
		\dfsymb{k}{i,j-1}
		+
		\tfrac{
			\delta^{\alpha}_2
			\delta^{i \le k-j+1}
		}{
			s^{2(k-j+1-i)}
		}
		\Biggl[
			\tfrac{1}{s}
			\delta_{\beta,\gamma}^1
			\Biggl(
				\dfsymb{k}{k-j+1,j-1,i+j-k-2}
			\\
			\notag
			&\quad\qquad\qquad
				-
				\sum_{\substack{p,q \ge 0 \\ p+q \le k-j+1-i}}
					\dfsymb{k-q,p-1}{i,j-1+p,-q-1}
					\vartheta_{k-j+1-i-p-q} 
			\Biggr)
			-
			s
			\delta_{\beta,\gamma}^2
			\dfsymb{k}{i,j-1}
			\vartheta_{k-j+1-i}
			\\
			\notag
			&\quad\qquad\qquad
			+
			\delta_{(\beta,\gamma)}^{(1,2)}
			\sum_{p=0}^{k-j+1-i}
				\dfsymb{k,p-1}{i,j-1+p}
				\vartheta_{k-j+1-i-p}
			+
			\delta_{(\beta,\gamma)}^{(2,1)}
			\sum_{q=0}^{k-j+1-i}
				\dfsymb{k-q}{i,j-1,-q-1}
				\vartheta_{k-j+1-i-q}
		\Biggr]
		\\[1ex]
	% –––––––––––– %
	% C conn tensor
	% –––––––––––– %
	\notag
	& \indC{\conn{C}}{(\alpha,i),(\beta,j)}{(\gamma,k)}
		=
		s \, 
		\delta^{(\alpha,\beta)}_{(1,2)}
		\delta_{\gamma}^1
		\delta^{i}_{j+k+2}
		\dfsymb{j,k}{i}
		+
		\tfrac{ 
			\delta^{\alpha,\beta}_2
			\delta^{i \le j+k+2}
		}{
			s^{2(j+k+2-i)}
		}
		\biggl[
			s
			\delta_{\gamma}^1
			\biggl(
				\dfsymb{j,k}{j+k+2,i-j-k-3}
			\\
			\notag
			&\quad\qquad\qquad\qquad\qquad\qquad
				-
				\sum_{m=0}^{k}
					\dfsymb{j,k-m}{i,-m-1}
					\vartheta_{j+k+2-i-m}
			\biggr)
			+
			\delta_{\gamma}^2
			\dfsymb{j,k}{i}
			\vartheta_{j+k+2-i}
		\biggr]
		\\[1ex]
	% –––––––––––– %
	% C disc tensor
	% –––––––––––– %
	\notag
	& \indC{\disc{C}}{(\alpha,i)}{(\beta,j),(\gamma,k)}
		=
		\delta^{\alpha}_1
		\delta_{\beta,\gamma}^1
		\delta^i_{j+k+2}
		\dfsymb{j,k}{i}
		+
		\tfrac{
			\delta^{\alpha}_2
			\delta^{i \le j+k+2}
		}{
			s^{2(j+k+2-i)}
		}
		\Biggl[
			\tfrac{1}{s}
			\delta_{\beta,\gamma}^1
			\Biggl(
				\dfsymb{j,k}{j+k+2,i-j-k-3}
			\\
			\notag
			&\quad\qquad\qquad
				-
				\sum_{\substack{p,q \ge 0 \\ p+q \le j+k+2-i}}
				\dfsymb{j-p,k-q}{i,-p-1,-q-1}
				\vartheta_{j+k+2-i-p-q}
			\Biggr)
			-
			s
			\delta_{\beta,\gamma}^2
			\dfsymb{j,k}{i}
			\vartheta_{j+k+2-i}
			\\
			\notag
			&\quad\qquad\qquad
			+
			\delta_{(\beta,\gamma)}^{(1,2)}
			\sum_{p = 0}^{j+k+2-i}
				\dfsymb{j-p,k}{i,-p-1}
				\vartheta_{j+k+2-i-p}
			+
			\delta_{(\beta,\gamma)}^{(2,1)}
			\sum_{q = 0}^{j+k+2-i}
				\dfsymb{j,k-q}{i,-q-1}
				\vartheta_{j+k+2-i-q}
		\Biggr]
		\\[1ex]
	% –––––––––––– %
	% D conn tensor
	% –––––––––––– %
	\label{eq:tensors:ext:2spin}
	& D^{(\alpha,i)}
		=
		\tfrac{1}{24}
		\delta^{\alpha}_2
		\left(
			s \delta^i_1
			-
			\tfrac{1}{s} \delta^i_0
		\right)
\end{align}
}

%–––––––––––––––––––––––––––––––––––––––––––%
\section{A spectral curve formulation}
\label{sec:residue:formulation}
%–––––––––––––––––––––––––––––––––––––––––––%

In this section, we provide an alternative definition of F-topological recursion in terms of spectral curves, along the lines of the original formulation of topological recursion by Eynard--Orantin \cite{EO07}.

%–––––––––––––––––––––––
\subsection{Definition of F-spectral curves and their associated F-topological recursion}
\label{Section61}
%–––––––––––––––––––––––
Define an \emph{F-spectral curve} as the data $(\Sigma,x,y,\conn{\omega_{0,2}},\disc{\omega_{0,2}},w)$, where:
\begin{itemize}
	\item $\Sigma$ is a smooth complex curve (not necessarily compact, nor connected);

	\item $x$ and $y$ are two meromorphic functions on $\Sigma$, such that $x$ has finitely many ramification points $\mathfrak{a} \subset \Sigma$ that are simple; additionally, we require $\dd y$ to be holomorphic and non-zero at the ramification points;

	\item $\conn{\omega}_{0,2}$ and $\disc{\omega}_{0,2}$ are two bidifferentials on $\Sigma^2$ (not necessarily symmetric), holomorphic except for a double pole along the diagonal with leading coefficient $1$ and no other poles;

	\item $w = (w^{\alpha})_{\alpha \in \mathfrak{a}}$ is a collection of scalar weights associated with the ramification points.
\end{itemize}
Since ramification points are simple, in the neighbourhood of each $\alpha \in \mathfrak{a}$ there is a holomorphic involution $\sigma^{\alpha}$ such that $x \circ \sigma^{\alpha} = x$ and $\sigma^{\alpha} \neq \id$. Let $\mathscr{O}$ (resp. $\mathscr{M}$) be the space of holomorphic functions (resp. meromorphic forms with poles at $\mathfrak{a}$ an vanishing residues) on $\Sigma$. Introduce the maps
\begin{equation} \label{eq:projectors}
	\cd{\mathscr{P}} \colon
	\mathscr{M} \longrightarrow \mathscr{M} \,,
	\qquad
	\chi(z) \longmapsto \cd{\mathscr{P}}[\chi](z_0)
	\coloneqq
	\sum_{\alpha \in \mathfrak{a}} \Res_{z = \alpha} \
		\left( \int_{\alpha}^{z} \cd{\omega}_{0,2}(z_0|\ph) \right) \chi(z)
\end{equation}
for $\star \in \set{\connsymb,\discsymb}$, and $z_0$ is considered outside of the contour defining the residue. The meromorphic form $\cd{\mathscr{P}}[\chi]$ is called the polar part of $\chi$, and it has the same divergent part of $\chi$ at $\mathfrak{a}$, while its holomorphic part gets modified in accordance with the choice of $\cd{\omega}_{0,2}$. The properties imposed on $\cd{\omega}_{0,2}$ imply (cf. \cite[Section~2]{BS17}, but we do not need symmetry in the two variables) that $\cd{\mathscr{P}}$ is a projector and $\Ker(\cd{\mathscr{P}}) = \dd\mathscr{O}$. As a consequence,
\begin{equation}\label{eq:global:dec}
	\mathscr{M} = \dd\mathscr{O} \oplus \cd{\mathscr{M}}_{-}
	\qquad\text{with}\qquad
	\cd{\mathscr{M}}_{-} \coloneqq \Im(\cd{\mathscr{P}}) \,.
\end{equation}
Further, $\conn{\mathscr{P}} \circ \disc{\mathscr{P}} = \conn{\mathscr{P}}$ and $\disc{\mathscr{P}} \circ \conn{\mathscr{P}} = \disc{\mathscr{P}}$.

We now define a collection of multidifferentials $\omega_{g,1 + n}$ on $\Sigma^{1 + n}$. They are indexed by integers $g,n \geq 0$ such that $2g - 2 + (1 + n) > 0$, and will be invariant under permutation of their $n$ last variables. We write $\omega_{g,1 + n}(z_0 | z_1,\ldots,z_n)$ to emphasise the special role played by the first variable. The definition proceeds by induction. We first set $\omega_{0,1} \coloneqq y \, \dd x$ and $\omega_{0,2} \coloneqq \disc{\omega}_{0,2}$. We introduce the two kernels
\begin{equation}
	K^{\star,\alpha}(z_0|z)
	\coloneqq
	\frac{
		\frac{1}{2} \int_{\sigma^{\alpha}(z)}^{z} \cd{\omega}_{0,2}(z_0|\ph)
	}{
		\omega_{0,1}(z) - \omega_{0,1}(\sigma^{\alpha}(z))
	}
\end{equation}
and define the operators $\cd{K} \colon \mathscr{M}^{\otimes 2} \to \mathscr{M}$ as
\begin{align}
	\label{eq:ker:conn}
	\conn{K}[\chi](z_0)
	&\coloneqq
	\sum_{\alpha \in \mathfrak{a}} \Res_{z = \alpha} \
		w^{\alpha} \, K^{\connsymb,\alpha}(z_0|z)
		(\conn{\mathscr{P}})^{\otimes 2}[\chi](z,\sigma^{\alpha}(z)) \,,
	\\
	\label{eq:ker:disc}
	\disc{K}[\chi](z_0)
	&\coloneqq
	\sum_{\alpha \in \mathfrak{a}} \Res_{z = \alpha} \
		K^{\discsymb,\alpha}(z_0|z)
		(\disc{\mathscr{P}})^{\otimes 2}[\chi](z,\sigma^{\alpha}(z)) \,.
\end{align}
Notice the multiplication by $w^{\alpha}$ in the connected operator. Then, for $2g - 2 + (1 + n) > 0$, set
\begin{multline}\label{eq:FTR:residue}
	\omega_{g,1+n}(z_0|z_1,\ldots,z_n)
	\coloneqq
	\conn{K} \!\left[
		\omega_{g -1,1+(n+1)}(\ph|\ph,z_1,\ldots,z_n)
	\right]\!(z_0) \\
	+
	\disc{K} \!\left[
		\sum_{\substack{ h + h' = g \\ J \sqcup J' = [n] }}^{*}
				\omega_{h,1+|J|}(\ph|z_J)
				\otimes
				\omega_{h',1+|J'|}(\ph|z_{J'})
		\right]\!(z_0) \,.
\end{multline}
The starred sum means excluding the two terms $(h,1+|J|) = (0,1)$ and $(h',1+|J'|) = (0,1)$. Further, two special cases need to be addressed. If $(g,1+n) = (1,1)$, then $\conn{K}$ acts on $\omega_{0,2}$ as in \eqref{eq:ker:conn} after setting
\begin{equation}
	(\conn{\mathscr{P}})^{\otimes 2}[\omega_{0,2}](z|z') \coloneqq \conn{\omega}_{0,2}(z|z') \,.
\end{equation}
The second special case involves the disconnected terms from \eqref{eq:FTR:residue} with $(h,1+|J|) = (0,2)$ or $(h',1+|J'|) = (0,2)$. In this case, $\disc{\mathscr{P}}$ acts on $\omega_{0,2}(\ph|z_i)$ as the identity.

The invariance of $\omega_{g,1 + n}$ under permutation of the $n$ last variables is clear from the definition. Besides, for all $2g - 2 + (1+ n) > 0$, the multidifferentials satisfy the linear loop equations with respect to any of its variables: for any $\alpha \in \mathfrak{a}$
\begin{equation}\label{eq:symmetry:multidiff}
\begin{aligned}
	&
	\omega_{g,1+n}(z_0|z_1,\ldots,z_n)
	+
	\omega_{g,1+n}(\sigma^{\alpha}(z_0)|z_1,\ldots,z_n)
	\text{ is holomorphic as $z_0 \rightarrow \alpha$ ,}
	\\
	&
	\omega_{g,1+n}(z_0|z_1,\ldots,z_n)
	+
	\omega_{g,1+n}(z_0|\sigma^{\alpha}(z_1),\ldots,z_n)
	\text{ is holomorphic as $z_1 \rightarrow \alpha$ .}
\end{aligned}
\end{equation}

\begin{rem}
	The recursion could be formulated directly in terms of  the multidifferentials $\disc{\omega}_{g,1+n} \coloneqq (\disc{\mathscr{P}})^{\otimes (1+n)}[\omega_{g,1+n}]$. Namely, we have
	\begin{multline}
		\disc{\omega}_{g,1+n}(z_0|z_1,\ldots,z_n)
		\coloneqq 
		\conn{\mathcal{K}} \!\left[
			\disc{\omega}_{g -1,1+(n+1)}(\ph|\ph,z_1,\ldots,z_n)
		\right]\!(z_0) \\
		+
		\disc{\mathcal{K}} \!\left[
			\sum_{\substack{ h + h' = g \\ J \sqcup J' = [n] }}^{*}
					\disc{\omega}_{h,1+|J|}(\ph|z_J)
					\otimes
					\disc{\omega}_{h',1+|J'|}(\ph|z_{J'})
			\right]\!(z_0) \,.
	\end{multline}
	The disconnected recursion operator is then the usual recursion kernel of topological recursion
	\begin{equation}
		\disc{\mathcal{K}}[\chi](z_0)
		\coloneqq
		\sum_{\alpha \in \mathfrak{a}} \Res_{z = \alpha} \
			K^{\discsymb,\alpha}(z_0|z) \, \chi\big(z,\sigma^{\alpha}(z)\big)\,,
	\end{equation} 
	while the connected recursion operator contains all the novelties:
	\begin{equation}
		\conn{\mathcal{K}}[\chi](z_0)
		\coloneqq
		\sum_{\alpha \in \mathfrak{a}} w^{\alpha} \Res_{z = \alpha} \
			K^{\discsymb,\alpha}(z_0|z) \, (\conn{\mathscr{P}})^{\otimes 2}[\chi]\big(z,\sigma^{\alpha}(z)\big)\,.
	\end{equation}
	In both formulae, $z_0$ is considered outside of the contour defining the residue. This alternative definition has the property that, for any $2g - 2 + (1 + n) > 0$, the multidifferential $\disc{\omega}_{g,1+n}$ is an element of $(\disc{\mathscr{M}}_-)^{\otimes(1+ n)}$. We chose \eqref{eq:FTR:residue} as the main definition, as it gives a more symmetric role to the connected and disconnected kernels and only differing by the presence of $w^{\alpha}$ in the connected one. Besides, there is a loss of information from $\omega_{g,1+n}$ to its projection $\disc{\omega}_{g,1+n}$.
\end{rem}

%–––––––––––––––––––––––
\subsection{F-Airy structures from F-spectral curves}
%–––––––––––––––––––––––
Let $(\Sigma,x,y,\conn{\omega}_{0,2},\disc{\omega}_{0,2},w)$ be an F-spectral curve. We now explain how to define an F-Airy structure whose amplitudes reconstruct the projected multidifferentials $(\mathscr{P}^{\star})^{\otimes (n + 1)}[\omega_{g,1+n}]$. This construction will however depend on a choice of up/down-morphisms.

First, we should discuss the local picture. Define $\mathscr{O}_{\loc}$ (resp. $\mathscr{M}_{\loc}$) as the space of germs of holomorphic functions at $\mathfrak{a}$ (resp. germs of meromorphic forms at $\mathfrak{a}$ without residues). The obvious restriction map $\mathscr{M} \rightarrow \mathscr{M}_{\loc}$ and the fact that \eqref{eq:projectors} only depends on the germ of $\chi$ at $\mathfrak{a}$ allows the definition of projectors $\cd{\mathscr{P}}_{\loc} \colon \mathscr{M}_{\loc} \rightarrow \mathscr{M}_{\loc}$ by the same formula. We have $\dd\mathscr{O}_{\loc} = \Ker(\cd{\mathscr{P}}_{\loc})$ and we obtain a decomposition analogous to \eqref{eq:global:dec}:
\begin{equation}\label{eq:local:dec}
	\mathscr{M}_{\loc} = \dd \mathscr{O}_{\loc} \oplus \cd{\mathscr{M}}_{-,\loc}
	\qquad\text{with}\qquad
	\cd{\mathscr{M}}_{-,\loc} \coloneqq \Im(\cd{\mathscr{P}}_{\loc}) \,.
\end{equation}
An extra feature of the local picture is that $\mathscr{M}_{\loc}$ is symplectic for the residue pairing
\begin{equation}\label{eq:residue:pair}
	\langle \chi_1,\chi_2 \rangle
	=
	\sum_{\alpha \in \mathfrak{a}} \Res_{z = \alpha} \left(\int^{z} \chi_1\right) \chi_2(z) \,,
\end{equation}
and the two summands in \eqref{eq:local:dec} are Lagrangian subspaces. In this formula $\int^z \chi_1$ denotes any germ of meromorphic function near $\mathfrak{a}$ such that $\dd(\int^{z} \chi_1) = \chi_1(z)$.

Ideally, one would like to define an F-Airy structure on the subspace $\mathscr{V}_+ \subset \mathscr{O}_{\loc}$ of germs of odd (with respect to the local involution) holomorphic functions at $\mathfrak{a}$. We denote $\cd{\mathscr{V}}_-$ the image by $\cd{\mathscr{P}}_{\loc}$ of the space of germs of meromorphic forms whose polar part is odd, precisely as in \eqref{eq:symmetry:multidiff}. However, this would not quite work due to infinite-dimensional issues (this has to do with the role of up/down-morphisms in the formulae of Proposition~\ref{prop:FAiry:Fsp}).

In order to solve these issues, let us make a choice of injective linear maps $\cd{\mathscr{V}}_- \rightarrow \mathscr{V}_+$ for $\star \in \set{\connsymb,\discsymb}$ with common image $\check{\mathscr{V}}_+$. Then, we consider the up/down-morphisms (analogous to Definition~\ref{def:up:down})
\begin{equation}
	\cd{\Delta} \colon \check{\mathscr{V}}_+ \longrightarrow \cd{\mathscr{V}}_-
	\qquad\quad\text{and}\qquad\quad
	\cd{\Up} \colon \cd{\mathscr{V}}_- \longrightarrow \check{\mathscr{V}}_+ \,,
\end{equation} 
where $\cd{\Up}$ is given by the aforementioned linear maps. We assume that the choices made are such that the corresponding down-morphisms satisfy the compatibility relations
\begin{equation}\label{eq:comp:up:down:proj}
	\disc{\mathscr{P}}_{\loc} \big|_{\conn{\mathscr{V}}_-}
	=
	\disc{\Delta} \circ \conn{\Up}
	\qquad\quad\text{and}\qquad\quad
	\conn{\mathscr{P}}_{\loc} \big|_{\disc{\mathscr{V}}_-}
	=
	\conn{\Delta} \circ \disc{\Up} \,. 
\end{equation}
Note that it is sufficient to choose up/down-morphisms for $\connsymb$ or $\discsymb$, as \eqref{eq:comp:up:down:proj} can then be used to define them for the remaining $\discsymb$ or $\connsymb$.

 We can now define $F_{g,1+n} \in \Hom(\Sym{n}{\check{\mathscr{V}}_+},\check{\mathscr{V}}_+)$ for $2g - 2 + (1 + n) > 0$ by the formula
\begin{equation}\label{eq:Fgn:tensors}
	F_{g,1+n}(f_1 \otimes \cdots \otimes f_n)
	\coloneqq
	(\cd{\Up} \circ \cd{\mathscr{P}}_{\loc})
	\left[
		\sum_{\alpha_1,\ldots,\alpha_n \in \mathfrak{a}}
		\left( \prod_{i = 1}^{n} \Res_{z_i = \alpha_i} f_i(z_i) \right)
		\omega_{g,1+n}(z_0|z_1,\ldots,z_n)
	\right] ,
\end{equation}
where we implicitly used the natural restriction morphism $\mathscr{M} \rightarrow \mathscr{M}_{\loc}$ before applying $\cd{\mathscr{P}}_{\loc}$. The compatibility condition \eqref{eq:comp:up:down:proj} guarantees that these tensors do not depend on the choice of $\star \in \set{\connsymb,\discsymb}$.

\begin{rem}
	By restricting the residue pairing \eqref{eq:residue:pair} to $\dd\mathscr{V}_+ \oplus \cd{\mathscr{V}}_- \subset \mathscr{M}_{\loc}$, we still obtain a symplectic space split as a direct sum of two Lagrangians. In particular, the symplectic structure gives an isomorphism between $\cd{\mathscr{V}}_-$ and the dual of $\dd \mathscr{V}_+$, hence with the dual of $\mathscr{V}_+$ as $\dd |_{\mathscr{V}_+}$ is invertible onto its image. Yet, after taking the residues in \eqref{eq:Fgn:tensors} we are left with an element of $\cd{\mathscr{V}}_-$. To get an output taking values in $\check{\mathscr{V}}_+$ (as $F_{g,1+n}$ should be) without losing information, we need a choice of injective linear map $\Up \colon \cd{\mathscr{V}}_- \rightarrow \check{\mathscr{V}}_+ \subset \mathscr{V}_+$, or equivalently, a choice of isomorphism $\check{\mathscr{V}}_+ \cong \mathscr{V}_+^*$.
\end{rem}

\begin{prop}\label{prop:FAiry:Fsp}
	The tensors $(F_{g,1+n})_{g,n}$ are the amplitudes of the F-Airy structure on $\check{\mathscr{V}}_+$ given by
	\begin{equation}\label{eq:ABCD:Fsc}
		\begin{split}
		& A(f_1 \otimes f_2) 
		=
			\overline{\disc{\Up}}\bigl[
				 \dd f_1 \bcdot \dd f_2 
			\bigr]
			\,, \\[1ex]
		& B(f_1 \otimes f_2)
		=
			\overline{\disc{\Up}}\bigl[
				\dd f_1 \bcdot \disc{\Delta} f_2
			\bigr]
			\,, \\[1ex]
		& \conn{C}(f)
		=
			\bigl( (\overline{\conn{\Up}} \otimes \id_{\mathscr{V}_+}) \circ \conn{\kappa} \bigr)
			\bigl[
				\conn{\Delta} f
			\bigr]
			\,, \\[1ex]
		& \disc{C}(f_1 \otimes f_2)
		=
			\overline{\disc{\Up}}\bigl[
				\disc{\Delta} f_1 \bcdot \disc{\Delta} f_2
			\bigr]
			\,, \\[1ex]
		& D
		=
			\conn{\Up}
			\Biggl[
				\sum_{\alpha \in \mathfrak{a}} w^{\alpha} \Res_{z = \alpha} \ K^{\connsymb,\alpha}(z_0|z) \, \conn{\omega}_{0,2}(z|\sigma^{\alpha}(z))
			\Biggr]
			\,. 
	\end{split} 
	\end{equation}
	Here the following notations/conventions have been used.
	\begin{itemize}
		\item The linear map $\overline{\cd{\Up}}$ is the extension of $\cd{\Up}$ to $\mathscr{M}_{\loc}$ by setting them to zero on $\mathscr{V}_+$ and on meromorphic $1$-forms that are even for the local involutions.

		\item The product $\bcdot$ between two germs of meromorphic forms at $\mathfrak{a}$ is given for $z$ near $\alpha \in \mathfrak{a}$ by
		\begin{equation}
			\chi_1(z) \bcdot \chi_2(z) \coloneqq \chi_1(z) \chi_2(z) \theta^{\alpha}(z)
			\qquad\text{with}\qquad
			\theta^{\alpha}(z) \coloneqq \frac{-2}{\omega_{0,1}(z) - \omega_{0,1}(\sigma^{\alpha}(z))} \,.
		\end{equation}

		\item The map $\conn{\kappa} \colon \conn{\mathscr{V}}_{-} \rightarrow \conn{\mathscr{V}}_-\, \widehat{\otimes}\,\check{\mathscr{V}}_+ \cong \conn{\mathscr{V}}_- \,\widehat{\otimes}\,\mathscr{V}_+$ is defined as
		\begin{equation}
			\chi
			\longmapsto
			\chi(z_1) w^{\alpha} \theta^{\alpha}(z_1) \, \frac{1}{2} \int^{z_2}_{\sigma^{\alpha}(z_2)} \conn{\omega}_{0,2}(z_1|\ph)
		\end{equation}
		for $z_1,z_2$ near the same $\alpha \in \mathfrak{a}$, and zero if $z_1,z_2$ are near different ramification points.
		\item $z_0$ is considered outside of the contour defining the residue in the formula for $D$.
	\end{itemize}
\end{prop}

The above initial data can equivalently be written in coordinates. For each $\alpha \in \mathfrak{a}$, we first choose a determination of the square-root to define the local coordinate near $\alpha$:
\begin{equation}
    \zeta^{\alpha}(z) = \sqrt{2(x(z) - x(\alpha))} \,.
\end{equation}
Then, we introduce the basis of $\cd{\mathscr{V}}_-$ indexed by $\alpha \in \mathfrak{a}$ and $k \geq 0$:
\begin{equation}\label{eq:xi:basis}
	\cdxi{(\alpha,k)}(z_0)
	=
	(2k+1)!! \
	\Res_{z = \alpha} \ \frac{\dd \zeta^{\alpha}(z)}{\zeta^{\alpha}(z)^{2k + 2}}
	\biggl(
		\int_{\alpha}^{z} \cd{\omega}_{0,2}(z_0|\ph)
	\biggr) \,.
\end{equation}
We get a basis $\xi_{(\alpha,k)} = \cd{\Up}[ \cdxi{(\alpha,k)} ]$ of $\check{\mathscr{V}}_+$ which is independent of $\star \in \set{\connsymb,\discsymb}$ due to the compatibility condition \eqref{eq:comp:up:down:proj}. From the definition of the tensors \eqref{eq:Fgn:tensors}, it is easy to see by induction on $2g - 2 + (1 + n) > 0$ that
\begin{equation}\label{eq:Fgn:Fsc:coords}
\begin{aligned}
	& (\cd{\mathscr{P}})^{\otimes(1+n)}\big[ \omega_{g,1+n} \big]
	=
	\indF{F}{g}{(\alpha_0,k_0)}{(\alpha_1,k_1),\dots,(\alpha_n,k_n)} \prod_{i = 1}^{n} \cdxi{(\alpha_i,k_i)} \,, \\
	&
	F_{g,1+n}\big[ \xi_{(\alpha_1,k_1)} \otimes \cdots \otimes \xi_{(\alpha_n,k_n)} \big]
	=
	\indF{F}{g}{(\alpha_0,k_0)}{(\alpha_1,k_1),\dots,(\alpha_n,k_n)} \, \xi_{(\alpha_0,k_0)} \,,
\end{aligned}
\end{equation}
for the same set of coefficients $F_{g;(\alpha_1,k_1),\dots,(\alpha_n,k_n)}^{(\alpha_0,k_0)}$. Note that, as $\cd{\omega}_{0,2}$ is globally defined on $\Sigma^2$, $\cdxi{(\alpha,k)}$ in \eqref{eq:xi:basis} exists not only as a germ at $\mathfrak{a}$ but rather as a globally defined meromorphic form on $\Sigma$, and this is how it should be considered in the first line of \eqref{eq:Fgn:Fsc:coords}.

\begin{proof}
	We first consider $\check{\mathscr{V}}_+$ to be spanned by the following vectors indexed by $(\alpha,k) \in \mathfrak{a} \times \mathbb{Z}_{\geq 0}$
	\begin{equation}
		\epsilon_{(\alpha,k)}(z)
		=
		\begin{cases}
			\dfrac{\zeta^{\alpha}(z)^{2k + 1}}{(2k + 1)!!}
			&
			\text{if $z$ is near $\alpha$,} \\
			0
			&
			\text{else.}
		\end{cases}
	\end{equation}
	By construction of the $\cd{\xi}$-basis of $\cd{\mathscr{V}}_-$ we have the series expansion as $z$ is near $\alpha$ (in this formula $\alpha$ is not summed over):
	\begin{equation}
		\cd{\omega}_{0,2}(z_0|z)
		\mathop{\approx}_{z \rightarrow \alpha}
		\sum_{k \geq 0}
			\cdxi{(\alpha,k)}(z_0) \,
			\dd\epsilon_{(\alpha,k)}(z)
		+
		\big( \text{odd in $z \leftrightarrow \sigma^{\alpha}(z)$} \big) \,.
	\end{equation}
	Accordingly, the kernels of connected or disconnected type admit the following expansions as $z$ is near $\alpha$ (again, $\alpha$ is not summed over):
	\begin{equation}
	\begin{aligned}
		&
		K^{\connsymb,\alpha}(z_0|z)
		\mathop{\approx}_{z \rightarrow \alpha}
		- \frac{w^{\alpha}}{2} \, \theta^{\alpha}(z)
		\sum_{k \geq 0}
			\xi^{\connsymb,(\alpha,k)}(z_0) \,
			\epsilon_{(\alpha,k)}(z) \,, \\
		&
		K^{\discsymb,\alpha}(z_0|z)
		\mathop{\approx}_{z \rightarrow \alpha}
		- \frac{1}{2} \, \theta^{\alpha}(z)
		\sum_{k \geq 0}
			\xi^{\discsymb,(\alpha,k)}(z_0) \,
			\epsilon_{(\alpha,k)}(z) \,.
	\end{aligned}
	\end{equation}
	Let us first assume the standard choice of up/down-morphisms, that is
	\begin{equation}\label{eq:up:down:xi}
		\xi_{(\alpha,k)} = \cd{\Up}\big[ \cdxi{(\alpha,k)} \big] = \epsilon_{(\alpha,k)} \,.
	\end{equation}
	In view of \eqref{eq:Fgn:Fsc:coords}, the recursive definition of the multidifferentials \eqref{eq:FTR:residue} implies that the tensors $(F_{g,1+n})_{g,n}$ coincide with the amplitudes of an F-Airy structure, with tensors $(A,B,\conn{C},\disc{C},D)$ to be identified. We claim that their coefficients in the $\epsilon$-basis read (again, $\alpha$ is not summed over):
	\begin{equation}\label{eq:ABCD:res}
	\begin{split}
		\ind{A}{(\alpha,i)}{(\beta,j),(\gamma,k)}
		& =
		\Res_{z = \alpha} \
			\theta^{\alpha}(z) \,
			\epsilon_{(\alpha,i)}(z) \,
			\dd \epsilon_{(\beta,j)}(z) \,
			\dd \epsilon_{(\gamma,k)}(z)
		\,, \\
		\ind{B}{(\alpha,i)}{(\beta,j),(\gamma,k)}
		& =
		\Res_{z = \alpha} \
			\theta^{\alpha}(z) \,
			\epsilon_{(\alpha,i)}(z) \,
			\dd \epsilon_{(\beta,j)}(z) \, 
			\xi^{\discsymb,(\gamma,k)}(z)
		\,, \\
		\indC{\conn{C}}{(\alpha,i),(\beta,j)}{(\gamma,k)}
		& =
		- \frac{w^{\alpha}}{2} \,
		\Res_{z = \alpha} \
			\theta^{\alpha}(z) \,
			\epsilon_{(\alpha,i)}(z)
			\Bigl(
				\xi^{\connsymb,(\beta,j)}(z) \,
				\xi^{\connsymb,(\gamma,k)}(\sigma^{\alpha}(z))
				+
				\big(z \leftrightarrow \sigma^{\alpha}(z)\big)
			\Bigr)
		\,, \\
		\indC{\disc{C}}{(\alpha,i)}{(\beta,j),(\gamma,k)}
		& =
		-\frac{1}{2}\Res_{z = \alpha} \
			\theta^{\alpha}(z) \,
			\epsilon_{(\alpha,i)}(z)
			\Bigl(
				\xi^{\discsymb,(\beta,j)}(z) \,
			\xi^{\discsymb,(\gamma,k)}(\sigma^{\alpha}(z))
			+
			\big(z \leftrightarrow \sigma^{\alpha}(z)\big)
		\Bigr)
		\,, \\
		D^{(\alpha,k)}
		& =
		- \frac{ w^{\alpha} }{2}
		\Res_{z = \alpha} \
			\theta^{\alpha}(z) \,
			\epsilon_{(\alpha,k)}(z) \,
			\conn{\omega}_{0,2}(z|\sigma^{\alpha}(z))
		\,.
	\end{split} 
	\end{equation}
	In the above equations, we got rid of the occurrence of the local involutions whenever possible, using the facts that the $\epsilon$-basis is odd with respect to the local involutions and the even part the $1$-forms $\cd{\xi}$ with respect to $\sigma^{\alpha}$ is holomorphic. As in the proof of Theorem~\ref{thm:top:F-CohFT:F-TR}, it is not hard to check that the above formulae are equivalent to the coordinate-free \eqref{eq:ABCD:Fsc}. We also remark that indices appearing in different positions in the left- and right-hand sides of \eqref{eq:ABCD:res} signals the presence of standard up/down-morphisms, which are simply Kronecker deltas. The coordinate-free expressions remain true if we use an arbitrary pair of compatible up/down-morphisms instead of the standard one.

	To complete the proof, it remains to justify equations \eqref{eq:ABCD:res}. The tensors $A$, $B$ and $\disc{C}$ are computed exactly as in \cite{ABCO24}, so we simply focus on the identification of $D$ and $\conn{C}$. For $D$ we compute
	\begin{equation}
	\begin{split}
		\cd{\mathscr{P}}[\omega_{1,1}](z_0)
		&=
		\cd{\mathscr{P}}\left[
			\sum_{\alpha \in \mathfrak{a}} \Res_{z=\alpha} \
				w^{\alpha} \,
				K^{\connsymb,\alpha}(z_0|z) \,
				\conn{\omega}_{0,2}(z|\sigma^{\alpha}(z))
			\right] \\
		&=
		\sum_{\substack{ \alpha \in \mathfrak{a} \\ k \ge 0}}
		\underbrace{
			\vphantom{\bigg(}
			\left(
			- \frac{w^{\alpha}}{2}
			\Res_{z=\alpha} \
				\theta^{\alpha}(z)
				 \,
				\epsilon_{(\alpha,k)}(z) \,
				\omega_{0,2}(z|\sigma^{\alpha}(z))
			\right)
		}_{= D^{(\alpha,k)}}
		\underbrace{
			\vphantom{\bigg(}
			\cd{\mathscr{P}}\left[
				\xi^{\connsymb,(\alpha,k)}(z_0)
			\right]
		}_{= \cdxi{(\alpha,k)}(z_0)} \\
		&=
		\sum_{\substack{ \alpha \in \mathfrak{a} \\ k \ge 0}}
			D^{(\alpha,k)} \, \cdxi{(\alpha,k)}(z_0) \,.
	\end{split}
	\end{equation}
	So indeed $\indF{F}{1}{(\alpha,k)}{\emptyset} = D^{(\alpha,k)}$. As for $\conn{C}$, we insert the decomposition of the projected correlators on the $\cd{\xi}$-basis in the topological recursion formula \eqref{eq:FTR:residue}. Focusing on the term of connected type, we find: 
	\begingroup
	\allowdisplaybreaks
	\begin{align}
		\notag
		&
		(\cd{\mathscr{P}})^{\otimes(1+n)}
		\Bigg[
			\sum_{\alpha_0 \in \mathfrak{a}} \Res_{z=\alpha_0} \
				w^{\alpha_0} \,
				K^{\connsymb,\alpha_0}(z_0|z) \,
				\sum_{\substack{
					(\beta,k),(\beta',k') \\
					(\alpha_1,k_1),\dots,(\alpha_n,k_n)
				}}
				\indF{F}{g-1}{(\beta,k)}{(\beta',k'),(\alpha_1,k_1),\dots,(\alpha_n,k_n)} \\
		\notag
		&\qquad\qquad\qquad\qquad\qquad\qquad\qquad\qquad\qquad\qquad
		\times
				\xi^{\connsymb,(\beta,k)}(z)
				\xi^{\connsymb,(\beta',k')}(\sigma^{\alpha_0}(z))
				\prod_{i=1}^n \cdxi{(\alpha_i,k_i)}(z_i)
			\Biggr] \\
		\notag
		& = 
		\frac{1}{2} \sum_{\substack{(\alpha_0,k_0) \\ (\beta,k), (\beta',k')}}
		\overbrace{
			\vphantom{\Big(}
			-\frac{w^{\alpha_0}}{2}
			\Res_{z = \alpha_0} \
				\theta_{\alpha_0}(z) \,
				\epsilon_{(\alpha_0,k_0)}(z) \,
				\Big(\xi^{\connsymb,(\beta,k)}(z) \,
				\xi^{\connsymb,(\beta',k')}(\sigma^{\alpha_0}(z)) + \big(z \leftrightarrow \sigma^{\alpha_0}(z)\big)\Big)
		}^{
			= \indC{\conn{C}}{(\alpha_0,k_0),(\beta',k')}{(\beta,k)}
		} \\
		\notag
		& \quad \times 
		\overbrace{
			\vphantom{\Big(}
			\cd{\mathscr{P}}\left[ \xi^{\connsymb,(\alpha_0,k_0)}(z_0) \right]
		}^{
			= \cdxi{(\alpha_0,k_0)}(z_0)
		}
		\sum_{ (\alpha_1,k_1),\dots,(\alpha_n,k_n) }
			\indF{F}{g-1}{(\beta,k)}{(\beta',k'),(\alpha_1,k_1),\dots,(\alpha_n,k_n)}
			\prod_{i=1}^n \cdxi{(\alpha_i,k_i)}(z_i) \\
		& =
		\,\,\frac{1}{2}
		\sum_{ (\alpha_0,k_0),\dots,(\alpha_n,k_n) }
		\Biggl(
			\sum_{ (\beta,k), (\beta',k')}
				\indC{\conn{C}}{(\alpha_0,k_0),(\beta',k')}{(\beta,k)}
				\indF{F}{g-1}{(\beta,k)}{(\beta',k'),(\alpha_1,k_1),\dots,(\alpha_n,k_n)}
		\Biggr)
			\prod_{i=0}^n \cdxi{(\alpha_i,k_i)}(z_i) \,.
	\end{align}
	\endgroup
	This choice of $\conn{C}$ therefore matches the form of the connected term in the topological recursion formula \eqref{eq:F-TR:coords} for amplitudes of F-Airy structures. This concludes the proof.
\end{proof}
 
%–––––––––––––––––––––––
\subsection{F-spectral curves for semisimple F-CohFTs of the form \texorpdfstring{$\hat{L}\hat{R}\hat{T}\Omega^0$}{LRTOmega0}}
%–––––––––––––––––––––––
In Section~\ref{sec:identification} we described F-Airy structures whose amplitudes compute the intersection indices of F-CohFTs that are obtained from topological F-CohFTs by the action of translations, F-Givental and changes of bases. Under a semisimplicity assumption (already present in the original dictionary of \cite{DOSS14}) they coincide with F-Airy structures from F-spectral curves that we now explicitly describe.

Recall the setup of Section~\ref{subsec:identif}: let $\Omega = \hat{L}\hat{R}\hat{T}\Omega^0$, where $\Omega^0$ is a topological F-CohFT on $V_0$, $T \in u^2 V_0\bbraket{u}$, $R \in \mathfrak{Giv}$, $L \in \GL(V_0)$. After a choice of up/down-morphisms $(\mathscr{U},\mathscr{D})$ and upon Laplace transform, the F-CohFT amplitudes associated with $\Omega$ coincide with the amplitudes of the F-Airy structure \eqref{eq:FAiry:loop:LRT} on $\Upsilon_+$. We call ${}^{\Omega}(A,B,\conn{C},\disc{C},D)$ this F-Airy structure and ${}^{\Omega}F_{g,1+n} \in \Hom(\Sym{n}{\Upsilon_+},\Upsilon_+)$ the corresponding amplitudes.

Assuming semisimplicity, we define a local spectral curve ${}^{\Omega}(\Sigma,x,y,\conn{\omega}_{0,2},\disc{\omega}_{0,2},w)$ as follows. Decompose $V_0 \cong \bigoplus_{\alpha \in \mathfrak{a}} \mathbb{C}\mathrm{e}_{\alpha}$ in the canonical basis, i.e. $\mathrm{e}_{\alpha} \bcdot \mathrm{e}_{\beta} = \delta_{\alpha,\beta}^{\gamma} \mathrm{e}_{\gamma}$ for any $\alpha,\beta \in \mathfrak{a}$. In particular, $\mathrm{e} = \sum_{\alpha \in \mathfrak{a}} \mathrm{e}_{\alpha}$ is the unit. Define $\Sigma$ as the local curve $\bigsqcup_{\alpha \in \mathfrak{a}} \Sigma_{\alpha}$ where $\Sigma_{\alpha}$ is a formal neighbourhood of $0$ in $\mathbb{C}$. Functions/forms on $\Sigma$ can be identified with $V_0$-valued functions/forms on a formal neighbourhood of $0$ in $\mathbb{C}$ (with standard coordinate $\zeta$), and with this identification we set:
\begin{equation}\label{eq:Fsc:LRT}
\begin{split}
	x(\zeta)
	&= 
		\frac{\zeta^2}{2} \mathrm{e}
		\,, \\
	y(\zeta)
	&=
		\sum_{\alpha \in \mathfrak{a}} \left(-\zeta + \frac{\dd \tau}{\dd x}(\zeta)\right) \bcdot \mathrm{e}_{\alpha}
		\,, \\
	\conn{\omega}_{0,2}(\zeta_1|\zeta_2)
	&=
		\left( \sum_{\alpha \in \mathfrak{a}} \mathrm{e}_{\alpha} \otimes \mathrm{e}_{\alpha} \right)
		\frac{\dd \zeta_1 \, \dd \zeta_2}{(\zeta_1 - \zeta_2)^2}
		\,, \\
	\disc{\omega}_{0,2}(\zeta_1|\zeta_2)
	&=
		\left( \sum_{\alpha \in \mathfrak{a}} \mathrm{e}_{\alpha} \otimes \mathrm{e}_{\alpha} \right)\frac{\dd \zeta_1 \, \dd \zeta_2}{(\zeta_1 - \zeta_2)^2} + \dd_{\zeta_1}\dd_{\zeta_2}\Epsilon_{R}(\zeta_1,\zeta_2)
		\,, \\
	w
	&=
		\bigl( w^{\alpha} \bigr)_{\alpha \in \mathfrak{a}}
		\,.
\end{split}
\end{equation}
The scalars $w^{\alpha}$ are simply the expansion coefficients of the distinguished vector $w$ in the canonical basis.  We identify $\widehat{\Upsilon}_+ \cong \mathscr{V}_+$ in the natural way, so that the expression for $\conn{\omega}_{0,2}$ implies $\conn{\mathscr{V}}_- = \Upsilon_-$. We also take $\check{\mathscr{V}}_+ = \Upsilon_+$. We then choose up/down-morphisms $\conn{\Up} = \Up$ and $\conn{\Delta} = \Delta$, and $(\disc{\Up},\disc{\Delta})$ are deduced by compatibility. Let us call ${}^{\Omega}(\tilde{A},\tilde{B},\conn{\tilde{C}},\disc{\tilde{C}},\tilde{D})$ the F-Airy structure specified by the formulae in Proposition~\ref{prop:FAiry:Fsp}, and $\tilde{F}_{g,1+n} \in \Hom(\Sym{n}{\check{\mathscr{V}}_+},\check{\mathscr{V}})$ the corresponding amplitudes.
 
The following result shows that these two F-Airy structures agree up to a change of bases. In particular, this means that the intersection indices of $\Omega$ can be computed by F-topological recursion on spectral curves as formulated in Section~\ref{Section61}.

\begin{prop}
	The F-Airy structure ${}^{\Omega}(A,B,\conn{C},\disc{C},D)$ is obtained by applying to the F-Airy structure ${}^{\Omega}(\tilde{A},\tilde{B},\conn{\tilde{C}},\disc{\tilde{C}},\tilde{D})$ the change of bases  with source isomorphism $\lambda_{{LR},\textup{s}}$ and target isomorphism $\lambda_{{LR},\textup{t}}$ defined in \eqref{eq:L:LR}.  In particular the amplitudes for $2g - 2 + 1 + n > 0$ are related by
	\begin{equation}
		F_{g,1+n} = \lambda_{{LR},\textup{t}} \circ \tilde{F}_{g,1+n} \circ (\lambda_{{LR},\textup{s}}^{-1})^{\otimes n} \,,
	\end{equation}
	where we may consider the extension of the tensors $F_{g,1+n}$ to completed loop spaces (see Remark~\ref{rem:finite:F}).
\end{prop}

\begin{proof}
	We compare Proposition~\ref{prop:FAiry:Fsp} with the formulae for the F-Airy structure associated with $\hat{L}\hat{R}\hat{T}\Omega^0$ and $(\Up,\Delta)$ at the end of Section~\ref{subsec:identif}. Notice that the isomorphism
	\begin{equation}
		\id_{\Upsilon_-} + \dd \circ \Epsilon_{R} \colon \Upsilon_- \longrightarrow \disc{\mathscr{V}}_-
	\end{equation}
	coincides with the restriction of $\disc{\mathscr{P}}_{\loc}$ to $\Upsilon_-$. Therefore, the choice of $\conn{\Up} = \Up$ and $\conn{\Delta} = \Delta$ together with the compatibility yields
	\begin{equation}
		\disc{\Delta} = ( \id_{\Upsilon_-} + \dd \circ \Epsilon_{R}) \circ \Delta \,.
	\end{equation}
	Then ${}^{\Omega}(\tilde{A},\tilde{B},\conn{\tilde{C}},\disc{\tilde{C}},D)$ matches precisely the F-Airy structure obtained at the end of Step 3 after Theorem~\ref{thm:ident:actions}. It remains to apply Step 4 (the change of bases with specified isomorphisms) to get the F-Airy structure ${}^{\Omega}(A,B,\conn{C},\disc{C},D)$.
\end{proof}
 
We remark that for topological F-CohFTs both $\conn{\omega}_{0,2}$ and $\disc{\omega}_{0,2}$ coincide with the `standard bidifferential' in the local coordinates $\zeta$ that transform as $\zeta \mapsto -\zeta$ under the local involutions. The F-Givental action generates a more general $\disc{\omega}_{0,2}$ but does not change $\conn{\omega}_{0,2}$, which remains the `standard' one.

Besides, the F-Airy structure in the proposition above remains unchanged if we modify the F-spectral curve by adding $V_0$-valued constants to $x$, an even (with respect to the local involutions) germ of holomorphic function to $y$, and a germ of holomorphic bidifferential to $\omega_{0,2}$ that is even in at least one variable. This freedom may be exploited to investigate the existence of a global F-spectral curve whose local behaviour matches the germ near its ramification points.

The spectral curve description \eqref{eq:Fsc:LRT} is rather compact in contrast to the (equivalent) F-Airy structure description \eqref{eq:FAiry:loop:LRT}. It also handles `by itself' the infinite-dimensional questions that were more annoying to treat in the tensorial presentation. Besides, the F-spectral curve approach does not need up/down-morphisms to be formulated (it is only necessary to choose some to compare it to F-CohFTs amplitudes). These advantages can be appreciated in the extended 2-spin class example studied in Section~\ref{subsec:ext:2spin:class}.

\begin{ex}
	A simple computation shows that the following local F-spectral curve is associated with the extended $2$-spin F-CohFT:
	\begin{equation}
	\begin{split}
		x(\zeta)
		&=
			\frac{\zeta^2}{2} \mathrm{e}
			\,, \\
		y(\zeta)
		&=
			-\zeta \mathrm{e}_1 + \frac{\ln(s - \zeta)}{s} \mathrm{e}_2
			\,, \\
		\conn{\omega}_{0,2}(\zeta_1|\zeta_2)
		&=
			(\mathrm{e}_1 \otimes \mathrm{e}_1 + \mathrm{e}_2 \otimes \mathrm{e}_2)\,\frac{\dd \zeta_1 \, \dd \zeta_2}{(\zeta_1 - \zeta_2)^2}
			\,, \\
		\disc{\omega}_{0,2}(\zeta_1|\zeta_2)
		&=
			(\mathrm{e}_1 \otimes \mathrm{e}_1 + \mathrm{e}_2 \otimes \mathrm{e}_2)\,\frac{\dd \zeta_1 \, \dd \zeta_2}{(\zeta_1 - \zeta_2)^2} \\
		& \qquad
			+ \mathrm{e}_2 \otimes \mathrm{e}_1 \sum_{k_1,k_2 \geq 0} \frac{(2k_1 + 2k_2 + 1)!!}{(2k_1 - 1)!!(2k_2 - 1)!!} \zeta_1^{2k_1} (-1)^{k_2} \zeta_2^{2k_2} \dd \zeta_1 \, \dd \zeta_2
			\,, \\
		w & = (0,-s^2) \,.
	\end{split}
	\end{equation}
	Note that the double series in $\disc{\omega}_{0,2}(\zeta_1|\zeta_2)$ takes the alternative form
	\begin{equation}
		F_2\bigl[\tfrac{3}{2};1,1;\tfrac{1}{2},\tfrac{1}{2}\bigr](\zeta_1^2,-\zeta_2^2) \, \dd \zeta_1 \, \dd \zeta_2\,.
	\end{equation}
	where $F_2$ is the second Appell series.

	Moreover, we notice that by formally setting $\mathrm{e}_2 = 0$ we retrieve the Airy spectral curve, which is known to compute $\psi$-class intersection numbers. This is in line with the construction of the extended $2$-spin class, which precisely extends the Witten $2$-spin class (i.e. the fundamental class) by adding the new direction $\mathrm{e}_2$.
\end{ex}

%–––––––––––––––––––––––––––––––––––––––––––%
\section{Conclusion}
%–––––––––––––––––––––––––––––––––––––––––––%

In this article, we introduced F-topological recursion in both the Airy structure and spectral curve formalisms. We analysed the symmetries of these structures, identifying the classical change of basis, the Bogoliubov transformation, and the translation, as well as a novel transformation: the tick action. We then compared our recursion with the F-Givental action on F-CohFTs, establishing a full parallelism between the two theories.
 
This work should be viewed as an extension to the F-world of a well-known square of structures in the non-F setting.
% \begin{itemize}
%   \item[(1)] Kontsevich--Soibelman Airy structures \cite{KS18},
%   \item[(2)] Chekhov--Eynard--Orantin topological recursion \cite{EO07},
%   \item[(3)] Kontsevich--Manin CohFTs \cite{KM94},
%   \item[(4)] Dubrovin's Frobenius manifolds \cite{Dub96}.
% \end{itemize}
\begin{center}
\begin{tikzcd}[row sep=5em, column sep=6em]
	\parbox{5.5cm}{\centering (1) Kontsevich--Soibelman\\Airy structures \cite{KS18}} 
	\arrow[r, dashed, bend right=5] 
  &	\parbox{5.5cm}{\centering (2) Chekhov--Eynard--Orantin\\topological recursion \cite{EO07}} 
    \arrow[l, bend right=5, "\text{loop eqs.}"']
    \arrow[d, dashed, bend right=5, xshift=-.15cm]
    %\arrow[d, "\text{local \& global (semisimple)}" description]
    %\arrow[dl, dashed, bend right=25, "\text{not always}" description]
    \\
	\parbox{5.5cm}{\centering (4) Dubrovin's\\Frobenius manifolds \cite{Dub96}} 
	\arrow[r, <->, "\text{genus 0}", "\text{semisimple}"']
  & \parbox{5.5cm}{\centering (3) Kontsevich--Manin\\CohFTs \cite{KM94}} 
    \arrow[u, bend right=5, xshift=.15cm, "\text{local, global?}"']
    \arrow[ul, ->, "\text{Virasoro}" description] 
\end{tikzcd}
\end{center}

Each of these has its own independent development, with several known interrelations:
\begin{itemize}
  \item (2)~$\Rightarrow$~(1): Spectral curves with simple ramification yield quadratic Airy structures (i.e., Virasoro constraints) via loop equations \cite{KO10,KS18,ABCO24}. Curves with higher ramification give rise to Airy structures beyond the quadratic case \cite{BBCCN24}.
  
  \item (3)~$\Rightarrow$~(2), local: Semisimple CohFTs give rise to local spectral curves, and topological recursion computes the all-genera potential of the corresponding theory \cite{Eyn14,DOSS14}.
  
  \item The composition (3)~$\Rightarrow$~(2)~$\Rightarrow$~(1) encodes the Virasoro constraints in canonical coordinates for descendant integrals of semisimple CohFTs.
  
  \item (3)~$\Rightarrow$~(2), global: Constructing a global spectral curve involves Dubrovin's superpotential, if it exists, which provides a Riemann surface with two functions $x$ and $y$ \cite{DNOPS18,DNOPS19}. How to construct a compatible bidifferential $\omega_{0,2}$ is a problem generally unsolved except in special cases (e.g., Hurwitz--Frobenius manifolds).

  \item (3)~$\Leftrightarrow$~(4): The genus-zero part of a CohFT corresponds to a Frobenius manifold \cite{KM94}. Conversely, any semisimple Frobenius manifold gives rise to a semisimple CohFT whose genus-zero part retrieves the original Frobenius manifold, and one can retrieve in this way semisimple CohFTs in all genera \cite{Tel12}.
  
  \item There is no general (1)~$\Rightarrow$~(2): not all Airy structures arise from spectral curves (e.g., finite-dimensional Airy structures based on simple Lie algebras \cite{HR21}). Nor is there a general (2)~$\Rightarrow$~(3): spectral curves associated with semisimple CohFTs are those for which the Laplace transform of $\omega_{0,2}$ factors through an $R$-matrix \cite{DOSS14}. Moreover, there is no general procedure to construct CohFTs from spectral curves with higher ramification.
\end{itemize}
 
In the F-setting, prior work has involved:  
\begin{itemize}
  \item[(3)] F-CohFTs \cite{BR21},  
  \item[(4)] flat F-manifolds \cite{HM99},  
\end{itemize} 
with the known implication (3)~$\Rightarrow$~(4). There is a known construction (4)~$\Rightarrow$~(3) of F-CohFTs from semi-simple F-CohFTs using the F-analogue of Givental symmetries \cite{ABLR23}, but the analogue of Teleman's reconstruction fails in general. Our paper introduces the analogues of (1) and (2), namely:
\begin{itemize}
  \item[(1)] F-Airy structures,  
  \item[(2)] F-topological recursion,  
\end{itemize}
together with the implication (3)~$\Rightarrow$~(2)~$\Rightarrow$~(1) for those F-CohFTs that lie in the F-Givental orbit of semisimple ones. It remains an open question to find an analogue of Dubrovin's superpotential in the F-setting, and the construction of a global F-spectral curve from a given F-CohFT is an open problem that remains only partially understood even in the non-F setting. These are therefore interesting but out-of-scope questions for this article, which does not address F-manifolds directly, though they serve as background motivation.

%–––––––––––––––––––––––––––––––––––––––––––%
\subsubsection*{Acknowledgements.}
%–––––––––––––––––––––––––––––––––––––––––––%
We thank Paolo Rossi for discussions on F-CohFTs: his talk at the MPIM-HU AG\&Phy seminar in June 2022 sparked this project, and he pointed a (now corrected) mistake in a first version. G.B. thanks the IHÉS for hospitality and excellent working conditions allowing the completion of this project, and Bruno Vallette for discussions on symmetries of algebraic structures. A.G. was supported by an ETH Fellowship (22-2~FEL-003) and a Hermann-Weyl-Instructorship from the Forschungsinstitut für Mathematik at ETH Zürich. G.U. was funded by the Deutsche Forschungsgemeinschaft (DFG, German Research Foundation) under Germany's Excellence Strategy – The Berlin Mathematics Research Center MATH+ (EXC-2046/1, project ID: 390685689).

% %–––––––––––––––––––––––––––––––––––––––––––%
% \subsubsection*{Data availability.}
% %–––––––––––––––––––––––––––––––––––––––––––%
% No data sets were generated during this study.

% %–––––––––––––––––––––––––––––––––––––––––––%
% \subsubsection*{Declarations.}
% %–––––––––––––––––––––––––––––––––––––––––––%

% %–––––––––––––––––––––––––––––––––––––––––––%
% \subsubsection*{Conflict of interest.}
% %–––––––––––––––––––––––––––––––––––––––––––%
% The authors declare no conflict of interest.

%–––––––––––––––––––––––––––––––––––––––––––%
%\bibliographystyle{alpha}
%\bibliography{BibliographyFAiry}
\printbibliography 
%–––––––––––––––––––––––––––––––––––––––––––%

\end{document}